%% file: Dimensionality.tex
\documentclass[11pt]{article}
\usepackage[letterpaper,top=1in, bottom=1in, left=1in, right=1in]{geometry}
\usepackage[compact]{titlesec}
\usepackage[bottom]{footmisc}
\usepackage{hyperref}
\hypersetup{
    colorlinks=true,
    linkcolor=blue,
   citecolor=blue
}
\usepackage{titling}

\usepackage{enumitem}
\newcommand{\enumSpaceTS}{-5.5pt}
\newcommand{\enumSpacePTS}{0pt}
\newcommand{\enumSpacePS}{0pt}
\newcommand{\enumSpaceIS}{2pt}

\usepackage{graphicx}
\graphicspath{{./figures/}}
\usepackage{caption}
\usepackage{subcaption}
\usepackage{color}
\newcommand{\plotHeight}{1.57in}

\usepackage[linesnumbered,ruled]{algorithm2e}
\makeatletter
\renewcommand{\@algocf@capt@plain}{above}
\makeatother

\usepackage[parfill]{parskip} 

\usepackage{amssymb,amsmath,amsthm,thmtools}
\newtheoremstyle{mystyle}{1.5ex}{}{\itshape}{}{\bfseries}{.}{1ex}{}
\theoremstyle{mystyle} \newtheorem{clm}{Claim}
\theoremstyle{mystyle} 
\theoremstyle{mystyle} \newtheorem{lem}{Lemma}
\theoremstyle{mystyle} 
\theoremstyle{mystyle} \newtheorem{thm}{Theorem}
\theoremstyle{mystyle} \newtheorem{ass}{Assumption}
\theoremstyle{mystyle} \newtheorem{defn}{Definition}
\declaretheoremstyle[%
  spaceabove=-6pt,%
  spacebelow=0ex,%
  headfont=\normalfont\itshape,%
  postheadspace=1em,%
  qed=\qedsymbol%
]{prfstyle} 
\declaretheorem[name={Proof},style=prfstyle,unnumbered,
]{prf}
\usepackage{mathtools}
\mathtoolsset{showonlyrefs}
\DeclarePairedDelimiter{\ceil}{\lceil}{\rceil}

\newcommand{\E}{\mathbb{E}} 
\renewcommand{\P}{\mathbb{P}} 
\newcommand{\R}{\mathbb{R}} 
\newcommand{\N}{\mathbb{N}} 
\newcommand{\Z}{\mathbb{Z}} 

\newcommand{\trans}{\mathsf{T}}
\renewcommand{\i}{\mathbf{i}}

\title{Personalized PageRank dimensionality and algorithmic implications}
\author{Daniel Vial, Vijay Subramanian \\ University of Michigan \\ \{dvial,vgsubram\}@umich.edu }
\date{\today}

\begin{document}

\begin{titlingpage}

\maketitle

\begin{abstract}
Many systems, including the Internet, social networks, and the power grid, can be represented as graphs. When analyzing graphs, it is often useful to compute scores describing the relative importance or distance between nodes. One example is Personalized PageRank (PPR), which assigns to each node $v$ a vector whose $i$-th entry describes the importance of the $i$-th node from the perspective of $v$. PPR has proven useful in many applications, such as recommending who users should follow on social networks (if this $i$-th entry is large, $v$ may be interested in following the $i$-th user). Unfortunately, computing $n$ such PPR vectors (where $n$ is the number of nodes) is infeasible for many graphs of interest.

In this work, we argue that the situation is not so dire. Our main result shows that the dimensionality of the set of PPR vectors scales sublinearly in $n$ with high probability, for a certain class of random graphs and for a notion of dimensionality similar to rank. Put differently, we argue that the effective dimension of this set is much less than $n$, despite the fact that the matrix containing these vectors has rank $n$. Furthermore, we show this dimensionality measure relates closely to the complexity of a PPR estimation scheme that was proposed (but not analyzed) by Jeh and Widom. This allows us to argue that accurately estimating all $n$ PPR vectors amounts to computing a vanishing fraction of the $n^2$ vector elements (when the technical assumptions of our main result are satisfied). Finally, we demonstrate empirically that similar conclusions hold when considering real-world networks, despite the assumptions of our theory not holding.
\end{abstract}

\end{titlingpage}

\input{intro}

\input{prelim}

\input{related}

\input{mainResult}

\input{algorithms}

\input{conclusions}

\newpage

\appendix

\textbf{Note on the organization of appendices:} Appendix \ref{secMainProof} outlines the key ideas and intuition behind the proof of Lemma \ref{lemMainTailBound}, which contains the bulk of our technical analysis and itself requires five lemmas. The proofs of these lemmas are found in subsections of Appendix \ref{appMainProofs}, in the order their statements appear in Appendix \ref{secMainProof}. Shorter proofs (those of Claim \ref{clmChoiceOfAlpha} and Theorem \ref{thmSublinear}) are found in Appendix \ref{appOtherProofs}. Finally, Appendix \ref{appExpDetails} contains details on the experiments of Section \ref{secAlgorithms}.

\input{mainResultProof}

\input{proofLinearCombo}

\input{proofMstepNeighborhood}

\input{proofCoupling}

\input{proofCouplingBroken}

\input{proofTailBound}

\input{simultaneousConstruction}

\input{shortProofs}

\input{experimentDetails}

\bibliographystyle{acm}
\bibliography{references}

\end{document}

%% file: intro.tex
\section{Introduction} \label{secIntro}

Many natural and man-made systems can be represented as graphs, sets of objects (called nodes) and pairwise relations between these objects (called edges). These include the brain, which contains neurons (nodes) that exchange signals through chemical pathways (edges), the Internet, which contains websites (nodes) that are connected via hyperlinks (edges), etc. To study graphs, researchers in diverse domains have used Personalized PageRank (PPR) \cite{gleich2015pagerank}. Informally, PPR assigns to each node $v$ a vector $\pi_v$, where $\pi_v(w)$ describes the importance of $w$ from the perspective of $v$. PPR has proven useful in many practical and graph theoretic applications. Examples include recommending who a user should follow on Twitter \cite{gupta2013wtf} (user $v$ may wish to follow user $w$ if $\pi_v(w)$ is large), and partitioning graphs locally around a seed node \cite{andersen2006local} (the set of nodes $w$ with large $\pi_v(w)$ can be viewed as a community surrounding $v$). Unfortunately, computing all $n$ PPR vectors (where $n$ is the number of nodes) is infeasible for the massive graphs encountered in practice.

In this work, we argue that all $n$ PPR vectors can be accurately estimated by computing only a vanishing fraction of the $n^2$ vector elements, with high probability and for a certain class of random graphs. This arises as a consequence of our main (structural) result, which shows that the dimensionality of the set of PPR vectors scales sublinearly in $n$ with high probability, for the same class of random graphs and for a notion of dimensionality somewhat similar to matrix rank. We note that the estimation scheme considered was first proposed by Jeh and Widom in \cite{jeh2003scaling} without a formal analysis, so another contribution of our paper is to address this lacuna.

The paper is organized as follows. We begin in Section \ref{secPrelim} with preliminary definitions. Section \ref{secRelated} discusses related work. In Section \ref{secMain}, we state our main result. We then discuss algorithmic implications and present empirical results in Section \ref{secAlgorithms}. Finally, we close in Section \ref{secConclusions}.

%% file: prelim.tex
\section{Preliminaries} \label{secPrelim}

We begin by defining the main ingredients of the paper. Most notation is standard or defined as needed, but we note the following is often used: for $x \in \R^n$ and $J \subset \{1,2,\ldots,n\}$, $x(J) \triangleq \sum_{j \in J} x(j)$, $e_J \in \{ 0,1 \}^n$ satisfies $e_J(j) = 1(j \in J)$ (where $1(\cdot)$ is the indicator function), and $e_j \triangleq e_{ \{ j \} }$. 

\subsection{Directed configuration model (DCM)} \label{secDcmDescription}

We consider a random graph model called the directed configuration model (DCM). For the DCM, we are given realizations of random sequences $\mathbf{N}_n = \{ N_v \}_{v \in V_n}$ and $\mathbf{D}_n = \{ D_v \}_{v \in V_n}$ satisfying $N_v, D_v \in \N\ \forall\ v \in V_n$ and $\sum_{v \in V_n} N_v = \sum_{v \in V_n} D_v \triangleq L_n$ (we assume $V_n = \{1,2,\ldots,n\}$ for simplicity).\footnote{More specifically, we would like $\mathbf{N}_n \sim f_{\textrm{in}}$ \emph{i.i.d.}\ and $\mathbf{D}_n \sim f_{\textrm{out}}$ \emph{i.i.d.}\ for given distributions $f_{\textrm{in}}$ and $f_{\textrm{out}}$, but this does not guarantee $\sum_{v \in V_n} N_v = \sum_{v \in V_n} D_v$. For this reason, the authors of \cite{chen2013directed} provide an method to generate these sequences such that $\sum_{v \in V_n} N_v = \sum_{v \in V_n} D_v$, and $\forall\ i,j \in \N$, $(N_1,\ldots,N_i,D_1,\ldots,D_j) \Rightarrow (\bar{N}_1,\ldots,\bar{N}_i,\bar{D}_1,\ldots,\bar{D}_j)$, where $\{ \bar{N}_l \}_{l=1}^i \sim f_{\textrm{in}}$ \emph{i.i.d.},  $\{ \bar{D}_l \}_{l=1}^j \sim f_{\textrm{out}}$ \emph{i.i.d.}, and $\Rightarrow$ denotes convergence in distribution.} Our goal is to construct a directed graph $G_n = (V_n, E_n)$, such that $v \in V_n$ has in- and out-degree $N_v$ and $D_v$, respectively. For this, we first assign $N_v$ incoming half-edges and $D_v$ outgoing half-edges to each $v \in V_n$; we call these half-edges \textit{instubs} and \textit{outstubs}, respectively. We then randomly pair half-edges in a breadth-first search fashion that proceeds as follows:
\begin{enumerate}[topsep=\enumSpaceTS,partopsep=\enumSpacePTS,parsep=\enumSpacePS,itemsep=\enumSpaceIS]
\item Choose $s \sim V_n$ uniformly. For each of the $D_s$ outstubs assigned to $s$, sample an instub uniformly from the set of all instubs (resampling if the sampled instub has already been paired), and pair the outstub and instub to form a directed edge out of $s$.
\item Let $A_1 = \{ v \in V_n \setminus \{ s \} : \textrm{an outstub of $s$ was paired with an instub of $v$} \}$. For each $v \in A_1$, pair the $D_v$ outstubs assigned to $v$ using the method that $s$'s outstubs were paired in Step 1.
\item Continue iteratively until all half-edges have been paired. Namely, during the $(m+1)$-th iteration we pair the oustubs of all $v \in A_{m}$, where %
$A_{m}$ are nodes at distance $m$ from $s$ (those $v$ for which a path of length $m$ from $s$ to $v$ exists, but no shorter path from $s$ to $v$ exists).
\end{enumerate}

We define this procedure formally in Appendix \ref{secCoupling}. For now, the important points to remember are that the initial node $s$ is chosen uniformly at random from $V_n$, and that, at the end of the $m$-th iteration, the $m$-step neighborhood out of $s$ has been constructed.  We emphasize the resulting graph will be a multi-graph in general, i.e.\ it will contain self-loops (edges $v \rightarrow v$ for $v \in V_n$) and multi-edges (more than one edge from $v \in V_n$ to $w \in V_n$). In \cite{chen2013directed}, the authors provide conditions under which a simple graph results with positive probability as $n \rightarrow \infty$, but these are stronger than the conditions we require to prove our main result. Hence, we assume $G_n$ is a multi-graph.

\subsection{Personalized PageRank (PPR)} \label{secPprDefinition}

To define PPR, we require some notation. First, let $M$ denote the adjacency matrix for some realization of the DCM, i.e.\ $M(i,j) \in \{ 0,1,\ldots,D_i \}$ is the number of directed edges from $i$ to $j$ ($i,j \in V_n$). Next, let $P$ be the row stochastic matrix with $P(i,j) = M(i,j) / D_i\ \forall\ i,j \in V_n$. Finally, let $\alpha_n \in (0,1)$, and let $1_n$ denote the length-$n$ vector of ones. We then have the following.
\begin{defn} \label{defnPpr}
For $v \in V_n$, the PPR row vector $\pi_v$ is the stationary distribution of the Markov chain with transition matrix $P_v = (1-\alpha_n) P + \alpha_n 1_n e_v^{\trans}$.
\end{defn}
Note that $M$, $P$, $P_v$, and $\pi_v$ all depend on $n$. However, to avoid cumbersome notation, we do not explicitly denote this, and the dependence on $n$ will be clear from context.

The Markov chain described in Definition \ref{defnPpr} has the following dynamics: follow a uniform random walk with probability $(1-\alpha_n)$, and jump to $v$ with probability $\alpha_n$. This motivates an interpretation of PPR as a \textit{centrality measure} of the nodes $V_n$ from the \textit{perspective} of $v$. To see this, let $\{X_i \}_{i=0}^{\infty}$ denote the Markov chain with transition matrix $P$. Then one can show (see Appendix \ref{appRenewalReward})
\begin{equation} 
\pi_v(u) = \alpha_n \E_{G_n} \left[ \textstyle \sum_{i=0}^{L-1} 1 ( X_i = u )  \middle| X_0 = v \right]\ \forall\ u \in V_n  ,
\end{equation}
where $L \sim \textrm{geometric}(\alpha_n)$, and where $\E_{G_n}[\cdot]$ denotes expectation with some realization of the DCM $G_n$ held fixed. Hence, $\pi_v(u)$ is large when $u$ is frequently visited (a notion of centrality) on $\textrm{geometric}(\alpha_n)$-length walks beginning at $v$ (a notion of $v$'s perspective).

We note the typical definition of PPR assumes $\alpha_n$ is constant; in contrast, we take $\alpha_n = O ( \frac{1}{\log n} )$. We argue in Section \ref{secChoiceOfAlpha} that this is appropriate when considering the asymptotic behavior of PPR on the DCM. Specifically, we argue that the size of the set of nodes that are important to $v$ grows with the graph, but grows slowly enough that a notion of $v$'s perspective remains, when $\alpha_n = O ( \frac{1}{\log n} )$. (In contrast, this set has constant size when $\alpha_n$ is constant.) Additionally, the spectral gap of $P_v$ is lower bounded by $\alpha_n$, so $\alpha_n \rightarrow 0$ as $n \rightarrow \infty$ results in this lower bound vanishing asymptotically. We note a line of work by Boldi \textit{et al.} \cite{boldi2005pagerank,boldi2009pagerank} analyzed the limit of PPR as $\alpha_n \rightarrow 0$ for a fixed graph $G_n$; in contrast, we fix a value of $\alpha_n$ for each $G_n$.

Finally, we emphasize the distinction between PPR and the more commonly known notion of PageRank, which we refer to as global PageRank. In short, global PageRank is the average of all PPR vectors, i.e.\ $\frac{1}{n} \sum_{v \in V_n} \pi_v$. Hence, global PageRank is a centrality measure from the perspective of a uniform node. More generally, given a distribution $\sigma_n$ on $V_n$, the PPR corresponding to $\sigma_n$ is $\mathrm{PPR}(\sigma_n) \triangleq \E_{\sigma_n} [ \pi_{\mathcal{V}} ]$, where the random variable $\mathcal{V}$ has $\sigma_n$ as its distribution.

\subsection{PPR dimensionality and algorithmic implications} \label{secMainMeas}

Our main goal is to investigate the dimensionality of the set of PPR vectors, $\{ \pi_v  \}_{v \in V_n}$. A standard measure of the dimension of such a set is the size of its largest linearly independent subset. However, $\forall\ n \in \N$, $\{ \pi_v  \}_{v \in V_n}$ is a linearly independent set itself \footnote{To see why, first suppose $( I - (1-\alpha_n) P )$ is not invertible. Then $( I - (1-\alpha_n) P ) \xi = 0$ for some $\xi \neq 0$, so $\xi / (1-\alpha_n) = P \xi$. But, by the Perron-Frobenius theorem, $P$ cannot have eigenvalue $1/(1-\alpha_n) > 1$, since it is row stochastic. Hence, $( I - (1-\alpha_n) P )$ is invertible, so by \eqref{eqPprPowerIter}, the matrix with rows $\{ \pi_v  \}_{v \in V_n}$ is invertible as well.}, so we will instead consider a different notion of dimensionality. This notion is motivated by the following observation: given vectors $\{ x_i \}$, the size of a linearly independent subset of $\{ x_i \}$ can be bounded by $| X' \cup X'' |$, where $X' \subset \{ x_i \}$ and $X'' = \{ x_i \notin X' : x_i \textrm{ is not a linear combination of } X' \}$. We will relax this slightly, by only including in $X''$ those $x_i \notin X'$ that are not ``close'' to a linear combination of $X'$. In particular, given $\epsilon > 0$, our notion of dimensionality is $\Delta_n(\epsilon)$, where
\begin{gather}
\Delta_n(\epsilon) = \min_{K_n \subset V_n} \left( | K_n | + | \{ v \in V_n \setminus K_n : B_v(\epsilon) \textrm{ holds} \} | \right) , \label{eqDimMeasure} \\
B_v (\epsilon) = \left\{ \min_{ \{ \beta_v(k)  \}_{k \in K_n} \subset \R }  \left\| \pi_v - \left( \alpha_n e_v^{\trans} + \textstyle \sum_{k  \in K_n} \beta_v(k) \pi_k \right) \right\|_1 \geq \epsilon \right\} \label{eqApproxFails} .
\end{gather}
Note we can also interpret \eqref{eqDimMeasure} algorithmically: if $\{ \pi_k \}_{k \in K_n}$ is known, $\pi_v$ can be accurately estimated by computing $\{ \beta_v(k) \}_{k \in K_n}$, when $v \in V_n \setminus K_n$ and $B_v(\epsilon)$ fails. Hence, \eqref{eqDimMeasure} is the number of vectors that must be computed to ensure all vectors are accurately estimated (see Section \ref{secAlgorithms}). We note $\alpha_n e_v^{\trans}$ is included in \eqref{eqApproxFails} because it is a known component of $\pi_v$; indeed, by Definition \ref{defnPpr},
\begin{equation}\label{eqPprPowerIter}
\pi_v = (1-\alpha_n) \pi_v P + \alpha_n e_v^{\trans} \Rightarrow \pi_v = \alpha_n e_v^{\trans} ( I - (1-\alpha_n) P )^{-1} = \alpha_n e_v^{\trans} + \alpha_n e_v^{\trans} \textstyle \sum_{i=1}^{\infty} (1-\alpha_n)^i P^i .
\end{equation}

For ease of analysis, we will upper bound $\Delta_n(\epsilon)$ by choosing $K_n$ solely based on the degree sequence. Specifically, let $\psi_n : \N^n \times \N^n \rightarrow \{0,1\}^n$, define $\mathbf{U}_n = \{ U_v : v \in V_n \} = \psi_n( \mathbf{N}_n, \mathbf{D}_n )$, and let $K_n = \{ v \in V_n : U_v = 0 \}$. For $\epsilon > 0$, we then define
\begin{equation}\label{eqDimMeasureRelax}
\Delta_{\psi_n}(\epsilon) = | K_n | + | \{ v \in V_n \setminus K_n : B_v(\epsilon) \textrm{ holds} \} | ,
\end{equation}
where the subscript $\psi_n$ indicates that the right side depends on $( \mathbf{N}_n, \mathbf{D}_n )$ through $\psi_n$. Our main result, Theorem \ref{thmSublinear}, shows that $\Delta_{\psi_n}(\epsilon)$ scales sublinearly in $n$ with high probability, under certain assumptions on the degree sequence and for a particular choice of $\psi_n$. In other words, though $\{ \pi_v  \}_{v \in V_n}$ is a linearly independent set (for every finite $n$), our notion of dimensionality suggests the effective dimension is (asymptotically) much smaller. 

We note that, in addition to bounding $\Delta_n(\epsilon)$ by $\Delta_{\psi_n}(\epsilon)$, we will  later bound $| \{ v \in V_n \setminus K_n : B_v(\epsilon) \textrm{ holds} \} |$ by choosing a specific $\{ \beta_v(k) \}_{k \in K_n}$, which is not necessarily the solution of the optimization problem in \eqref{eqApproxFails}. Hence, the exact solution of \eqref{eqDimMeasure} remains an open question. Furthermore, in light of the preceding algorithmic interpretation of \eqref{eqDimMeasure}, another open problem is  to solve \eqref{eqDimMeasure} while ensuring $\{ \beta_v(k) \}_{k \in K_n}$ can be efficiently computed when $v \in V_n \setminus K_n$ and $B_v(\epsilon)$ fails.

Finally, recall $(\mathbf{N}_n, \mathbf{D}_n)$ is a random sequence; hence, with $\psi_n$ fixed, $(\mathbf{N}_n, \mathbf{D}_n, \mathbf{U}_n)$ is a random sequence as well. Towards proving our main result, intermediate results will be established with $(\mathbf{N}_n, \mathbf{D}_n, \mathbf{U}_n)$ held fixed, after which conditional expectation with respect to $(\mathbf{N}_n, \mathbf{D}_n, \mathbf{U}_n)$ will be taken. This motivates the following definitions: $\E_n [ \cdot ] = \E [ \cdot | \mathbf{N}_n, \mathbf{D}_n, \mathbf{U}_n], \P_n [ \cdot ] = \P [ \cdot | \mathbf{N}_n, \mathbf{D}_n, \mathbf{U}_n]$.

%% file: related.tex
\section{Related work} \label{secRelated}

Before proceeding to our results, we comment on relationships to prior work. We focus on \cite{jeh2003scaling} and \cite{chen2017generalized}, the papers most closely related to our own.

In \cite{jeh2003scaling}, Jeh and Widom propose a scheme for estimating all PPR vectors, $\{ \pi_v \}_{v \in V_n}$. The scheme relies crucially on the Hubs Theorem in \cite{jeh2003scaling}, which states that the PPR vector $\pi_v, v \in V_n \setminus K_n$, can be written as a linear combination of $\{ \pi_k \}_{k \in K_n}$ and another vector. The Hubs Theorem is central to our results as well; an alternative formulation appears as Lemma \ref{lemLinearCombo} here. We discuss the algorithm of Jeh and Widom in more detail in Section \ref{secAlgorithms}.

Unfortunately, the authors of \cite{jeh2003scaling} present no analysis of their scheme. Hence, it is unclear how $K_n$ should be chosen and how large it must be to guarantee accurate estimation. Our work addresses this shortcoming. Specifically, as discussed briefly in the introduction and in more detail in Section \ref{secAlgorithms}, our dimensionality measure \eqref{eqDimMeasureRelax} relates to the complexity of this scheme.

In \cite{chen2017generalized}, Chen, Litvak, and Olvera-Cravioto consider the limiting value of $\textrm{PPR}(\sigma_n)$ as $\sigma_n$ weakly converges to probability distribution $\sigma$. Specifically, they show that the PPR value of a uniformly chosen node is given by the solution of a recursive distributional equation (RDE) \cite{aldous2005survey}. They also show (roughly) that PPR values follow a power law when in-degrees follow a power law, establishing the so-called ``power law hypothesis.'' Similar results were later established for a family of inhomogeneous directed graphs in \cite{lee2017pagerank}. On the other hand, \cite{chen2017generalized} was preceded by \cite{chen2014pagerank}, where the power law hypothesis was established for global PageRank; further back, the hypothesis was studied under more restrictive assumptions in \cite{litvak2007degree,volkovich2010asymptotic,volkovich2007determining}.

While \cite{chen2014pagerank,chen2017generalized,lee2017pagerank,litvak2007degree,volkovich2010asymptotic,volkovich2007determining} share a goal of understanding the power law behavior of PPR on random graphs, our goal is to instead understand structural properties of the PPR vectors collectively, with the focus of this paper being dimensionality. Since dimensionality carries with it algorithmic implications, our work is perhaps more useful from a practical perspective when compared to this body of work. However, the analytical approaches of these works will be extremely useful to us. Specifically, the proof of our main result follows an approach similar to \cite{chen2017generalized}, and we use a modified version of Lemma 5.4 from \cite{chen2017generalized}, which appears as Lemma \ref{lemCouplingFailure} here.

In short, our work can be seen as an attempt to combine the strengths of \cite{jeh2003scaling}, which is entirely algorithmic, and \cite{chen2017generalized}, which is entirely analytical. Specifically, we leverage the analytical approach from \cite{chen2017generalized} to obtain guarantees on the algorithm from \cite{jeh2003scaling}.

More broadly, references for PageRank and PPR include \cite{page1999pagerank}, in which PageRank and PPR were first proposed, and \cite{haveliwala2002topic}, an early study of PPR (there called ``topic-sensitive'' PageRank). Beyond \cite{jeh2003scaling}, many other works have proposed efficient computation and estimation algorithms for PPR; a small sample includes those using linear algebraic techniques \cite{shin2015bear,tong2008random}, those using dynamic programming \cite{andersen2008local,andersen2006local}, and those using randomized schemes \cite{avrachenkov2007monte,lofgren2016personalized}. In addition to the body of work on the power law hypothesis, analysis of PPR on random graphs includes \cite{avrachenkov2015pagerank}. Here it is shown that, for undirected random graphs with a certain expansion property, $\textrm{PPR}(\sigma_n)$ can be well approximated (in the total variation norm) as a convex combination of $\sigma_n$ and the degree distribution.

The DCM was proposed and analyzed in \cite{chen2013directed} as an extension of the (undirected) configuration model, the development of which began in \cite{bender1978asymptotic,bollobas1980probabilistic,wormald1980some}. The configuration model (and variants) have been studied in detail; for example, \cite{van2005distances} considers graph diameter in this model, while \cite{molloy1995critical} studies the emergence of a giant component.

%% file: mainResult.tex
\section{Dimensionality analysis} \label{secMain}

In this section, we present our dimensionality analysis. We begin by defining our assumptions and proposing a specific choice of $\alpha_n$. We then state the result and comment on our assumptions.

\subsection{Assumptions on degree sequence}

To prove our main result, we require Assumption \ref{assDegSeq}, which states that certain empirical moments of the sequence $(\mathbf{N}_n, \mathbf{D}_n, \mathbf{U}_n)$ exist with high probability, and furthermore, converge to limits at a uniform rate. Since we follow the analytical approach of \cite{chen2017generalized}, this assumption is similar to the main assumption in that work. We offer more specific comments shortly.

\begin{ass} \label{assDegSeq}
We have $\P [ \Omega_n^C ] = O(n^{-\delta})$ for some $\delta \in (0,1)$, where $\Omega_n = \cap_{i=1}^6 \Omega_{n,i}$ and for some constants $\gamma, p \in (0,1)$ and $\eta_i, \zeta^*, \lambda^* \in (0, \infty)$,
\begin{equation}
\begin{split}
\Omega_{n,1} & = \left\{ \left| \frac{\sum_{h=1}^n N_h}{n} - \eta_1  \right| \leq n^{-\gamma}  \right\} , \\
\Omega_{n,2} & = \left\{ \left| \frac{\sum_{h=1}^n N_h D_h}{n} - \eta_2  \right| \leq n^{-\gamma}  \right\} , \\
\Omega_{n,3} & = \left\{ \left| \frac{\sum_{h=1}^n U_h N_h^2}{n} - \eta_3 \right| \leq n^{-\gamma}  \right\}  , \\
\end{split}
\quad\quad
\begin{split}
\Omega_{n,4} & = \left\{ \left| \frac{\sum_{h=1}^n U_h D_h}{\sum_{h=1}^n U_h} - \zeta^*  \right| \leq n^{-\gamma}  \right\} , \\
\Omega_{n,5} & = \left\{ \left| \frac{\sum_{h=1}^n U_h N_h}{\sum_{h=1}^n U_h} - \lambda^*  \right| \leq n^{-\gamma}  \right\} , \\
\Omega_{n,6} & = \left\{ \left| \frac{\sum_{h=1}^n U_h N_h}{\sum_{h=1}^n N_h} - p  \right| \leq n^{-\gamma}  \right\} .
\end{split}
\end{equation}
Furthermore, we have $\zeta \triangleq \eta_2 / \eta_1 > 1$, and we define $\lambda = \eta_3 / \eta_1$.
\end{ass}

The constants $\zeta$ and $p$ will appear in our main result, and both have simple interpretations: letting $v_n$ satisfy $\P [ v_n = v ] \propto N_v\ \forall\ v \in V_n, n \in \N$, it is straightforward to show $\lim_{n \rightarrow \infty} \E [ D_{v_n} | \Omega_n ] = \zeta$ and $\lim_{n \rightarrow \infty} \E [ U_{v_n} | \Omega_n ] = p$, i.e.\ $\zeta$ and $p$ give the limiting expected out-degree and the limiting probability of belonging to $V_n \setminus K_n$, respectively. (The other constants in Algorithm \ref{assDegSeq} will not appear in our main result, but they have similar interpretations.) 
We also remark that $\zeta > 1$ is not necessary to establish our results but, given this interpretation, is the more interesting case.

\subsection{Choice of $\alpha_n$} \label{secChoiceOfAlpha}

As mentioned in Section \ref{secPprDefinition}, we take $\alpha_n = O ( \frac{1}{\log n} )$ in this work. Having defined Assumption \ref{assDegSeq}, we choose a specific value of $\alpha_n$. For this, we first present the following claim.
\begin{clm} \label{clmChoiceOfAlpha}
Let $\tau \in (0,1)$ be a constant, and let $s \sim V_n \setminus K_n$ uniformly. For $l \in \N$, let $V_{n,s}(l)$ denote the $l$-step neighborhood out of $s$, i.e.\ $V_{n,s}(l) = \{ v \in V_n : \textrm{$\exists$ a path of at most $l$ steps from $s$ to $v$} \}$. If $\alpha_n = \rho \log(1/\tau) \log (\zeta) / \log (n)  = O \left(1 / \log n \right)$ for some $\rho > 1$, let $l = \ceil*{ \log(1/\tau) / \alpha_n }$. Then
\begin{equation}
\liminf_{n \rightarrow \infty} \pi_s \left( V_{n,s} \left( l \right) \right) \geq 1 - \tau\ a.s., \quad \E \left[ \left| V_{n,s} \left( l \right) \right| \middle| \Omega_n \right] = O \left( n^{1/\rho} \right) .
\end{equation}
If instead $\alpha_n = \alpha$ is a constant, let $l = \ceil*{  \log(\tau) / \log(1-\alpha)  }$. Then
\begin{equation}
\liminf_{n \rightarrow \infty} \pi_s \left( V_{n,s} \left( l \right) \right) \geq 1 - \tau\ a.s., \quad \E \left[ \left| V_{n,s} \left( l \right) \right| \middle| \Omega_n \right] = O ( 1 ) .
\end{equation}
\end{clm}
\begin{prf}
See Appendix \ref{appProofChoiceOfAlpha}.
\end{prf}
Loosely speaking, Claim \ref{clmChoiceOfAlpha} states that, for both choices of $\alpha_n$, all but $\tau$ of $s$'s PPR concentrates on a small neighborhood surrounding $s$, for any $\tau > 0$. The difference is the size of this neighborhood: when $\alpha_n = O ( \frac{1}{\log n} )$, the neighborhood grows with the graph; when $\alpha_n$ is constant, the neighborhood has constant size. From the PPR interpretation of Section \ref{secPprDefinition}, this suggests that the number of nodes that are important to $s$ grows in the former case but remains fixed in the latter case. We believe the former case is more appropriate. Additionally, the growth of the important set of nodes remains sublinear in $n$ in the former case; intuitively, this says that a vanishing fraction of all nodes are important to $s$, i.e.\ a notion of $s$'s perspective remains. Finally, Claim \ref{clmChoiceOfAlpha} suggests that $\Delta_n(\epsilon)$ is necessarily linear when $\alpha_n$ is constant: since PPR vectors are supported on constant size sets in this case, we expect $K_n$ must be linear to cover a linear number of these sets.

\subsection{Main result} \label{secMainResultPresent}

We now turn to our main result, which relies on the following key lemma.
\begin{lem} \label{lemMainTailBound}
Given Assumption \ref{assDegSeq}, we have for $s \sim V_n \setminus K_n$ uniformly and for any $\epsilon > 0$,
\begin{align}
\P \left[ B_s(\epsilon) \middle| U_s = 1 \right] = O \left( n^{ -  c(\epsilon) } \right) , \quad c(\epsilon) \triangleq  \min \left\{ \delta ,  \tfrac{\log (1/p)}{ 2 \log ( \zeta / p ) } , \tfrac{((1-p) \epsilon )^2}{  2 \rho \log(1/\tau) \log \zeta } \right\} > 0 ,
\end{align}
where $\delta,p,\zeta$ are defined in Assumption \ref{assDegSeq} and $\rho, \tau$ are defined in Claim \ref{clmChoiceOfAlpha}.
\end{lem}
The proof of Lemma \ref{lemMainTailBound} is lengthy; we outline it in Appendix \ref{secMainProof} and provide the details in Appendix \ref{appMainProofs}. At a high level, our approach is similar to \cite{chen2017generalized} and proceeds as follows:
\begin{enumerate}[topsep=\enumSpaceTS,partopsep=\enumSpacePTS,parsep=\enumSpacePS,itemsep=\enumSpaceIS]
\item Show that, for a certain choice of $\{ \beta_s(k) \}_{k \in K_n}$, the error term $\| \pi_s - ( \alpha_n e_s^{\trans} + \sum_{k \in K_n} \beta_s(k) \pi_k ) \|_1$ in $B_s(\epsilon)$ can be bounded by only examining the $m$-step neighborhood out of $s$. 
\item Argue that, conditioned on certain events not occurring during the first $m$ steps of the graph construction, this bound follows the same distribution as a quantity defined on a tree.
\item Bound the probability of these events occurring during the first $m$ iterations. 
\item Bound $\P [ B_s(\epsilon) | U_s = 1 ] $ conditioned on the events not occurring by analyzing the tree quantity.
\end{enumerate}

Before proceeding, we pause to state the choice of $\{ \beta_s(k) \}_{k \in K_n}$ from Step 1, which will be used in Section \ref{secAlgorithms}. First, for any realization of the DCM and for $v \in V_n \setminus K_n$, we define
\begin{equation}\label{eqModifiedMCtransition}
\tilde{P}_v = (1-\alpha_n) \tilde{P} + ( \alpha_n e_{V_n \setminus K_n} + e_{K_n} ) e_v^{\trans} , \quad \tilde{P}(i,j) = U_i P(i,j) , 
\end{equation} 
where $P$ is defined in Section \ref{secPprDefinition}. Note $\tilde{P}_v$ is the transition matrix of a Markov chain similar to that in Definition \ref{defnPpr}; however, upon reaching $K_n$, the random walker jumps back to $v$ with probability 1. Letting $\tilde{\pi}_v$ denote the stationary distribution of this chain, one can show (see Appendix \ref{secErrorMstep})
\begin{equation}\label{eqLinearCombo}
\pi_v(w) = \frac{\alpha_n U_w \tilde{\pi}_v(w) + \sum_{k \in K_n} \tilde{\pi}_v(k) \pi_k(w) }{\alpha_n + (1-\alpha_n) \tilde{\pi}_v(K_n) }\ \forall\ w \in V_n .
\end{equation}
We note \eqref{eqLinearCombo} is an alternate formulation of the Hubs Theorem. With \eqref{eqLinearCombo} in mind, we define
\begin{equation}\label{eqChoiceOfBetaV}
\beta_v(k) = \frac{ \tilde{\pi}_v(k) }{ \alpha_n + (1-\alpha_n) \tilde{\pi}_v(K_n)  }\ \forall\ k \in K_n ,
\end{equation}
and we take $\{ \beta_s(k) \}_{k \in K_n}$ as in \eqref{eqChoiceOfBetaV} in Step 1. We also note this provides another interpretation of Lemma \ref{lemMainTailBound}. Informally, since $\alpha_n \rightarrow 0$, \eqref{eqChoiceOfBetaV} implies $\sum_{k \in K_n} \beta_s(k) \approx 1$ for large $n$, so $\sum_{k \in K_n} \beta_s(k) \pi_k$ is nearly a convex combination. Hence, when $B_s(\epsilon)$ fails, $\pi_s - \alpha_n e_s^{\trans} \approx \pi_s$ is close to the convex hull of $\{ \pi_k \}_{k \in K_n}$, a small subset of the $n$-dimensional simplex to which $\pi_s$ belongs.

We now turn to the main result. First, note Lemma \ref{lemMainTailBound} will allow us to show the second summand in \eqref{eqDimMeasureRelax} is bounded (in expectation) by $O(n^{1-c(\epsilon)})$, which is sublinear. Hence, to ensure \eqref{eqDimMeasureRelax} is sublinear, it only remains to choose $\psi_n$ such that $|K_n|$ is sublinear as well. On the other hand, $\Omega_{n,6}$ in Assumption \ref{assDegSeq} requires $K_n$ to contain a constant fraction of all instubs, suggesting we should choose $K_n$ to be nodes with high in-degree. Together, these observations motivate our choice of $\psi_n$: for $\kappa \in (0,1)$ we define $\psi_{n,\kappa}$ as the function that chooses the $n^{\kappa}$ nodes of highest in-degree as $K_n$. Formally, $\psi_{n,\kappa}$ is the function that maps $(\mathbf{N}_n,\mathbf{D}_n)$ to $\mathbf{U}_n = \{ U_v : v \in V_n \}$ with $U_v = 1 ( \sigma^{-1}(v) > n^{\kappa} )$, where $\sigma : \{1,2,\ldots,n\} \rightarrow V_n $ is such that $N_{\sigma(1)} \geq N_{\sigma(2)} \geq \cdots \geq N_{\sigma(n)}$. 

With this in place, we present Theorem \ref{thmSublinear}. Together with Assumption \ref{assDegSeq}, it states the following: when certain moments of the degree sequence exist, and when a sublinear number of nodes contains a constant fraction of instubs, the dimension of the set of PPR vectors scales sublinearly. 
\begin{thm} \label{thmSublinear}
Assume $\exists\ \kappa \in (0,1)$ such that the sequence $\{ \mathbf{N}_n,\mathbf{D}_n,\mathbf{U}_n \}_{n \in \N}$ satisfies Assumption \ref{assDegSeq} when $\mathbf{U}_n = \psi_{n,\kappa} (\mathbf{N}_n,\mathbf{D}_n)\ \forall\ n \in \N$. Then $\forall\ \epsilon > 0$,
\begin{equation}
\E [ \Delta_{\psi_{n,\kappa}}(\epsilon) ]  = O \left( n^{ \max \{ \kappa, 1-c(\epsilon) \} } \right) ,
\end{equation}
where $c(\epsilon)$ is defined in Lemma \ref{lemMainTailBound}. As a consequence, $\forall\ \bar{c} \in ( \max \{ \kappa, 1-c(\epsilon) \} , 1 )$ and $C > 0$,
\begin{equation}
\P \Big[ \Delta_{\psi_{n,\beta}} ( \epsilon ) \geq C n^{\bar{c}} \Big] = O \left( n^{ \max \{ \kappa, 1-c(\epsilon) \} - \bar{c} } \right) .
\end{equation}
\end{thm}
\begin{prf}
See Appendix \ref{appProofOfThmSublinear}.
\end{prf}
To illustrate the theorem, we give an example in \eqref{eqToyExample}.  Here $\mathbf{U}_n = \psi_{n,\frac{1}{2}}(\mathbf{N}_n,\mathbf{D}_n)$ yields $(\mathbf{N}_n,\mathbf{D}_n,\mathbf{U}_n)$ satisfying Assumption \ref{assDegSeq}, i.e.\ the assumptions of Theorem \ref{thmSublinear} are satisfied with $\kappa = \frac{1}{2}$.
\begin{equation}\label{eqToyExample}
\mathbf{N}_n = ( \underbrace{ O(\sqrt{n}) , O(\sqrt{n}) , \ldots , O(\sqrt{n}) }_{\textrm{repeated $\sqrt{n}$ times}} , \underbrace{ O(1), O(1), \ldots , O(1) }_{\textrm{repeated $n - \sqrt{n}$ times}} ) , \quad \mathbf{D}_n = ( O(1), O(1), \ldots , O(1) ) .
\end{equation}

\subsection{Comments on assumptions} \label{secRemarksOnAssumption}

We begin with comments on $\Omega_{n,5}$ in Assumption \ref{assDegSeq}. First, note that, given $\Omega_{n,1}$ and $\Omega_{n,6}$, $\Omega_{n,5}$ implicitly requires $\frac{1}{n} \sum_{h=1}^n U_h$ to converge to a specific limit: indeed, assuming it converges,
\begin{equation}
\lambda^* = \lim_{n \rightarrow \infty} \frac{\sum_{h=1}^n U_h N_h}{\sum_{h=1}^n U_h} = \left( \lim_{n \rightarrow \infty} \frac{\sum_{h=1}^n U_h N_h}{\sum_{h=1}^n N_h} \right) \frac{ \lim_{n \rightarrow \infty} \frac{1}{n} \sum_{h=1}^n N_h}{\lim_{n \rightarrow \infty} \frac{1}{n} \sum_{h=1}^n U_h} = \frac{p \eta_1}{\lim_{n \rightarrow \infty} \frac{1}{n} \sum_{h=1}^n U_h} .
\end{equation}
With $|K_n|$ sublinear in Theorem \ref{thmSublinear}, $\frac{1}{n} \sum_{h=1}^n U_h \rightarrow 1$, so we require $\lambda^* = p \eta_1$.

We next argue $\Omega_{n,4}$ is not restrictive (at least in its own right). In fact, it is essentially implied by sublinearity of $|K_n|$ in Theorem \ref{thmSublinear} and $\Omega_{n,1}$, since then the fraction in $\Omega_{n,4}$ satisfies
\begin{equation}
\frac{\sum_{h=1}^n U_h D_h}{\sum_{h=1}^n U_h} \leq \frac{ \frac{1}{n} \sum_{h=1}^n D_h }{ \frac{1}{n} \sum_{h=1}^n U_h } \rightarrow \eta_1 \in (0,\infty) .
\end{equation}

Next, we note $\{ \Omega_{n,i} \}_{i=1}^3$ are similar to assumptions found in \cite{chen2017generalized} and are fairly standard given our approach, which leverages the fact that the random graph is asymptotically locally treelike \cite{bordenave2012lecture}. In fact, $\Omega_{n,3}$ is a weaker assumption than that required in \cite{chen2017generalized}, which is why (as mentioned in Section \ref{secRelated}) we use a modified version of one of their lemmas. See Appendix \ref{secCouplingFailure} for details.

Finally, $\Omega_{n,6}$ requires $\frac{1}{L_n} \sum_{v \in V_n} U_v N_v$ to converge to $p < 1$ with $|K_n|$ sublinear in Theorem \ref{thmSublinear}. We offer empirical evidence that this occurs for certain graphs of interest. Specifically, in Figure \ref{figFractionKn}, $\frac{1}{L_n} \sum_{v \in V_n} U_v N_v$ remains constant and strictly less than 1 as $n$ grows, for a variety of sublinear $|K_n|$ choices. For this plot, in-degrees were sampled from a power law distribution with exponent $2$, i.e.\ $\P [ N_v = x ] \propto x^{-2}$. This in-degree distribution is commonly seen in real graphs and has been studied extensively, e.g.\ \cite{barabasi1999emergence,clauset2009power}. As an example, Figure \ref{figDegreeHist} compares the histogram of these in-degrees with the in-degrees of the Twitter graph (available at \cite{lawDatasets} from WebGraph \cite{boldi2004webgraph}). The histograms are similar for most values of $x$; both are roughly linear with slopes $\approx -2$ over $x \in [10,5000]$. In short, a common model of in-degree distributions empirically satisfies $\Omega_{n,6}$ with $|K_n|$ sublinear.

\begin{figure}
\centering
\begin{subfigure}{0.49\textwidth}
\centering
\includegraphics[height=\plotHeight]{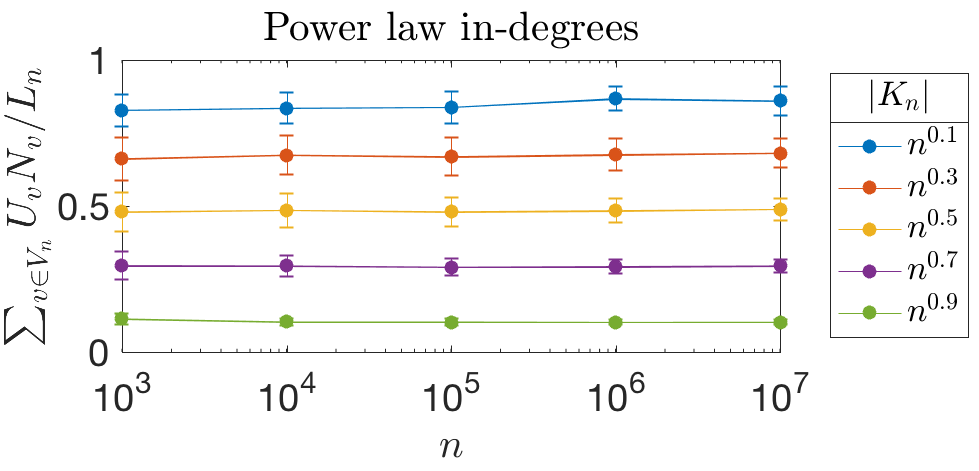}
\caption{For power law in-degrees, $V_n \setminus K_n$ contains a constant  fraction of instubs with $|K_n$ sublinear.} \label{figFractionKn}
\end{subfigure}%
\hspace{0.018\textwidth}%
\begin{subfigure}{0.49\textwidth}
\centering
\includegraphics[height=\plotHeight]{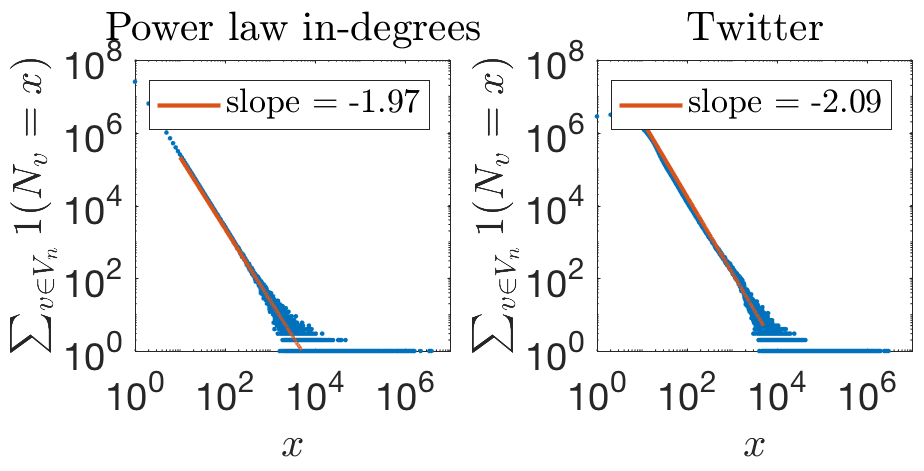}
\caption{The in-degrees for Fig.\ \ref{figFractionKn} are similar to in-degrees for Twitter graph from \cite{lawDatasets}. (Here $n \approx 4 \times 10^7$.)} \label{figDegreeHist}
\end{subfigure}
\caption{$\Omega_{n,6}$ is empirically satisfied with $|K_n|$ sublinear for power law in-degree distributions.} \label{figOmega6validation}
\end{figure}

%% file: algorithms.tex
\section{Algorithms and experiments} \label{secAlgorithms}

In this section, we use our dimensionality analysis to analyze the algorithm from \cite{jeh2003scaling} mentioned in Section \ref{secRelated}. We then present empirical results to complement our analysis.

\subsection{Algorithm to estimate $\{ \pi_v \}_{v \in V_n}$} \label{secSchemeToEstimate}

In Section 4.4.3 of \cite{jeh2003scaling}, Jeh and Widom propose the following algorithm to estimate $\{ \pi_v \}_{v \in V_n}$. First, compute $\{ \pi_k \}_{k \in K_n}$. Next, for $v \in V_n \setminus K_n$, compute \eqref{eqChoiceOfBetaV} and estimate $\pi_v$ as
\begin{equation}\label{eqHatPiV}
\hat{\pi}_v = \frac{ \sum_{k \in K_n} \tilde{\pi}_v(k) \pi_k }{ \alpha_n + (1-\alpha_n) \tilde{\pi}_v(K_n) } .
\end{equation}
The basic idea behind this scheme is that, from \eqref{eqLinearCombo}, $\hat{\pi}_v$ may be close to $\pi_v$; however, no formal analysis is provided. Here we show that our dimensionality result provides such an analysis. 

For this, letting $\epsilon > 0$ and using \eqref{eqLinearComboIntoL1NoAbsVal} from Appendix \ref{secProofLemL1bound}, it is straightforward to show
\begin{equation}\label{eqErrorBoundRuntime}
\left\| \pi_v - \left( \alpha_n e_v^{\trans} + \hat{\pi}_v \right) \right\|_1 < \epsilon \Leftrightarrow \tilde{\pi}_v(K_n) >  \frac{ \alpha_n ( 1 - ( \epsilon + \alpha_n ) ) }{   \epsilon + \alpha_n ( 2 - ( \epsilon + \alpha_n ) )   } .
\end{equation}
In other words, \eqref{eqErrorBoundRuntime} shows we can use $\{ \tilde{\pi}_v(k) \}_{k \in K_n}$ to compute the $l_1$ estimation error indirectly, i.e.\ without actually computing $\pi_v$. This suggests a new scheme, which proceeds as follows. First, compute $\{ \pi_k \}_{k \in K_n}$ (as in the existing scheme). Next, for $v \in V_n \setminus K_n$, compute $\{ \tilde{\pi}_v(k) \}_{k \in K_n}$. If \eqref{eqErrorBoundRuntime} holds, estimate $\pi_v$ as $\alpha_n e_v^{\trans} + \hat{\pi}_v$; else, compute $\pi_v$. 

Using this scheme, we either compute $\pi_v$ exactly, or we obtain an estimate within $\epsilon$ of $\pi_v$ (in the $l_1$ norm), $\forall\ v \in V_n$. The remaining question is the scheme's complexity, which we take to be the number of PPR values that are computed. First, for $k \in K_n$, $n$ such values ($\pi_k$) are computed. Next, for $v \in V_n \setminus K_n$, $|K_n|$ such values ($\{ \tilde{\pi}_v(k) \}_{k \in K_n}$) are computed. Finally, an additional $n$ such values ($\pi_v$) are computed for $v \in V_n \setminus K_n$ s.t.\ \eqref{eqErrorBoundRuntime} fails; by definition, this occurs for $\Delta_{\psi_{n,\kappa}}(\epsilon) - |K_n|$ such $v$ when $K_n$ is chosen by $\psi_{n,\kappa}$. Hence, the number of PPR values computed is
\begin{equation}
n |K_n| + |K_n| | V_n \setminus K_n | + n ( \Delta_{\psi_{n,\kappa}}(\epsilon) - |K_n| ) \leq 2 n \Delta_{\psi_{n,\kappa}}(\epsilon) = O ( n \Delta_{\psi_{n,\kappa}}(\epsilon) ) ,
\end{equation}
which is sub-quadratic with high probability when Theorem \ref{thmSublinear} applies. (We have assumed the computation of $\{ \tilde{\pi}_v(k) \}_{k \in K_n}$ is no more costly than the computation of PPR values on the original graph; this is because $\{ \tilde{\pi}_v(k) \}_{k \in K_n}$ are computed on a sparser graph.) Hence, all $n$ PPR vectors can be accurately estimated by computing a vanishing fraction of the $n^2$ vector elements.

Finally, we remark that this scheme can also be viewed as approximating $\Pi$, the matrix with $v$-th row $\Pi(v,:) = \pi_v$. To see this, let $\hat{\Pi}_{\epsilon}$ be the estimate of $\Pi$ from the scheme, i.e.\ $\hat{\Pi}_{\epsilon}(v,:) = \alpha_n e_v^{\trans} + \hat{\pi}_v$ if $v \in V_n \setminus K_n$ and \eqref{eqErrorBoundRuntime} holds, $\hat{\Pi}_{\epsilon}(v,:) = \pi_v$ otherwise. Then, by \eqref{eqErrorBoundRuntime}, $\| \Pi(v,:) - \hat{\Pi}_{\epsilon}(v,:) \|_1 < \epsilon\ \forall\ v \in V_n$, so $\| \Pi - \hat{\Pi}_{\epsilon}  \|_{\infty} < \epsilon$ (where $\| A \|_{\infty} = \max_{ \| x \|_{\infty} = 1 } \| A x \|_{\infty} = \max_i \sum_j | A(i,j) |$ is the $l_{\infty}$ norm of the matrix $A$). Hence, the scheme approximates $\Pi$ with bounded error in the $l_{\infty}$ norm.

\subsection{Empirical results}

We now demonstrate the performance of this algorithm using two datasets from the Stanford Network Analysis Platform (SNAP) \cite{snapnets}: soc-Pokec, a social network, and web-Google, a partial web graph (see Appendix \ref{appDatasets} for details). For both graphs, we choose the top $n^{0.8}$ nodes by in-degree as $K_n$ (i.e.\ $\mathbf{U}_n = \psi_{n,0.8}(\mathbf{N}_n,\mathbf{D}_n)$), set $\alpha_n = \frac{1}{\log n}$, and, $\forall\ v \in V_n \setminus K_n$, compute a bound on the error $\| \pi_v - ( \alpha_n e_v^{\trans} + \hat{\pi}_v ) \|_1$ using a power iteration scheme described in Appendix \ref{appEstErrBound}. Figure \ref{figAllErrPlotsHist} shows histograms of the error bound, while Figure \ref{figAllErrPlotsDim} shows our dimensionality measure. Note (as proven in Appendix \ref{appEstErrBound}), error is zero when $v \in V_{n,0}$, where
\begin{equation}\label{eqVn0defn}
V_{n,0} = \left\{ v \in V_n \setminus K_n : \not \exists\ (w,w') \in E_n \textrm{ s.t.\ } w = v, w' \in V_n \setminus K_n \right\} .
\end{equation}
(In words, the error is zero when no outgoing neighbors of $v$ belong to $V_n \setminus K_n$.) As a result, the spikes at $\epsilon = 0$ in Figure \ref{figAllErrPlotsHist} have height $|V_{n,0}|/n$, and $\Delta_{n,\kappa}(0) = |V_{n,0}|$ in Figure \ref{figAllErrPlotsDim}. Additionally, we show in Appendix \ref{appEstErrBound} that error is bounded by $(1-\alpha_n)$; hence, the spikes at right in Figure \ref{figAllErrPlotsHist}, and the ``dips'' at right in Figure \ref{figAllErrPlotsDim}, occur at $\epsilon = (1-\alpha_n$). Between these spikes, the soc-Pokec histogram quickly decay beyond $\epsilon \approx 0.3$; this corresponds to the dimensionality being nearly flat beyond $\epsilon \approx 0.3$ in Figure \ref{figAllErrPlotsDim}. (For web-Google, similar behavior occurs, though it is less pronounced). Finally, we highlight two points on Figure \ref{figAllErrPlotsDim}, $( \frac{1-\alpha_n}{3}, 0.09 )$ for soc-Pokec and $( \frac{1-\alpha_n}{3}, 0.15 )$ for web-Google. The soc-Pokec point, for example, shows that computing $9 \%$ of PPR vectors guarantees the $l_1$ estimation error for other PPR vectors is below $\frac{1-\alpha_n}{3}$ (i.e.\ the worst-case error is reduced by a factor of 3). See Appendix \ref{appAllErrorExp} for further empirical results for these datasets.

Figure \ref{figAllErrPlotsDim} also highlights another aspect of $\Delta_{n,\kappa}(\epsilon)$. Specifically, the discussion at the end of Section \ref{secSchemeToEstimate} and the steep decay in $\Delta_{n,\kappa}(\epsilon)$ Figure \ref{figAllErrPlotsDim} suggests that most of the ``energy'' of $\Pi$ is contained in a small number of dimensions, in the $l_{\infty}$ norm. Hence, $\Delta_{n,\kappa}(\epsilon)$ is roughly analogous to stable rank, a more common dimensionality measure that instead measures energy using singular values (namely, stable rank is $\sum_i \sigma_i^2 / \sigma_1^2$, where $\{ \sigma_i \}$ are the ordered singular values).

In Appendix \ref{appEstErrBound}, we also describe how the power iteration scheme allows us to compute a bound on the average error $\frac{1}{ |V_n \setminus K_n|} \sum_{v \in V_n \setminus K_n} \| \pi_v - ( \alpha_n e_v^{\trans} + \hat{\pi}_v ) \|_1$ indirectly (i.e., without actually computing the error for each $v$). Hence, we show the average error bound for a wider variety of SNAP datasets in Figure \ref{figRealAvgErrPlots}. Interestingly, the two social networks soc-LiveJournal1 and soc-Pokec have similar behavior, as do the two web graphs web-BerkStan and web-Stanford (web-Google is somewhat of an outlier; we believe its average error is lowest in part because its $|V_{n,0}|$ is largest). Finally, in Figure \ref{figSynAvgErrPlots}, we show the average error bound computed on a DCM with power law in-degrees. As suggested by Lemma \ref{lemMainTailBound}, average error shrinks as $n$ grows (despite $|K_n|/n$ shrinking as well); this is in part because, from Figure \ref{figFractionKn}, the fraction of instubs belonging to $K_n$ is constant.

\begin{figure}[t]
\centering
\begin{subfigure}{0.49\textwidth}
\centering
\includegraphics[height=\plotHeight]{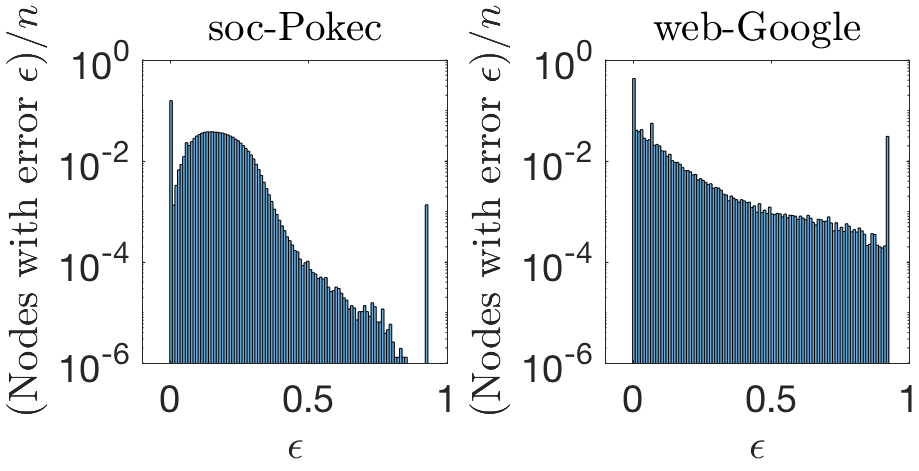}
\caption{(Normalized) histograms of the estimation error $\| \pi_v - ( \alpha_n e_v^{\trans} + \hat{\pi}_v ) \|_1$ $\forall\ v \in V_n \setminus K_n$.} \label{figAllErrPlotsHist}
\end{subfigure}%
\hspace{0.018\textwidth}%
\begin{subfigure}{0.49\textwidth}
\centering
\includegraphics[height=\plotHeight]{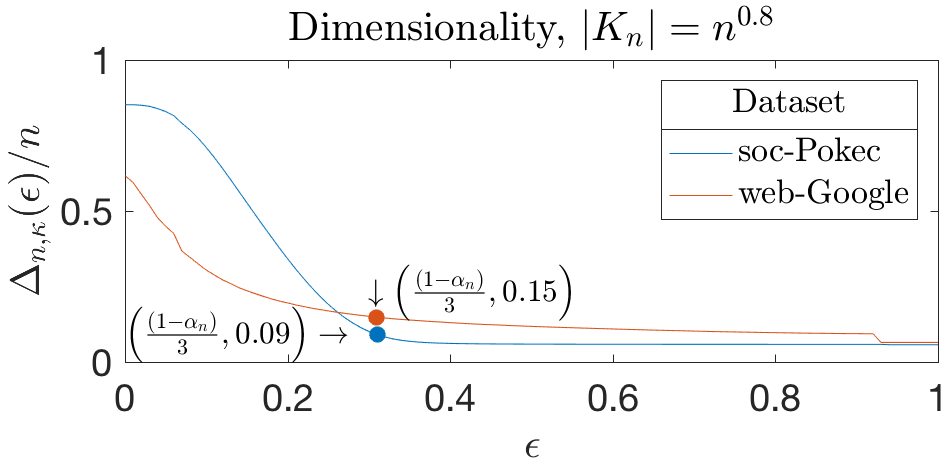}
\caption{Computing 9\% or 15\% of PPR vectors reduces worst-case estimation error by a factor of 3.} \label{figAllErrPlotsDim}
\end{subfigure}
\caption{Error and dimensionality for soc-Pokec (social network) and web-Google (web graph)} \label{figAllErrPlots}
\end{figure}

\begin{figure}[t]
\centering
\begin{subfigure}{0.49\textwidth}
\centering
\includegraphics[height=\plotHeight]{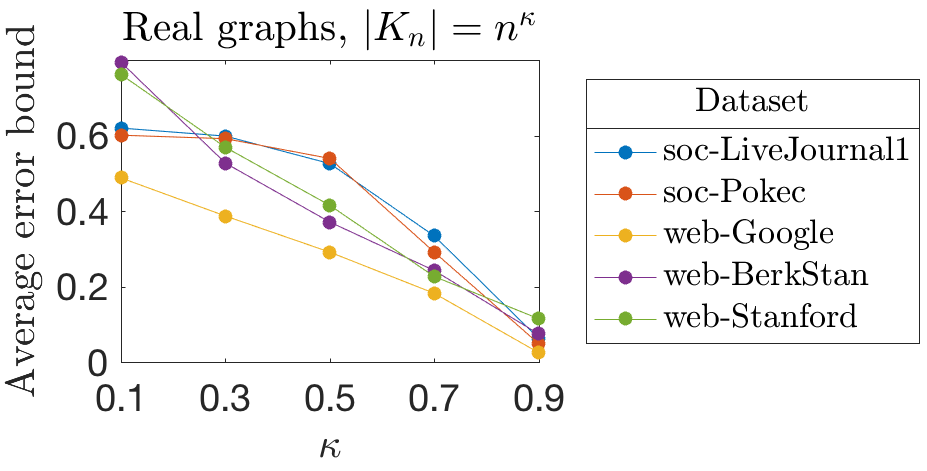}
\caption{Average error decreases as $|K_n|$ grows for a variety of social networks and web graphs.} \label{figRealAvgErrPlots}
\end{subfigure}%
\hspace{0.018\textwidth}%
\begin{subfigure}{0.49\textwidth}
\centering
\includegraphics[height=\plotHeight]{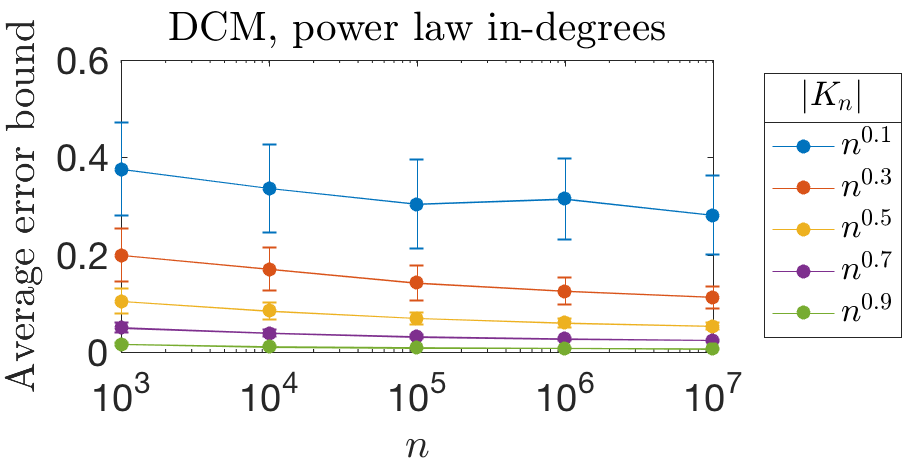}
\caption{For $s \sim V_n \setminus K_n$ uniformly on the DCM, error decreases as $n$ grows, despite $|K_n|/n$ decreasing.} \label{figSynAvgErrPlots}
\end{subfigure}%
\caption{Average error experiments for real and synthetic datasets} \label{figAvgErrPlots}
\end{figure}

%% file: conclusions.tex
\section{Conclusions} \label{secConclusions}

In this work, we argued (analytically for the DCM and empirically for other graphs) that the dimensionality of $\{ \pi_v \}_{v \in V_n}$ scales sublinearly in $n$. We also used our analysis to bound the complexity of the algorithm from \cite{jeh2003scaling}. Our analysis suggests several avenues for future work. First, the proof of Lemma \ref{lemMainTailBound} can be modified to analyze the tail of the $l_{\infty}$ error (this would essentially involve replacing Lemma \ref{lemTreeTailBound} with a tail bound on a maximum instead of a sum). Hence, bounding \textit{absolute} error for the estimate of $\pi_s(v)$ for any $v \in V_n$ is a straightforward extension; a more useful but less immediate analysis would involve bounding \textit{relative} error. Second, examining PPR dimensionality for other random graph models may be of interest. For example, several papers have analyzed PPR on preferential attachment models \cite{avrachenkov2006pagerank,garavaglia2018local}; we suspect a dimensionality analysis for such graphs would yield a message similar to our work ($K_n$ should contain nodes with highest in-degree). A more interesting class of graphs would be the stochastic block model; here it may be more beneficial to choose $K_n$ such that each community contains a nonempty subset of $K_n$.

%% file: mainResultProof.tex
\section{Lemma \ref{lemMainTailBound} proof outline} \label{secMainProof}

In this appendix, we outline the proof of Lemma \ref{lemMainTailBound}. Our approach follows the outline described in Section \ref{secMainResultPresent}. Specifically, we consider Steps 1-4 of the outline in Appendices \ref{secErrorMstep}-\ref{secBranchingTailBound}, respectively. In Appendix \ref{secMainBound}, we combine the results to prove the lemma.

\subsection{Error bound in $m$-step neighborhood (Step 1)} \label{secErrorMstep}

Our first goal is to bound the error term $\| \pi_s - ( \alpha_n e_s^{\trans} + \sum_{k \in K_n} \beta_s(k) \pi_k ) \|_1$ for a particular choice of $\{ \beta_s(k) \}_{k \in K_n}$. For this, we require an intermediate result; namely, \eqref{eqLinearCombo} from Section \ref{secMainResultPresent}, which we formalize as Lemma \ref{lemLinearCombo} here. Recall from Section \ref{secMainResultPresent} that $\tilde{\pi}_s$ is the stationary distribution of the Markov chain with transition matrix $\tilde{P}_s = (1-\alpha_n) \tilde{P} + ( \alpha_n e_{V_n \setminus K_n} + e_{K_n} ) e_s^{\trans}$, where $\tilde{P}$ satisfies $\tilde{P}(i,j) = U_i P(i,j)$.

As mentioned in Section \ref{secMain}, Lemma \ref{lemLinearCombo} is an alternate formulation of the Hubs Theorem from \cite{jeh2003scaling}. At a high level, both formulations view $\pi_s(v)$ as the probability of paths from $s$ to $v$ and partition these paths into those that avoid $K_n$ (which have probability proportional to $\tilde{\pi}_s(v)$) and those through $K_n$ (which have probability proportional to $\tilde{\pi}_s(k) \pi_k(v)$). The difference between the two formulations is that we explicitly construct a new Markov chain that does not include paths through $K_n$ (i.e.\ the chain with transition matrix $\tilde{P}_s$), while \cite{jeh2003scaling} does not. Our formulation admits a probabilistic proof; in contrast, the proof in \cite{jeh2003scaling} is linear algebraic.

\begin{lem} \label{lemLinearCombo}
Consider any realization of the DCM and assume $U_s = 1$. Then
\begin{equation}
\pi_s(v) = \frac{\alpha_n U_v \tilde{\pi}_s(v) + \sum_{k \in K_n} \tilde{\pi}_s(k) \pi_k(v) }{\alpha_n + (1-\alpha_n) \tilde{\pi}_s(K_n) }\ \forall\ v \in V_n .
\end{equation}
\end{lem}
\begin{proof}
See Appendix \ref{secProofLemLinearCombo}.
\end{proof}

We next bound the error term using a particular $\{ \beta_s(k) \}_{k \in K_n}$; namely, that suggested by Lemma \ref{lemLinearCombo}. Our bound leverages the fact that the transition matrix $\tilde{P}_s$ is written as the sum of two matrices, one of which is rank one. This allows us to use the Sherman-Morrison-Woodbury formula (see e.g.\ Section 6.4 of \cite{laub2005matrix}) to derive a bound on the error term in terms of the row vector 
\begin{equation}\label{eqMuMDefn}
\mu_{s}^{(m)} = e_s^{\trans} \sum_{j=0}^m (1-\alpha_n)^j \tilde{P}^j ,
\end{equation}
which, as desired, only depends on the $m$ step neighborhood out of $s$.

\begin{lem} \label{lemL1bound}
Consider any realization of the DCM and assume $U_s = 1$. For all $k \in K_n$, let
\begin{equation}
\beta_{s}(k) = \frac{\tilde{\pi}_s(k)}{\alpha_n + (1-\alpha_n) \tilde{\pi}_s(K_n) } .
\end{equation}
Then for each $m \in \N$,
\begin{equation}
\left\| \pi_s - \left( \alpha_n e_s^{\trans} + \sum_{k \in K_n} \beta_s(k) \pi_k \right) \right\|_1 \leq \alpha_n \left( \mu_s^{(m-1)}(V_n \setminus K_n) - 1 \right)  + e_s^{\trans} (1-\alpha_n)^m \tilde{P}^m e_{V_n \setminus K_n} .
\end{equation}
\end{lem}
\begin{proof}
See Appendix \ref{secProofLemL1bound}.
\end{proof}

\subsection{Coupling with branching process (Step 2)} \label{secCoupling} 

Next, we show that the error bound in Lemma \ref{lemL1bound} follows the same distribution as a related quantity defined in terms of a branching process. Before presenting this result, we formally define the DCM construction and the branching process.

We begin with the DCM. As described in Section \ref{secDcmDescription}, the basic idea is to randomly pair outgoing half-edges (which we call \textit{outstubs}) with incoming half-edges (which we call \textit{instubs}) in a breadth-first search fashion.  We begin by sampling a node $s$ uniformly at random from $V_n$. In the first iteration, for each outstub belonging to $s$, we sample an instub uniformly (resampling if the sampled instub has already been paired), and we pair the outstub and instub. We allow the possibility that the sampled instub belongs to $s$ (in which case a self-loop is formed) or that multiple outstubs of $s$ are paired with instubs belonging to the same node (in which case multiple edges are formed between $s$ and that node).\footnote{Because of this, the resulting graph will in general be a multi-graph. We note the authors of \cite{chen2013directed} prove that a simple graph (no self-loops or multi-edges) results with positive probability as $n \rightarrow \infty$; however, this requires stronger assumptions on the degree sequence than Assumption \ref{assDegSeq}, which is all that we require to prove our main result. Specifically, guaranteeing that a simple graph emerges with positive probability as $n \rightarrow \infty$ requires empirical variances of the in- and out-degree sequences to converge.}

At the conclusion of the first iteration, we denote by $A_1$ the subset of $V_n \setminus \{ s \}$ containing those nodes that have had at least one instub paired with an outstub of $s$. In the second iteration, we pair all outstubs of all nodes in $A_1$ in the manner described previously. In general, we pair all outstubs of all nodes in $A_{m-1}$ during the $m$-th iteration, where $A_{m-1}$ is the set of nodes $v$ at distance $m-1$ from $s$. In other words, paths out of $s$ of length $m$ are constructed during the $m$-th iteration. When all outstubs of all nodes have been paired, the construction is complete. 

To facilitate the graph construction and the coupling argument, we define labels for each instub $e$ and for each node $v$, denoted $g(e)$ and $g(v)$. The instub label $g(e)$ is necessary because if $e$ is sampled for pairing, we must check whether $e$ has already been paired. Hence, we define
\begin{equation} \label{eqInstubLabels}
g(e) = \begin{cases} 1 , & \textrm{$e$ is currently unpaired} \\ 0, & \textrm{$e$ is currently paired} \end{cases} .
\end{equation}
The node label $g(v)$ is defined as
\begin{equation}\label{eqNodeLabels}
g(v) = \begin{cases} A, & \textrm{$v$ does not currently belong to graph} \\ B, & \textrm{$v$ belongs to graph, $U_v = 0$} \\ C,  & \textrm{$v$ belongs to graph, $U_v = 1$, all paths from $s$ to $v$ visit some $w \in V_n$ s.t.\ $U_w = 0$} \\  D, & \textrm{$v$ belongs to graph, $U_v = 1$, some path from $s$ to $v$ avoids all $w \in V_n$ s.t.\ $U_w = 0$} \end{cases} .
\end{equation}

To illustrate these node labels, we show a graph after three iterations of the construction in Figure \ref{figGraphLabelExample}. The node at the top of the figure is $s$. Circle and square nodes, respectively, depict those nodes $v$ with $U_v = 1$ and $U_v = 0$, respectively (i.e., those belonging to $V_n \setminus K_n$ and $K_n$, respectively). Short arrows depict half-edges (i.e.\ unpaired instubs and  outstubs), while longer arrows depict edges (i.e.\ instubs and outstubs that have been paired). Node labels are displayed on each node.

Node labels will be useful in the coupling argument to come. In particular, the term $\mu_s^{(m)} ( V_n \setminus K_n )$ in Lemma \ref{lemL1bound} only depends on the subgraph containing label $D$ nodes within $m$ steps of $s$ (depicted by orange dashed edges in Figure \ref{figGraphLabelExample} for $m = 3$). This observation follows since $\mu_s^{(m)}(v)$ (by definition) is nonzero if and only if there exists a path from $s$ to $v$ that avoids $K_n$. 

\begin{figure}
\centering
\includegraphics[height=1.8in]{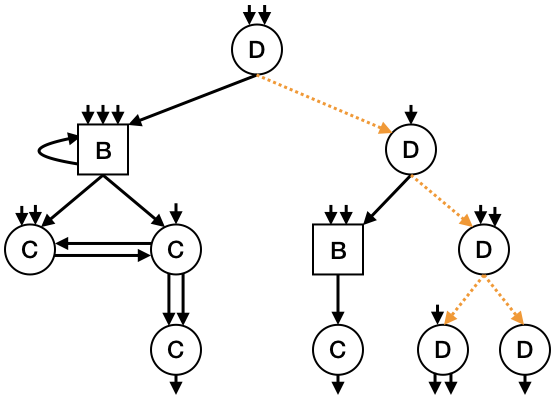}
\caption{Graph after three iterations; $\mu_s^{(3)}(V_n \setminus K_n)$ depends only on orange dashed subgraph}
\label{figGraphLabelExample}
\end{figure}

The formal graph construction is given in Algorithm \ref{algGraphConstruction}. Most notation has already been defined, but we use three additional pieces of notation in the algorithm: $I_n$ is the set of all instubs, $\{ (v',j) \}_{j=1}^{D_{v'}}$ is the set of outstubs belonging to $v' \in V_n$, and $\tau_G$ is a variable that tracks the first iteration at which certain events occur (these events relate to the coupling and will be discussed shortly). 

Before proceeding, we offer several comments to relate Algorithm \ref{algGraphConstruction} to the preceding (informal) description of the graph construction:
\begin{itemize}
\item In Line \ref{algInitS}-\ref{algInitOthers}, we initialize the algorithm. Namely, we sample the first node $s$, define the label $g(s)$ according to \eqref{eqNodeLabels}, and set $A_0 = \{ s \}$ (i.e.\ the only node at distance zero from $s$ is $s$ itself). We then set $g(e) = 1$ for all instubs $e$ (since no instubs have been paired) and $g(v) = A\ \forall\ v \neq s$ (since only $s$ belongs to the graph at this stage of the algorithm).
\item The remainder of the algorithm involves iterating over $m$ (outer for loop), iterating over nodes $v'$ at distance $m-1$ from $s$ (middle for loop), and iterating over outstubs belonging to $v'$ (inner for loop). For each such outstub, denoted $(v',j)$, the following steps occur:
\begin{itemize}
\item In Lines \ref{algBeginFindInstub}-\ref{algEndFindInstub}, we uniformly sample an instub $e$, resampling (if needed) until an unpaired instub is found. (Line \ref{algTauGresample} relates to the coupling argument and will be discussed shortly.)
\item After sampling an unpaired instub $e$, we pair $(v',j)$ with $e$ and set $g(e) = 0$ to reflect the fact that $e$ has been paired (Line \ref{algPair}). If the node $v$ to which $e$ belongs did not previously belong to the graph (i.e.\ if $g(v) = A$), then $v$ is at distance $m$ from $s$, so we add $v$ to $A_m$ (Line \ref{algAddToAm}). (Line \ref{algTauGdToCD} relates to the coupling and will be discussed shortly.)
\item In Lines \ref{algBeginLabelUpdate}-\ref{algEndLabelUpdate}, we update the label of $v$ according to \eqref{eqNodeLabels}. Note that, if $g(v') = D$ and $g(v) = C$, \eqref{eqNodeLabels} implies that a path from $s$ to $v$ avoiding $K_n$ did not exist before $(v',j)$ and $e$ were paired, but now such a path does exist. Hence, if some node $w$ s.t.\ $g(w) = C$ can be reached from $v$ while avoiding $K_n$, a path from $s$ to $w$ avoiding $K_n$ now exists as well. For this reason, we must change the label of such $w$ from $C$ to $D$ (Line \ref{algEndLabelUpdate}).
\end{itemize}
\item After these steps occur, if all instubs have been paired, the algorithm terminates (Line \ref{algTermination}).
\end{itemize}

\begin{algorithm}[t]
\DontPrintSemicolon
\caption{Graph Construction} \label{algGraphConstruction}

Choose $s$ from $V_n$ uniformly, set $g(s) = D$ if $U_s = 1$ and $g(s) = B$ if $U_s = 0$, set $A_0 = \{ s \}$ \label{algInitS}

Set $g(e) = 1\ \forall\ e \in I_n$, set $g(v) = A\ \forall\ v \in V_n \setminus \{ s \}$, set $\tau_G = \infty$ \label{algInitOthers}

\For{$m=1$ \KwTo $\infty$}{

Set $A_m = \emptyset$

\For{$v' \in A_{m-1}$}{

\For{$j = 1$ \KwTo $D_{v'}$}{ 

// find instub for pairing

Uniformly sample instub $e$ \label{algBeginFindInstub}

\lIf{$g(e) = 0, \tau_G = \infty$}{set $\tau_G= m$} \label{algTauGresample}

\While{$g(e) = 0$}{

Uniformly sample instub $e$

} \label{algEndFindInstub}

Pair $(v',j)$ with $e$, set $g(e) = 0$, denote instub node by $v$ \label{algPair}

\lIf{$g(v) = A$}{set $A_m = A_m \cup \{ v \}$} \label{algAddToAm}

\lIf{$g(v') = D, g(v) \in \{C,D\}, \tau_G = \infty$}{set $\tau_G = m$} \label{algTauGdToCD}

// update label

\lIf{$U_v = 0, g(v) = A$}{set $g(v) = B$} \label{algBeginLabelUpdate}
\lElseIf{$U_v = 1, g(v') = B, g(v) = A$}{set $g(v) = C$}
\lElseIf{$U_v = 1, g(v') \in \{C,D\}, g(v)= A$}{set $g(v) = g(v')$}
\lElseIf{$g(v') = D, g(v) = C$}{set $g(v) = D$, set $g(w) = D\ \forall\ w \in V_n$ s.t.\ $g(w) = C$ and $v \rightarrow w$ path avoiding all $w' \in V_n$ s.t.\ $U_{w'} = 0$ exists} \label{algEndLabelUpdate}

// termination

\lIf{$g(e') = 0\ \forall\ e' \in I_n$}{{\bf{return}}} \label{algTermination}

}

}

}

\end{algorithm}

Our next goal is to define a branching process and a quantity related to the error bound in Lemma \ref{lemL1bound}, so that this error bound can instead be analyzed on the tree resulting from the branching process. Before defining this tree construction, we offer some intuition, which helps explain $\tau_G$.

First, recall the error bound in Lemma \ref{lemL1bound} depends only on the $m$-step neighborhood out of $s$. Hence, a typical approach to analyzing the bound would be to argue that this neighborhood is treelike, and then to analyze the bound on a related tree. However, this is more than we require. To see this, we return to the example from Figure \ref{figGraphLabelExample}. As argued previously, the error bound only depends on the orange dashed subgraph. Hence, the related tree we construct will (roughly speaking) only contain this subgraph. Put differently, rather than require the entire $m$-step neighborhood to be treelike, we only require the $m$-step neighborhood of label $D$ nodes to be treelike.

This discussion also helps explain the variable $\tau_G$ in Algorithm \ref{algGraphConstruction}. Observe that we set $\tau_G = m$ if we pair an outstub of $v' \in A_{m-1}$ with an instub of $v$, where $g(v') = D$ and $g(v) \in \{ C,D \}$ (Line \ref{algTauGdToCD} in Algorithm \ref{algGraphConstruction}). As shown in Figure \ref{figTauGexamples}, these events (may) destroy the tree structure of the label $D$ subgraph. Additionally, we set $\tau_G = m$ if we sample an instub that has already been paired while attempting to pair an outstub of $v' \in A_{m-1}$ (Line \ref{algTauGresample} in Algorithm \ref{algGraphConstruction}). This is to ensure nodes have \textit{i.i.d.}\ attributes $(N_v, D_v, U_v)$, as nodes in the tree construction will have.

\begin{figure}
\centering
\begin{subfigure}{\textwidth}
  \centering
  \includegraphics[width=5in]{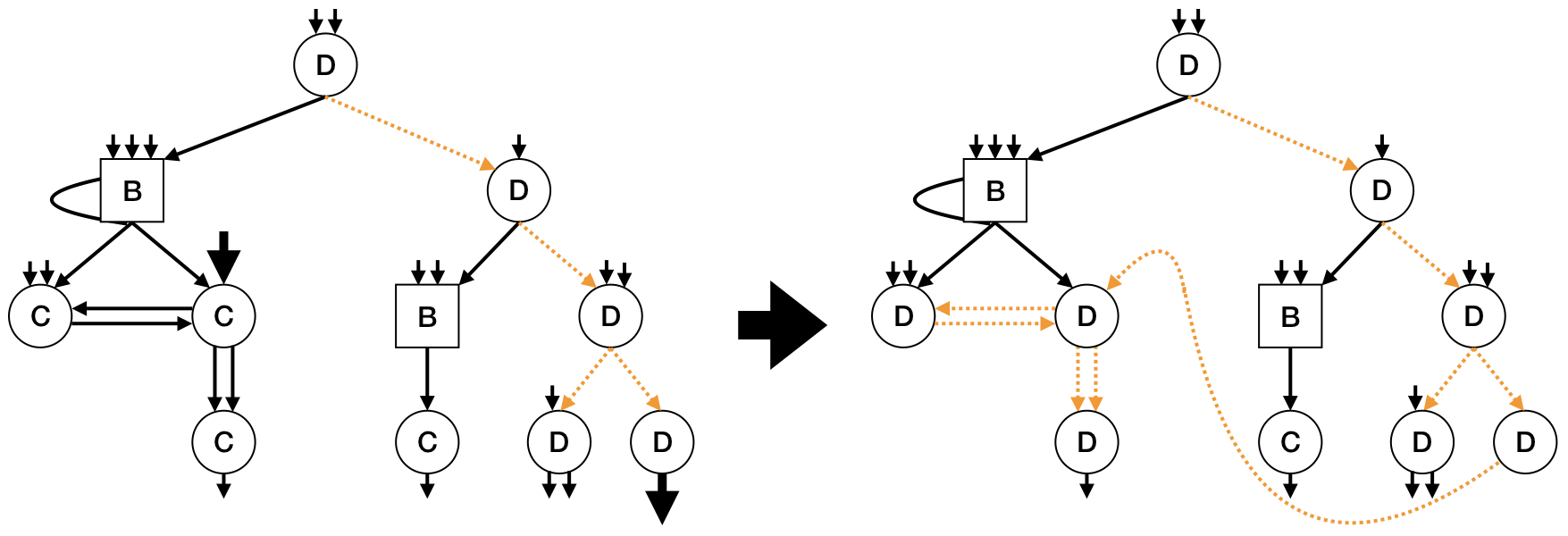}
  \caption{Enlarged outstub belongs to label $D$ node, enlarged instub belongs to label $C$ node}
  \label{figTauGexampleDC}
\end{subfigure}
\begin{subfigure}{\textwidth}
  \centering
  \includegraphics[width=5in]{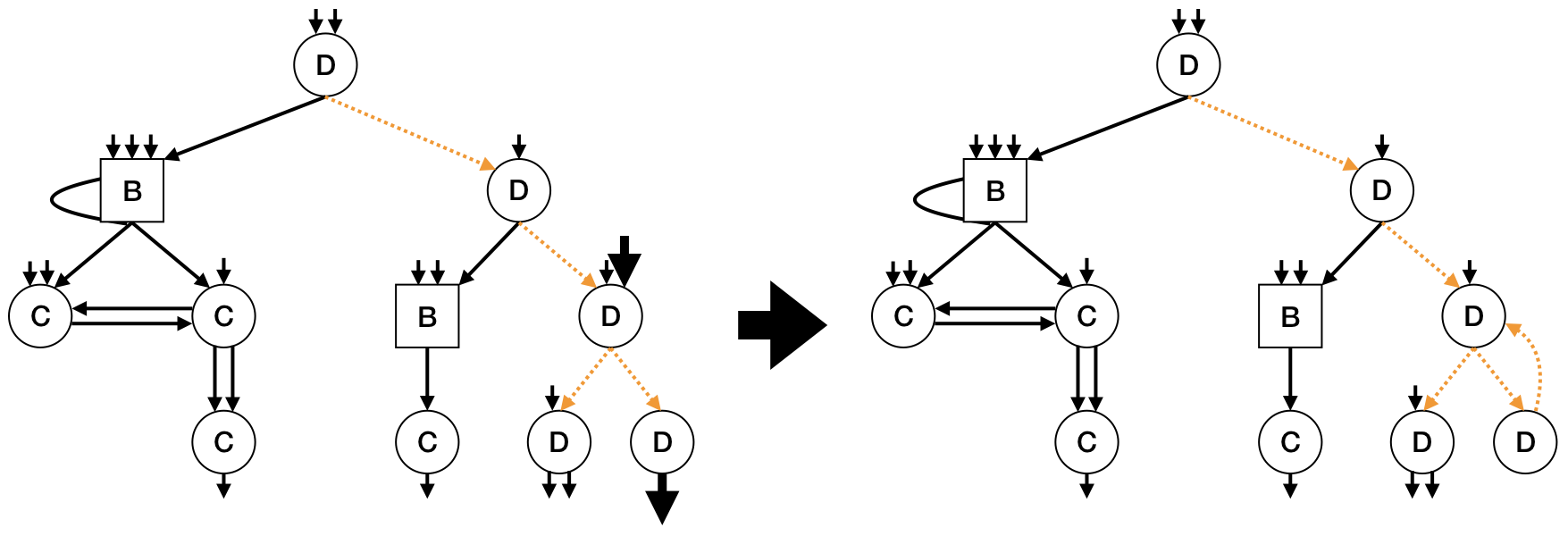}
  \caption{Enlarged outstub belongs to label $D$ node, enlarged instub belongs to label $D$ node}
  \label{figTauGexampleDD}
\end{subfigure}
\caption{If the enlarged instub is sampled for pairing with the enlarged outstub at left, then, after updating node labels, the orange dashed subgraph of label $D$ nodes at right is no longer a tree.}
\label{figTauGexamples}
\end{figure}

This intuition motivates our tree construction. In the tree construction, we begin with a root node denoted by $\phi$, and we assign attributes $(N_{\phi}, D_{\phi}, U_{\phi})$. Here $N_{\phi}$ is the number of instubs of $\phi$, all of  which will remain unpaired for the duration of the algorithm (so that the tree structure is maintained); $D_{\phi}$ is the number of offspring of $\phi$; and $U_{\phi} = 1$. To each offspring of $\phi$, denoted $1,2,\ldots,D_{\phi}$, we assign attributes $(N_i, D_i, U_i)$. Here $N_i$ denotes the number of instubs of $i$; one of these is paired with the $i$-th outstub of $\phi$, while the other $N_i-1$ remain unpaired (again, to preserve the tree structure).  Furthermore, unlike the root node, node $i$ receives $D_i$ offspring \textit{only if} $U_i = 1$; otherwise, the outstubs remain unpaired. This is explained by Figure \ref{figGraphLabelExample}, since only the orange dashed subgraph affects the quantity of interest.

The set of nodes $1,2,\ldots,D_{\phi}$ is denoted by $\hat{A}_1$. In general, we denote by $\hat{A}_m$ the $m$-th generation of the tree, i.e.\ the set of nodes at distance $m$ from the root node. The generic node in $\hat{A}_m, m > 1$ is denoted by $\i$, where $\i = (i_1, i_2, \ldots , i_m)$ is an ordered list of natural numbers that traces the unique path from $\phi$ to $\i$: specifically, this path is $\phi, i_1, (i_1,i_2), \ldots, \i$. The offspring of $\i$ (assuming $U_{\i} = 1$) are denoted by $\{ (\i,j) \}_{j=1}^{D_{\i}}$, where $(\i,j) = (i_1, i_2, \ldots , i_m, j)$ is the concatenation operation.

To assign attributes, we require two distributions: given $\{ N_h, D_h, U_h \}_{h=1}^n$, we define $f_n : \N \times \N \times \{0,1\} \rightarrow [0,1]$ and $f_n^* : \N \times \N \rightarrow [0,1]$ according to \eqref{eqTreeConstDist}. Note that $f_n$ is the distribution of node attributes for nodes sampled proportional to in-degree, whereas $f_n^*$ is the distribution of node attributes for nodes sampled uniformly at random from $V_n \setminus K_n$. Because non-root nodes are sampled proportional to in-degree in the graph construction (until an edge must be resampled, i.e.\ until we set $\tau_G = m$), non-root node attributes are sampled from $f_n$ in the tree construction. Similarly, since the first node is sampled uniformly from $V_n \setminus K_n$ in the case of interest of the graph construction, root node attributes attributes are sampled from $f_n^*$ in the tree construction.
\begin{equation} \label{eqTreeConstDist}
f_n(i,j,k) = \sum_{h=1}^n \frac{N_h}{L_n} 1( N_h = i , D_h = j , U_h = k ) , \quad f_n^*(i,j) = \frac{\sum_{h=1}^n U_h 1( N_h = i , D_h = j )}{\sum_{h=1}^n U_h} .
\end{equation}

The tree construction is given formally in Algorithm \ref{algTreeConstruction}. We denote by $\hat{G}_n = ( \hat{V}_n, \hat{E}_n )$ the resulting tree. Note the tree construction continues indefinitely, so the subscript $n$ does not refer to the number of nodes in the tree; rather, it refers to the length of the sequence $\{ N_h, D_h, U_h \}_{h=1}^n$ from which the distributions $f_n, f_n^*$ are defined. Finally, in Figure \ref{figTreeExample}, we show an example of the tree construction, which corresponds to the graph construction of Figures \ref{figGraphLabelExample} (i.e.\ the dashed orange subgraph has the same structure).

\begin{figure}
  \centering
  \includegraphics[height=2in]{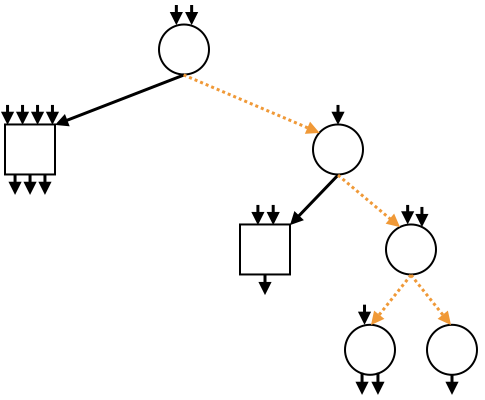}
  \caption{Branching process after three generations}
  \label{figTreeExample}
\end{figure}

\begin{algorithm}[t]
\DontPrintSemicolon
\caption{Tree Construction} \label{algTreeConstruction}

Draw root attributes $(N_{\phi}, D_{\phi}) \sim f_n^*$, set $U_{\phi} = 1$, set $\hat{A}_0 = \{ \phi \}$

\For{$m=1$ \KwTo $\infty$}{

Set $\hat{A}_m = \emptyset$

\For{$\i \in \hat{A}_{m-1}$}{

\If{$U_{\i} = 1$}{

\For{$j=1$ \KwTo $D_{\i}$}{

Add offspring $( \i ,j)$ to $\i$, let $(N_{( \i ,j)} , D_{( \i ,j)} , U_{( \i ,j)}) \sim f_n$, set $\hat{A}_m = \hat{A}_m \cup \{ ( \i ,j) \}$

}

}

}

}

\end{algorithm}

Having defined the tree construction, we now define the aforementioned quantity that follows the same distribution as the error bound in Lemma \ref{lemL1bound}. Specifically, we define $\hat{\mu}_{\phi}$ recursively as
\begin{equation}\label{eqMuHatRecursion}
\hat{\mu}_{\phi}(\phi) = 1, \quad \hat{\mu}_{\phi}( (\i,j) ) = \hat{\mu}_{\phi}(\i) \frac{(1-\alpha_n) U_{\i}}{D_{\i}}, (\i,j) \in \hat{A}_l, l > 0 ,
\end{equation}
where (by convention), $\i = \phi$ when $(\i,j) = i_1 \in \N$, i.e.\ when $(\i,j) \in \hat{A}_1$. Note that \eqref{eqMuHatRecursion} is the same as \eqref{eqMuMDefn} but computed on the tree $\hat{G}_n$; because there is a unique path from $\phi$ to $\i$ for each $\i \in \hat{V}_n$, this recursive definition is more convenient than the matrix definition.

We next state Lemma \ref{lemMuCoupling}, whose proof is deferred to Appendix \ref{secProofLemMuCoupling}. The proof essentially formalizes the intuition that we have presented: when $\tau_G > m$, the error bound from Lemma \ref{lemL1bound} is computed on a treelike subgraph and therefore follows the distribution of the analogous tree quantity.

\begin{lem} \label{lemMuCoupling}
For any $m \in \N$,
\begin{equation}
\mu_s^{(m)} ( V_n \setminus K_n ) | \{ \tau_G > m, U_s = 1 \} \stackrel{\mathcal{D}}{=} \sum_{j=0}^{m} \sum_{\i \in \hat{A}_j} U_{\i} \hat{\mu}_{\phi}(\i) ,
\end{equation}
where $\stackrel{\mathcal{D}}{=}$ denotes equality in distribution.
\end{lem}
\begin{proof}
See Appendix \ref{secProofLemMuCoupling}.
\end{proof}

We can now explain the remainder of our approach to proving the lemma. Using Lemmas  \ref{lemL1bound} and \ref{lemMuCoupling}, and noting that $\sum_{\i \in \hat{A}_0} U_{\i} \hat{\mu}_{\phi}(\i) = U_{\phi} \hat{\mu}_{\phi}(\phi) = 1$, we have
\begin{align}
\P \left[ B_s(\epsilon) \middle| U_s = 1 \right] & \leq \P \left[ \alpha_n \left( \mu_s^{(m-1)}(V_n \setminus K_n) - 1 \right)  + e_s^{\trans} (1-\alpha)^m \tilde{P}^m e_{V_n \setminus K_n} \geq \epsilon \middle| U_s = 1 \right] \\
& \leq \P [ \tau_G \leq m | U_s = 1  ] + \P \left[ \alpha_n \sum_{j=1}^{m-1} \sum_{\i \in \hat{A}_j} U_{\i} \hat{\mu}_{\phi}(\i) + \sum_{\i \in \hat{A}_m} U_{\i} \hat{\mu}_{\phi}(\i) \geq \epsilon  \right] . \label{eqTailBoundTwoSummands} 
\end{align}

Hence, our approach to bounding \eqref{eqApproxFails} will be to further bound the two summands in \eqref{eqTailBoundTwoSummands}. Since \eqref{eqTailBoundTwoSummands} holds for any $m \in \N$, our final step will be to choose $m$ to optimize the sum of these bounds. In particular, we will choose $m$ to balance the two bounds. This is because the first summand is increasing in $m$, while the second is decreasing in $m$.

\subsection{Coupling failure (Step 3)} \label{secCouplingFailure}

Our bound for the first summand in \eqref{eqTailBoundTwoSummands} is given in Lemma \ref{lemCouplingFailure}. This result is similar to Lemma 5.4 of \cite{chen2017generalized}, and we simply modify their techniques to prove the bound. However, Assumption \ref{assDegSeq} is different than the assumption required for the result in \cite{chen2017generalized}. This difference arises because the result in \cite{chen2017generalized} requires the entire $m$-step neighborhood to be treelike, while we only require the $m$-step neighborhood of label $D$ nodes to be treelike. This allows us to relax the assumption from \cite{chen2017generalized}, which requires $\frac{1}{n} \sum_{h=1}^n N_h^2$ to converge (we only require $\frac{1}{n} \sum_{h=1}^n N_h^2 U_h$ to converge). In fact, the example degree sequence presented after Theorem \ref{thmSublinear} in Section \ref{secMainResultPresent} satisfies
\begin{equation}
\frac{1}{n} \sum_{h=1}^n N_h^2 U_h = O(1) , \quad \frac{1}{n} \sum_{h=1}^n N_h^2 = O(\sqrt{n}) .
\end{equation}
Hence, there are degree sequences for which the existing lemma does not apply, but for which our version of the lemma does apply. This is why we do not directly use Lemma 5.4 from \cite{chen2017generalized}.

\begin{lem} \label{lemCouplingFailure}
Given Assumption \ref{assDegSeq}, for any $m_n \rightarrow \infty$ as $n \rightarrow \infty$ s.t.\ $m_n = O(n^{\gamma})$, we have
\begin{equation}
\P [ \tau_G \leq m_n | U_s = 1 ] = O \left( n^{-\delta} + \zeta^{m_n} / \sqrt{n} \right) ,
\end{equation}
where $\gamma, \delta, \zeta$ are defined in Assumption \ref{assDegSeq}.
\end{lem}
\begin{proof}
See Appendix \ref{secProofLemCouplingFailure}.
\end{proof}

\subsection{Tail bound on branching process quantity (Step 4)} \label{secBranchingTailBound}

Our final step is to bound the second summand in \eqref{eqTailBoundTwoSummands}. Our approach is to bound the probability that either $\alpha_n \sum_{j=1}^{m-1} \sum_{\i \in \hat{A}_j} U_{\i} \hat{\mu}_{\phi}(\i)$ or $\sum_{\i \in \hat{A}_m} U_{\i} \hat{\mu}_{\phi}(\i)$ exceeds $\epsilon/2$. For the first term, the recursive definition of $\hat{\mu}_{\phi}$ yields a martingale structure that allows us to use an approach similar to the method of bounded differences (see, for example, Section 5.4 of \cite{dubhashi2009concentration}). The second term arises from the tail of the $m$-step neighborhood approximation from Appendix \ref{secErrorMstep}, and it is not a sum of terms; hence, its expected value decays geometrically fast in $m$, so Markov's inequality suffices. 
\begin{lem} \label{lemTreeTailBound}
Given Assumption \ref{assDegSeq}, for any $\epsilon > 0$, any $m_n \rightarrow \infty$ as $n \rightarrow \infty$ s.t.\ $m_n = O(n^{\gamma})$, and any $\alpha_n \rightarrow 0$ as $n \rightarrow \infty$, we have
\begin{equation}
\P \left[ \alpha_n \sum_{j=1}^{m-1} \sum_{\i \in \hat{A}_j} U_{\i} \hat{\mu}_{\phi}(\i) + \sum_{\i \in \hat{A}_m} U_{\i} \hat{\mu}_{\phi}(\i) \geq \epsilon  \right] = O \left( n^{-\delta} + p^{m_n} + e^{ - ((1-p) \epsilon )^2 / ( 2 \alpha_n ) } \right) ,
\end{equation}
where $p, \delta$ are defined in Assumption \ref{assDegSeq}.
\end{lem}
\begin{proof}
See Appendix \ref{secProofLemTreeTailBound}.
\end{proof}

\subsection{Completing the proof of Lemma \ref{lemMainTailBound}} \label{secMainBound}

Finally, we can combine the results of this section to prove Lemma \ref{lemMainTailBound}. First, we substitute the results of Lemmas \ref{lemCouplingFailure} and \ref{lemTreeTailBound} into \eqref{eqTailBoundTwoSummands} to obtain (when Assumption \ref{assDegSeq} holds)
\begin{align}
\P \left[ B_s(\epsilon) \middle| U_s = 1 \right] = O \left( n^{-\delta} + \frac{\zeta^{m_n}}{\sqrt{n}} + p^{m_n} + e^{ - ((1-p) \epsilon )^2 / ( 2 \alpha_n ) }   \right) .
\end{align}
Next, we choose $m_n = \frac{ \log n }{ 2 \log ( \zeta / p ) }$ to equate the middle two terms, which gives
\begin{equation}
\frac{\zeta^{m_n}}{\sqrt{n}} = p^{m_n} = n^{ - \log (1/p) / ( 2 \log ( \zeta / p ) ) }  .
\end{equation}
For the third term, we let $\alpha_n = \rho \log(1/\tau) \log \zeta / \log n$ as in Claim \ref{clmChoiceOfAlpha} to obtain
\begin{equation}
\exp \left( - \frac{ ((1-p) \epsilon )^2 }{ 2 \alpha_n } \right) = \exp \left( - \frac{ ((1-p) \epsilon )^2 }{ 2 \rho \log(1/\tau) \log \zeta } \log n \right) = n^{ -  ((1-p) \epsilon )^2 / ( 2 \rho \log(1/\tau) \log \zeta ) } .
\end{equation}
Hence, we ultimately obtain 
\begin{equation}
\P \left[ B_s(\epsilon) \middle| U_s = 1 \right] = O ( n^{ -  c } ) , \quad c \triangleq \min \left\{ \delta ,  \frac{\log (1/p)}{ 2 \log ( \zeta / p ) } , \frac{((1-p) \epsilon )^2}{  2 \rho \log(1/\tau) \log \zeta } \right\} > 0 ,
\end{equation}
as claimed in Lemma \ref{lemMainTailBound}.

%% file: proofLinearCombo.tex
\section{Lemma \ref{lemMainTailBound} proof details} \label{appMainProofs}

\subsection{Proof of Lemma \ref{lemLinearCombo}} \label{secProofLemLinearCombo}

The lemma relates the stationary distributions of several Markov chains: those with transition matrices $P_s$, $\tilde{P}_s$, and $P_k, k \in K_n$. We will denote these Markov chains by $\{ X_i^s \}_{i=0}^{\infty}$, $\{ \tilde{X}_i^s \}_{i=0}^{\infty}$, and $\{ X_i^k \}_{i=0}^{\infty}, k \in K_n$, respectively, in this proof. Our basic approach will be to relate the stationary distributions indirectly via a renewal-reward interpretation of PPR. Hence, we begin by defining this interpretation in Appendix \ref{appRenewalReward}. We then prove the lemma in Appendix \ref{appLinearComboProof}. Recall from the main text that $\P_{G_n}[\cdot]$ and $\E_{G_n}[\cdot]$ denote probability and expectation with the DCM fixed (as in the statement of the lemma).

\subsubsection{Renewal-reward interpretation of PPR} \label{appRenewalReward}

From the dynamics of $\{ X_i^s \}_{i=0}^{\infty}$ described in Section \ref{secPprDefinition}, we can view the time instances of jumps to $s$ as forming a Bernoulli process with parameter $\alpha_n$, independent of the random walk. Furthermore, for each $v \in V_n$, we can define a reward function $1 ( X_i^s = v )$. Then, letting $L_s$ denote the time of the first jump to $s$, we define
\begin{equation}\label{eqAccRewDefn}
\tau_s(v) = \sum_{i=0}^{L_s-1} 1 ( X_i^s = v ) ,
\end{equation}
which, when $X_0^s = s$, gives the accumulated reward during the first inter-renewal interval. From the renewal-reward theorem (see, for example, Section 5.4 of \cite{gallager2013stochastic}), it follows that
\begin{equation}\label{eqRenRewThmApplication}
\lim_{t \rightarrow \infty} \frac{1}{t} \sum_{i=0}^{t-1} 1 ( X_i^s = v ) = \alpha_n \E_{G_n} [ \tau_s(v) | X_0^s = s ] ,
\end{equation}
where we have also used the fact that $L_s \sim \textrm{geometric}(\alpha_n)$. On the other hand, assuming $P_s$ is irreducible (which we will return to argue is without loss of generality), we have
\begin{equation}\label{eqCesaroAvg}
\pi_s(v) = \lim_{t \rightarrow \infty} \frac{1}{t} \sum_{i=0}^{t-1} 1 ( X_i^s = v ) .
\end{equation}
Hence, combining \eqref{eqRenRewThmApplication} and \eqref{eqCesaroAvg} yields
\begin{equation}\label{eqRenRewInterpretation}
\pi_s(v) = \alpha_n \E_{G_n} [ \tau_s(v) | X_0^s = s ]\ \forall\ v \in V_n .
\end{equation}
Similarly, for $k \in K_n$, we have $\pi_k(v) = \alpha_n \E_{G_n}[ \tau_k(v) | X_0^k = k ]$, where $\tau_k(v)$ is defined as in \eqref{eqAccRewDefn}.

For the chain $\{ \tilde{X}_i^s \}_{i=0}^{\infty}$, we have a similar (though slightly more subtle) renewal-reward interpretation. In particular, recall the dynamics of this chain are as follows: from $v \in V_n \setminus K_n$, follow the random walk with probability $1-\alpha_n$ and jump to $s$ with probability $\alpha_n$; from $k \in K_n$, jump to $s$ with probability 1. Hence, while the time instances of jumps to $s$ do not form a Bernoulli process on this chain, they still form a renewal process: inter-renewal intervals are independent (due to the Markov property) and identically-distributed (due to the time invariance of the Markov chain). Additionally, assuming $\tilde{X}_0^s = s$, the first renewal occurs at $\min \{ \tilde{L}_s, \tilde{H} + 1 \}$, where $\tilde{L}_s \sim \textrm{geometric}(\alpha_n)$ and $\tilde{H} = \inf \{ i \in \Z_+ : \tilde{X}_i^s \in K_n \}$ is the hitting time of $K_n$. It follows that
\begin{equation}
\tilde{\pi}_s(v) =\frac{ \E_{G_n} [ \tilde{\tau}_s(v) | \tilde{X}_0^s = s ] }{ \E_{G_n} [ \min \{ \tilde{L}_s, \tilde{H}+1 \} | \tilde{X}_0^s = s ] }\ \forall\ v \in V_n ,
\end{equation}
where $\tilde{\tau}_s(v) = \sum_{i=0}^{ \min \{ \tilde{L}_s - 1 , \tilde{H} \}  } 1 ( \tilde{X}_i^s = v )$.

Before proceeding, we argue the assumption of irreducibility is without loss of generality for the Markov chains at hand. Consider, for example, $\{ X_i^s \}_{i=0}^{\infty}$. If this chain is not irreducible, we can define $V_{n,s} \subset V_n$ as the states for which a path of positive probability from $s$ to $v$ exists. Then the Markov chain restricted to states $V_{n,s}$ is irreducible: for any $v, w \in V_{n,s}$, we can jump from $v$ to $s$ and then reach $w$ from $s$. We can then compute the stationary distribution $\{ \pi_s(v) \}_{v \in V_{n,s} }$ for this irreducible chain and set $\pi_s(v) = 0\ \forall\ v \in V_n \setminus V_{n,s}$ (intuitively, $v$ is unimportant to $s$ if $s$ cannot reach $v$, so its PPR should be zero). Note this is consistent with the derivation above. In particular, \eqref{eqRenRewThmApplication} and \eqref{eqCesaroAvg} hold for the chain restricted to states $V_{n,s}$, so \eqref{eqRenRewInterpretation} holds for $v \in V_{n,s}$; on the other hand, both sides of \eqref{eqRenRewInterpretation} are zero for $v \notin V_{n,s}$. 

\subsubsection{Proof of the lemma} \label{appLinearComboProof}

Equipped with this renewal-reward interpretation, we will relate $\pi_s$, $\tilde{\pi}_s$, and $\pi_k, k \in K_n$ by relating $\E_{G_n} [ \tau_s(v) | X_0^s = s ]$, $\E_{G_n} [ \tilde{\tau}_s(v) | \tilde{X}_0^s = s ]$, and $\E_{G_n} [ \tau_k(v) | X_0^k = k ], k \in K_n$. For this, we define $H = \inf \{ i \in \Z_+ : X_i^s \in K_n \}$, the quantity analogous to $\tilde{H}$ instead defined on $\{ X_i^s \}_{i=0}^{\infty}$.

Because the dynamics of $\{ X_i^s \}_{i=0}^{\infty}$ and $\{ \tilde{X}_i^s \}_{i=0}^{\infty}$ only differ when $K_n$ is reached, we can immediately obtain several relationship between the quantities computed on these chains. In particular, if $K_n$ is \textit{not} reached before the first renewal (i.e.\ if $L_s \leq H$, $\tilde{L}_s \leq \tilde{H}$), the chains have identical dynamics. Therefore, we can write
\begin{equation}
\E_{G_n} [ \tau_s(v) | L_s \leq H , X_0^s = s ] = \E_{G_n} [ \tilde{\tau}_s(v) | \tilde{L}_s \leq \tilde{H} , \tilde{X}_0^s = s ]\ \forall\ v \in V_n .
\end{equation}
Furthermore, $\tilde{\tau}_s(v) = 0$ when $v \in K_n$ and $\tilde{L}_s \leq \tilde{H}$ (i.e.\ when $K_n$ is not reached before the first renewal), so we may rewrite this as 
\begin{equation}\label{eqTauTauTildeNoVisitK}
\E_{G_n} [ \tau_s(v) | L_s \leq H , X_0^s = s ]  = U_v \E_{G_n} [ \tilde{\tau}_s(v) | \tilde{L}_s \leq \tilde{H} , \tilde{X}_0^s = s ]\ \forall\ v \in V_n .
\end{equation}

By a similar argument, if $K_n$ \textit{is} reached before the first renewal ($L_s > H$, $\tilde{L}_s > \tilde{H}$), the dynamics of the chains differ after $H, \tilde{H}$, but they remain the same up to and including $H, \tilde{H}$. Hence,
\begin{equation} \label{eqXHandLgtrHprob}
\P_{G_n} [ X_H^s = k , L_s > H | X_0^s = s ] = \P_{G_n} [ \tilde{X}_{\tilde{H}}^s = k , \tilde{L}_s > \tilde{H} | \tilde{X}_0^s = s ]\ \forall\ k \in K_n ,
\end{equation}
which also implies
\begin{equation}  \label{eqLleqH}
\P_{G_n} [ L_s \leq H | X_0^s = s ] = \P_{G_n} [\tilde{L}_s \leq \tilde{H} | \tilde{X}_0^s = s ] .
\end{equation}
We can obtain another expression for the right side of \eqref{eqXHandLgtrHprob}. In particular, since jumps from $k$ to $s$ occur with probability 1 on the $\{ \tilde{X}_i^s \}_{i = 0}^{\infty}$ chain, $k$ is visited at most one time before the first renewal, i.e.\ $\tilde{\tau}_s(k) \in \{0,1\}$. Furthermore, $\tilde{\tau}_s(k) = 1$ if and only if  $\tilde{L}_s > \tilde{H}$ and $\tilde{X}_{\tilde{H}}^s = k$. Hence,
\begin{equation}\label{eqXHandLgtrHexp}
\P_{G_n} [ X_H^s = k , L_s > H | X_0^s = s ] = \E_{G_n} [ \tilde{\tau}_s(k) | \tilde{X}_0^s = s ]\ \forall\ k \in K_n .
\end{equation}
If instead $K_n$ is reached, the dynamics of $\{ X_i^s \}_{i=0}^{\infty}$ and $\{ \tilde{X}_i^s \}_{i=0}^{\infty}$ differ. In this case, we claim
\begin{align}\label{eqVisitsBeforeAndAfterK}
\E_{G_n} [ \tau_s(v) | X_H^s = k,  L_s > H , X_0^s = s ] & = U_v \E_{G_n} [\tilde{\tau}_s(v) | \tilde{X}_{\tilde{H}}^s = k,  \tilde{L}_s > \tilde{H} , \tilde{X}_0^s = s ] \\
& \quad + \E_{G_n} [ \tau_k(v) | X_0^k = k ]   ,
\end{align}
which we will return to prove shortly. (In essence, \eqref{eqVisitsBeforeAndAfterK} counts the number visits to $v$ before and after reaching $k$ using the $\{ \tilde{X}_i^s \}_{i=0}^{\infty}$ and $\{ X_i^k \}_{i=1}^k$ chains, respectively.) 

Using \eqref{eqTauTauTildeNoVisitK}, \eqref{eqXHandLgtrHprob}, \eqref{eqLleqH}, \eqref{eqXHandLgtrHexp}, and \eqref{eqVisitsBeforeAndAfterK}, and the law of total expectation, then gives
\begin{align}
\E_{G_n} [ \tau_s(v) | X_0^s = s ]  
& = U_v \E_{G_n} [ \tilde{\tau}_s(v) | \tilde{X}_0^s = s ] + \sum_{k \in K_n} \E_{G_n} [ \tilde{\tau}_s(k) | \tilde{X}_0^s = s ] \E_{G_n} [ \tau_k(v) | X_0^k = k ] .
\end{align}
We then use the renewal-reward interpretation from Appendix \ref{appRenewalReward} to translate this equation back to stationary distributions. Specifically, multiplying by $\alpha_n$ on both sides, and multiplying and dividing by $\E_{G_n} [ \min \{ \tilde{L}_s, \tilde{H} + 1 \} | \tilde{X}_0^s = s ]$ on the right side, gives
\begin{equation}\label{eqLinearComboBeforeNormalization}
\pi_s(v) = \E_{G_n} [ \min \{ \tilde{L}_s, \tilde{H}  + 1  \} | \tilde{X}_0^s = s ] \left( \alpha_n U_v \tilde{\pi}_s(v) + \sum_{k \in K_n} \tilde{\pi}_s(k) \pi_k(v) \right) .
\end{equation}
Then, summing over $v \in V_n$ (assuming stationary distributions are normalized to sum to 1),
\begin{align}\label{eqLinearComboNormalization}
& 1 = \E_{G_n} [ \min \{ \tilde{L}_s, \tilde{H} + 1  \} | \tilde{X}_0^s = s ] \left( \alpha_n \tilde{\pi}_s ( V_n \setminus K_n ) + \tilde{\pi}_s ( K_n ) \right) \\
& \Rightarrow \E_{G_n} [ \min \{ \tilde{L}_s, \tilde{H} + 1  \} | \tilde{X}_0^s = s ] = \frac{1}{ \alpha_n + ( 1- \alpha_n ) \tilde{\pi}_s ( K_n ) } .
\end{align}
Finally, combining \eqref{eqLinearComboBeforeNormalization} and \eqref{eqLinearComboNormalization} completes the proof.

We now return to prove \eqref{eqVisitsBeforeAndAfterK}. For this, we first have by definition of $\tau_s(v)$,
\begin{align}\label{eqLastTermTwoSummands}
\E_{G_n} [ \tau_s(v) | X_H^s = k,  L_s > H , X_0^s = s ] & = \E_{G_n} \left[ \sum_{i=0}^{H-1} 1 ( X_i^s = v  ) \middle| X_H^s = k,  L_s > H , X_0^s = s \right] \\
& \quad + \E_{G_n} \left[ \sum_{i=H}^{L_s-1} 1 ( X_i^s = v  ) \middle| X_H^s = k,  L_s > H , X_0^s = s \right] .
\end{align}

Now consider the first summand in \eqref{eqLastTermTwoSummands}. By the preceding arguments, $\{ {X}_i^s \}_{i=0}^{\infty}$  and $\{ \tilde{X}_i^s \}_{i=0}^{\infty}$ have the same dynamics before $H, \tilde{H}$, so
\begin{equation}
\E_{G_n} \left[ \sum_{i=0}^{H-1} 1 ( X_i^s = v  ) \middle| X_H^s = k,  L_s > H , X_0^s = s \right] = \E_{G_n} \left[ \sum_{i=0}^{\tilde{H}-1} 1 ( \tilde{X}_i^s = v  ) \middle| \tilde{X}_{\tilde{H}}^s = k,  \tilde{L}_s > \tilde{H} , \tilde{X}_0^s = s \right] .
\end{equation}
For $v \in V_n \setminus K_n$ (i.e.\ $U_v = 1$), we can write
\begin{align}
& \E_{G_n} \left[ \sum_{i=0}^{\tilde{H}-1} 1 ( \tilde{X}_i^s = v  ) \middle| \tilde{X}_{\tilde{H}}^s = k,  \tilde{L}_s > \tilde{H} , \tilde{X}_0^s = s \right] \\
& \quad = \E_{G_n} \left[ \sum_{i=0}^{ \tilde{H} } 1 ( \tilde{X}_i^s = v  )  \middle| \tilde{X}_{\tilde{H}}^s = k,  \tilde{L}_s > \tilde{H} , \tilde{X}_0^s = s \right] \\
& \quad = \E_{G_n} \left[ \sum_{i=0}^{ \min \{ \tilde{L}_s - 1 , \tilde{H} \} } 1 ( \tilde{X}_i^s = v  ) \middle| \tilde{X}_{\tilde{H}}^s = k,  \tilde{L}_s > \tilde{H} , \tilde{X}_0^s = s \right] \\
& \quad = \E_{G_n} \left[\tilde{\tau}_s(v) \middle| \tilde{X}_{\tilde{H}}^s = k,  \tilde{L}_s > \tilde{H} , \tilde{X}_0^s = s \right] ,
\end{align}
where the first equality holds since $v \in V_n \setminus K_n$ and by conditioning on $\{ \tilde{X}_{\tilde{H}}^s = k \}$ ($k \in K_n$), the second holds by conditioning on $\{ \tilde{L}_s > \tilde{H}  \}$, and the third holds by definition of $\tilde{\tau}_s(v)$. Note that if $v \in K_n$ (i.e.\ $U_v = 0$), we simply have 
\begin{equation}
\E_{G_n} \left[ \sum_{i=0}^{\tilde{H}-1} 1 ( \tilde{X}_i^s = v  ) \middle| \tilde{X}_{\tilde{H}}^s = k,  \tilde{L}_s > \tilde{H} , \tilde{X}_0^s = s \right] = 0 ,
\end{equation}
which holds by definition of $\tilde{H}$. To summarize, we have shown
\begin{equation}\label{eqFirstSummandCompleteLC}
\E_{G_n} \left[ \sum_{i=0}^{H-1} 1 ( X_i^s = v  ) \middle| X_H^s = k,  L_s > H , X_0^s = s \right] = U_v \E_{G_n} \left[\tilde{\tau}_s(v) \middle| \tilde{X}_{\tilde{H}}^s = k,  \tilde{L}_s > \tilde{H} , \tilde{X}_0^s = s \right] .
\end{equation}

Next, consider the second summand in \eqref{eqLastTermTwoSummands}. We rewrite this term as
\begin{equation}\label{eqRewriteSecondSummand}
\frac{ \E_{G_n} [ \sum_{i=H}^{L_s-1} 1 ( X_i^s = v , X_H^s = k, X_0^s = s ) 1 (  L_s > H  ) ] }{ \P_{G_n} [ X_H^s = k,  L_s > H , X_0^s = s ] } ,
\end{equation}
and we focus on the numerator. First, we note $1 (  L_s > H  )  = \sum_{l > h} 1 ( L_s = l, H = h)$, where the sum is taken over $\{ (l,h) \in \Z_+ \times \Z_+ : l > h \}$. Substituting and using linearity gives
\begin{align}
& \sum_{l > h} \E_{G_n} \left[ \sum_{i=H}^{L_s-1} 1 ( X_i^s = v  , X_H^s = k, X_0^s = s ) 1 ( L_s = l , H = h ) \right] \\
& \quad = \sum_{l > h} \E_{G_n} \left[ \sum_{i=h}^{l-1} 1 ( X_i^s = v  , X_h^s = k, X_0^s = s ) 1 ( L_s = l , H = h ) \right] \\
& \quad = \sum_{l > h} \sum_{i=h}^{l-1} \P_{G_n} \left[ X_i^s = v  , X_h^s = k, X_0^s = s , L_s = l , H = h \right]  \\
& \quad = \sum_{l > h} \sum_{i=h}^{l-1} \P_{G_n} \left[ X_i^s = v , X_0^s = s , L_s = l , H = h \middle| X_h^s = k \right] \P_{G_n} \left[ X_h^s = k \right] . \label{eqPreMarkovProperty}
\end{align}

We next aim to apply the Markov property to the conditional probability in the last equation. For this, we write $\{ L_s = l \} = A_{s,l} \cap ( \cap_{j=0}^{l-1} A_{s,j}^C )$, where $A_{s,j}$ denotes the event that a jump to $s$ occurs at step $j$ of the random walk. We then have
\begin{equation}
\{ X_i^s = v , X_0^s = s , L_s = l , H = h \} = \left\{ X_i^s = v , A_{s,l} , \cap_{j=h+1}^{l-1} A_{s,j}^C \right\} \cap \left\{ H = h , \cap_{j =0}^{h} A_{s,j}^C , X_0^s = s \right\}
\end{equation}
where on the right side, the first event is the future and the second event is the past, when $h$ is viewed as the present. Hence, the Markov property implies
\begin{align}
& \P_{G_n} \left[ X_i^s = v , X_0^s = s , L_s = l , H = h \middle| X_h^s = k \right] \\
& \quad = \P_{G_n} \left[ X_i^s = v , A_{s,l} , \cap_{j=h}^{l-1} A_{s,j}^C \middle| X_h^s = k \right] \P_{G_n} \left[ H = h , \cap_{j =0}^{h-1} A_{s,j}^C , X_0^s = s \middle| X_h^s = k \right] . \label{eqPostMarkovProperty}
\end{align}
Furthermore, by the time invariance of the Markov chain,
\begin{align}
\P_{G_n} \left[ X_i^s = v , A_{s,l} , \cap_{j=h}^{l-1} A_{s,j}^C \middle| X_h^s = k \right] & = \P_{G_n} \left[ X_{i-h}^s = v , A_{s,l-h} , \cap_{j=0}^{l-h-1} A_{s,j}^C \middle| X_0^s = k \right] \\
& = \P_{G_n} \left[ X_{i-h}^s = v , L_s = l-h \middle| X_0^s = k \right] . \label{eqTimeInvariance}
\end{align}
Finally, by definition of $A_{s,j}$, we have
\begin{equation}\label{eqDefinitionOfAsj}
\P_{G_n} \left[ H = h , \cap_{j =0}^{h-1} A_{s,j}^C , X_0^s = s \middle| X_h^s = k \right] = \P_{G_n} \left[ H = h , L_s > h , X_0^s = s \middle| X_h^s = k \right] .
\end{equation}

Combining \eqref{eqPreMarkovProperty}, \eqref{eqPostMarkovProperty}, \eqref{eqTimeInvariance}, and \eqref{eqDefinitionOfAsj} then yields
\begin{align}
& \sum_{l > h} \E_{G_n} \left[ \sum_{i=H}^{L_s-1} 1 ( X_i^s = v  , X_H^s = k, X_0^s = s ) 1 ( L_s = l , H = h ) \right] \\
& \quad = \sum_{l > h} \sum_{i=h}^{l-1} \P_{G_n} \left[ X_{i-h}^s = v , L_s = l-h \middle| X_0^s = k \right] \P_{G_n} \left[ H = h , L_s > h , X_0^s = s  , X_h^s = k \right] \\
& \quad = \sum_{h \in \Z_+} \P_{G_n} \left[ H = h , L_s > h , X_0^s = s  , X_h^s = k \right]  \sum_{l = h+1}^{\infty} \sum_{i=0}^{l-h-1} \P_{G_n} \left[ X_{i}^s = v , L_s = l-h \middle| X_0^s = k \right] , \label{eqNumeratorAlmostDone}
\end{align}
where in the second equality we have simply rearranged terms and rewritten summation indices. For the inner double summation, we can obtain
\begin{align}
& \sum_{l = h+1}^{\infty} \sum_{i=0}^{l-h-1} \P_{G_n} \left[ X_{i}^s = v , L_s = l-h \middle| X_0^s = k \right] \\
& \quad =  \sum_{l = h+1}^{\infty} \E_{G_n} \left[  \sum_{i=0}^{l-h-1}  1 ( X_i^s = v ) 1 ( L_s = l-h , X_0^s = k ) \right] \frac{1}{ \P_{G_n} [ X_0^s = k ] } \\
& \quad =  \sum_{l = h+1}^{\infty} \E_{G_n} \left[  \sum_{i=0}^{L_s-1}  1 ( X_i^s = v ) 1 ( L_s = l-h , X_0^s = k ) \right] \frac{1}{ \P_{G_n} [ X_0^s = k ] } \\
& \quad =  \E_{G_n} \left[  \sum_{i=0}^{L_s-1}  1 ( X_i^s = v ) 1 ( X_0^s = k ) \sum_{l = h+1}^{\infty} 1 ( L_s = l-h )  \right] \frac{1}{ \P_{G_n} [ X_0^s = k ] } \\
& \quad =  \E_{G_n} \left[  \sum_{i=0}^{L_s-1}  1 ( X_i^s = v ) \middle| X_0^s = k \right] = \E_{G_n} [ \tau_s(v) | X_0^s = k ] = \E_{G_n} [ \tau_k(v) | X_0^k = k ]  ,
\end{align}
where the first three steps are straightforward, the fourth step uses the fact that $L_s$ is integer-valued and \textit{a.s.}\ finite, and the fifth step follows by definition. The final inequality follows because $\tau_s(v)$ and $\tau_k(v)$ count the number of visits to $v$ on the $\{ X_i^s \}_{i=0}^{\infty}$ and $\{ X_i^k \}_{i=0}^{\infty}$ chains before jumps occur, and before jumps occur, these chains have the same dynamics (since they only differ in jump locations, $s$ versus $k$). Substituting into \eqref{eqNumeratorAlmostDone} gives
\begin{align}
& \sum_{l > h} \E_{G_n} \left[ \sum_{i=H}^{L_s-1} 1 ( X_i^s = v  , X_H^s = k, X_0^s = s ) 1 ( L_s = l , H = h ) \right] \\
& \quad =\E_{G_n} [ \tau_k(v) | X_0^k = k ]  \sum_{h \in \Z_+} \P_{G_n} \left[ H = h , L_s > h , X_0^s = s  , X_h^s = k \right] \\
& \quad = \E_{G_n} [ \tau_k(v) | X_0^k = k ]  \P_{G_n} \left[ L_s > H , X_0^s = s  , X_H^s = k \right] . \label{eqNumeratorDone}
\end{align}
Hence, combining \eqref{eqRewriteSecondSummand} and \eqref{eqNumeratorDone} yields
\begin{equation}\label{eqSecondSummandCompleteLC}
\E_{G_n} \left[ \sum_{i=H}^{L_s-1} 1 ( X_i^s = v  ) \middle| X_H^s = k,  L_s > H , X_0^s = s \right] = \E_{G_n} [ \tau_k(v) | X_0^k = k ]  .
\end{equation}
Finally, we substitute \eqref{eqFirstSummandCompleteLC} and \eqref{eqSecondSummandCompleteLC} into \eqref{eqLastTermTwoSummands} to complete the proof of \eqref{eqVisitsBeforeAndAfterK}.

%% file: proofMstepNeighborhood.tex
\subsection{Proof of Lemma \ref{lemL1bound}} \label{secProofLemL1bound}

We aim to bound $\| \pi_s - ( \alpha_n e_s^{\trans} + \sum_{k \in K_n} \beta_s(k) \pi_k ) \|_1$, where for each $k \in K_n$,
\begin{equation}
\beta_{s}(k) = \frac{\tilde{\pi}_s(k)}{\alpha_n + (1-\alpha_n) \tilde{\pi}_s(K_n) } .
\end{equation}
Using Lemma \ref{lemLinearCombo}, we can write
\begin{equation}\label{eqLinearComboIntoL1}
\left\| \pi_s - \left( \alpha_n e_s^{\trans} + \sum_{k \in K_n} \beta_s(k) \pi_k \right) \right\|_1 = \sum_{v \in V_n} \left| \frac{\alpha_n U_v \tilde{\pi}_s(v)  }{\alpha_n + (1-\alpha_n) \tilde{\pi}_s(K_n)} - \alpha_n 1(v = s) \right| .
\end{equation}
We next claim that each summand in \eqref{eqLinearComboIntoL1} is non-negative. This is obvious for $v \neq s$. For $v = s$, since $U_s = 1$ by assumption, it suffices to show
\begin{equation}\label{eqDropAbsValS}
\tilde{\pi}_s(s) \geq \alpha_n + (1-\alpha_n) \tilde{\pi}_s(K_n) .
\end{equation}
To this end, first note that since $\tilde{\pi}_s = \tilde{\pi}_s \tilde{P}_s$ and $\tilde{\pi}_s 1_n = 1$, we can write
\begin{equation}
\tilde{\pi}_s = (1-\alpha_n) \tilde{\pi}_s \left( \tilde{P} + e_{K_n} e_s^{\trans} \right) + \alpha_n e_s^{\trans} ,
\end{equation}
which implies
\begin{equation} \label{eqPiKsNeumann}
\tilde{\pi}_s = \alpha_n e_s^{\trans} \left( I - (1-\alpha_n) \left( \tilde{P} + e_{K_n} e_s^{\trans} \right) \right)^{-1} = \alpha_n e_s^{\trans} \sum_{i=0}^{\infty} (1-\alpha_n)^i \left( \tilde{P} + e_{K_n} e_s^{\trans} \right)^i .
\end{equation}
Using \eqref{eqPiKsNeumann}, we have
\begin{align}
\tilde{\pi}_s(s) & = \alpha_n e_s^{\trans} \sum_{i=0}^{\infty} (1-\alpha_n)^i \left( \tilde{P} + e_{K_n} e_s^{\trans} \right)^i e_s = \alpha_n + \alpha_n e_s^{\trans} \sum_{i=1}^{\infty} (1-\alpha_n)^i \left( \tilde{P} + e_{K_n} e_s^{\trans} \right)^i e_s \\
& = \alpha_n + \alpha_n ( 1-\alpha_n ) e_s^{\trans} \sum_{i=0}^{\infty} (1-\alpha_n)^i \left( \tilde{P} + e_{K_n} e_s^{\trans} \right)^i \left( \tilde{P} + e_{K_n} e_s^{\trans} \right) e_s \\
& \geq  \alpha_n + \alpha_n ( 1-\alpha_n ) e_s^{\trans} \sum_{i=0}^{\infty} (1-\alpha_n)^i \left( \tilde{P} + e_{K_n} e_s^{\trans} \right)^i  e_{K_n} e_s^{\trans} e_s  = \alpha_n + (1-\alpha_n) \tilde{\pi}_s e_{K_n} ,
\end{align}
where for the inequality we simply dropped a nonnegative term. This establishes \eqref{eqDropAbsValS}, since $\tilde{\pi}_s e_{K_n} = \tilde{\pi}_s(K_n)$. Hence, the expression in \eqref{eqLinearComboIntoL1} simplifies to
\begin{equation}\label{eqLinearComboIntoL1NoAbsVal}
\left\| \pi_s - \left( \alpha_n e_s^{\trans} + \sum_{k \in K_n} \beta_s(k) \pi_k \right) \right\|_1 = \alpha_n \left( \frac{\tilde{\pi}_s(V_n \setminus K_n)}{\alpha_n + (1-\alpha_n) \tilde{\pi}_s(K_n)} - 1 \right) .
\end{equation}
We next bound the right side of \eqref{eqLinearComboIntoL1NoAbsVal} in terms of $\mu_s^{(m)}$, as in the statement of the lemma. We begin by establishing a relationship between $\tilde{\pi}_s$ and $\mu_s$, where
\begin{equation}\label{eqMuDefn}
\mu_s = \lim_{m \rightarrow \infty} \mu_s^{(m)} = e_s^{\trans} \sum_{i=0}^{\infty} (1-\alpha_n)^i \tilde{P}^i =e_s^{\trans}  \left( I - (1-\alpha_n) \tilde{P} \right)^{-1} .
\end{equation} 
To this end, consider the matrix inversion in \eqref{eqPiKsNeumann}. By the Sherman-Morrison-Woodbury formula (see, for example, Section 6.4 of \cite{laub2005matrix}),
\begin{align}
& \left( I - (1-\alpha_n) \left( \tilde{P} + e_{K_n} e_s^{\trans} \right) \right)^{-1} = \left( \left( I - (1-\alpha_n) \tilde{P} \right) - (1-\alpha_n) e_{K_n} e_s^{\trans} \right)^{-1} \\
& \quad = \left( I - (1-\alpha_n) \tilde{P} \right)^{-1} + \frac{ \left( I - (1-\alpha_n) \tilde{P} \right)^{-1} (1-\alpha_n) e_{K_n} e_s^{\trans} \left( I - (1-\alpha_n) \tilde{P} \right)^{-1} }{ 1 -  e_s^{\trans} \left( I - (1-\alpha_n) \tilde{P} \right)^{-1} (1-\alpha_n) e_{K_n} } \label{eqShermanApplication} .
\end{align}
It follows that, for each $v \in V_n$,
\begin{align}
\tilde{\pi}_s(v) & = \alpha_n e_s^{\trans} \left( I - (1-\alpha_n) \left( \tilde{P} + e_{K_n} e_s^{\trans} \right) \right)^{-1} e_v \\
&  = \alpha_n e_s^{\trans} \left(  \left( I - (1-\alpha_n) \tilde{P} \right)^{-1} + \frac{ \left( I - (1-\alpha_n) \tilde{P} \right)^{-1} (1-\alpha_n) e_{K_n} e_s^{\trans} \left( I - (1-\alpha_n) \tilde{P} \right)^{-1} }{ 1 -  e_s^{\trans} \left( I - (1-\alpha_n) \tilde{P} \right)^{-1} (1-\alpha_n) e_{K_n} } \right) e_v \\
& = \alpha_n  \mu_s(v) \left( 1 + \frac{(1-\alpha_n) \mu_s(K_n)}{1-(1-\alpha_n) \mu_s(K_n) } \right)  = \frac{ \alpha_n \mu_s(v) }{ 1 - (1-\alpha_n) \mu_s(K_n) } \label{eqPiMuRelationship} ,
\end{align}
where the first three equalities follow from \eqref{eqPiKsNeumann}, \eqref{eqShermanApplication}, and \eqref{eqMuDefn}, respectively, and the fourth involves simple manipulations. We can then combine \eqref{eqLinearComboIntoL1NoAbsVal} and \eqref{eqPiMuRelationship} to obtain
\begin{equation}
\left\| \pi_s - \left( \alpha_n e_s^{\trans} + \sum_{k \in K} \beta_s(k) \pi_k \right) \right\|_1 =  \alpha_n \left( \mu_s(V_n \setminus K_n) - 1 \right) . \label{eqL1BoundNoM}
\end{equation}

Next, we observe
\begin{align}
\mu_s(V_n \setminus K_n) & = \mu_s^{(m)}(V_n \setminus K_n) + e_s^{\trans} \sum_{i=m+1}^{\infty} (1-\alpha_n)^i \tilde{P}^i e_{V_n \setminus K_n} \\
& = \mu_s^{(m)}(V_n \setminus K_n) + e_s^{\trans} (1-\alpha_n)^m \tilde{P}^m \sum_{i=1}^{\infty} (1-\alpha_n)^i \tilde{P}^i e_{V_n \setminus K_n} \\
& = \mu_s^{(m)}(V_n \setminus K_n) + \left( \mu_s^{(m)} - \mu_s^{(m-1)}   \right) \sum_{i=1}^{\infty} (1-\alpha_n)^i \tilde{P}^i e_{V_n \setminus K_n} , \label{eqMuStoMuSmBound}
\end{align}
where we have used \eqref{eqMuMDefn} and \eqref{eqMuDefn}. We next claim $\tilde{P} e_{V_n \setminus K_n} \leq e_{V_n \setminus K_n}$, where the inequality is taken componentwise. To see this, let $( \tilde{P} e_{V_n \setminus K_n} )(i)$ denote the $i$-th component of $\tilde{P} e_{V_n \setminus K_n}$. Then
\begin{equation}\label{eqPeVminKcomp}
( \tilde{P} e_{V_n \setminus K_n} )(i)= \sum_{j=1}^n \tilde{P}(i,j) e_{V_n \setminus K_n}(j) = U_i  \sum_{j=1}^n P(i,j) e_{V_n \setminus K_n}(j) \leq U_i \sum_{j=1}^n P(i,j) = U_i = e_{V_n \setminus K_n}(i) ,
\end{equation}
where the second equality uses the definition of $\tilde{P}$, the third equality holds because $P$ is row stochastic, and the remaining steps are straightforward. It follows that
\begin{equation}\label{eqComponentwiseIterate}
\sum_{i=1}^{\infty} (1-\alpha_n)^i \tilde{P}^i e_{V_n \setminus K_n} \leq \left( \sum_{i=1}^{\infty} (1-\alpha_n)^i \right) e_{V_n \setminus K_n} = \left( \frac{1-\alpha_n}{\alpha_n}  \right) e_{V_n \setminus K_n} ,
\end{equation}
where the inequality is again componentwise. Combining \eqref{eqMuStoMuSmBound} and \eqref{eqComponentwiseIterate} gives
\begin{align}
\mu_s(V_n \setminus K_n) & \leq \mu_s^{(m)}(V_n \setminus K_n) + \left( \mu_s^{(m)} - \mu_s^{(m-1)}  \right) \left( \frac{1-\alpha_n}{\alpha_n}  \right) e_{V_n \setminus K_n} \\
& = \frac{1}{ \alpha_n } e_s^{\trans} (1-\alpha_n)^m \tilde{P}^m e_{V_n \setminus K_n} + \mu_s^{(m-1)}(V_n \setminus K_n), \label{eqMuStoMuSmFinal}
\end{align}
where we have also used $\mu_s^{(m)} \geq \mu_s^{(m-1)}$ (componentwise). Finally, \eqref{eqL1BoundNoM} and \eqref{eqMuStoMuSmFinal} imply
\begin{equation}
\left\| \pi_s - \left( \alpha_n e_s^{\trans} + \sum_{k \in K} \beta_s(k) \pi_k \right) \right\|_1 \leq \alpha_n \left( \mu_s^{(m-1)}(V_n \setminus K_n) - 1 \right) + e_s^{\trans} (1-\alpha_n)^m \tilde{P}^m e_{V_n \setminus K_n} ,
\end{equation}
which is what we set out to prove.

%% file: proofCoupling.tex
\subsection{Proof of Lemma \ref{lemMuCoupling}} \label{secProofLemMuCoupling}
 
We will use Algorithm \ref{algSimulConstruction} in Appendix \ref{secAppSimultaneous}, which simultaneously constructs a graph and a tree. We will let $H_n$ and $\hat{H}_n$ denote this graph and this tree, respectively. From $H_n$, we define
\begin{equation}\label{eqNuMDefn}
\nu_s^{(m)} = e_s^{\trans} \sum_{j=0}^m (1-\alpha_n)^j  \tilde{Q}^j ,
\end{equation}
where $\tilde{Q}(i,j) = U_i Q(i,j)$ and $Q$ is the adjacency matrix of $H_n$, normalized to be row stochastic. Note this is simply \eqref{eqMuMDefn}, i.e.\ the definition as $\mu_s^{(m)}$, but computed on $H_n$ (while $\mu_s^{(m)}$ is computed on $G_n$). Similarly, using $\hat{H}_n$, recursively define
\begin{equation}\label{eqNuHatRecursion}
\hat{\nu}_{\phi}(\phi) = 1, \quad \hat{\nu}_{\phi}( (\i,j) ) = \hat{\nu}_{\phi}(\i) \frac{(1-\alpha_n) U_{\i}}{D_{\i}}, (\i,j) \in \hat{A}_l, l > 0 ,
\end{equation}
which is \eqref{eqMuHatRecursion} but computed on $\hat{H}_n$ instead of $\hat{G}_n$. With this notation in place, we will show
\begin{gather}
\mu_s^{(m)}(V_n \setminus K_n) | \{ \tau_G > m, U_s = 1 \} \stackrel{\mathcal{D}}{=} \nu_s^{(m)}(V_n \setminus K_n) | \{ \tau_S > m \}, \label{eqGraphGraphCoupling}  \\
\nu_s^{(m)}(V_n \setminus K_n) = \sum_{j=0}^m \sum_{\i \in \hat{A}_j} U_{\i} \hat{\nu}_{\phi}(\i) \textrm{ when $\tau_S > m$} , \label{eqGraphTreeCoupling} \\
 \sum_{j=0}^m \sum_{\i \in \hat{A}_j} U_{\i} \hat{\nu}_{\phi}(\i) | \{ \tau_S > m \} \stackrel{\mathcal{D}}{=} \sum_{j=0}^m \sum_{\i \in \hat{A}_j} U_{\i} \hat{\mu}_{\phi}(\i) , \label{eqTreeTreeCoupling}
\end{gather}
which, taken together, establish the lemma. (We remind the reader that $\tau_G$ and $\tau_S$, respectively, denote the first iteration at which certain events occur in Algorithm \ref{algGraphConstruction} and Algorithm \ref{algSimulConstruction}, respectively. Specifically, these events are the following: an instub belonging to $v$ with label $g(v) \in \{C,D\}$ is sampled for pairing to an oustub of $v'$ with label $g(v') = D$, or an instub $e$ with label $g(e) = 0$ is sampled for pairing with \textit{any} outstub.)

We begin with \eqref{eqGraphGraphCoupling}. First, observe that by definition $\mu_s^{(m)}(V_n \setminus K_n)$ and $\nu_s^{(m)}(V_n \setminus K_n)$ depend only the $m$-step neighborhood out of $s$ (i.e.\ the subgraph with nodes $\cup_{j=0}^m A_j$) in $G_n$ and $H_n$, respectively. When $\tau_G > m, U_s = 1$ in Algorithm \ref{algGraphConstruction} and $\tau_S > m$ in Algorithm \ref{algSimulConstruction}, these neighborhoods are constructed by the same procedure. Thus, \eqref{eqGraphGraphCoupling} follows.

We next consider \eqref{eqTreeTreeCoupling}, which holds by a similar argument. Specifically, the left and right sides of \eqref{eqTreeTreeCoupling} depend on the first $m$ generations of $\hat{G}_n$ and $\hat{H}_n$, respectively. In Algorithm \ref{algTreeConstruction}, these first $m$ generations of $\hat{G}_n$ are constructed as follows: the root node $\phi$ has attributes $(N_{\phi}, D_{\phi}) \sim f_n^*$ and $U_{\phi} = 1$, non-root nodes $\i$ have attributes $(N_{\i},D_{\i},U_{\i}) \sim f_n$, and $D_{\i}$ offspring are born to $\i$ if and only if $U_{\i} = 1$. In Algorithm \ref{algSimulConstruction}, the root node in $\hat{H}_n$ also has attributes has attributes $(N_{\phi}, D_{\phi}) \sim f_n^*$ and $U_{\phi} = 1$; furthermore, with $\tau_S > m$, non-root nodes $\i$ have attributes $(N_{\i},D_{\i},U_{\i}) \sim f_n$ and $D_{\i}$ offspring are born for either value of $U_{\i}$. Hence, when $\tau_S > m$, modifying the construction of the first $m$ generations of $\hat{H}_n$ such that offspring are born only when $U_{\i} = 1$ yields the construction of the first $m$ generations of $\hat{G}_n$. But, by \eqref{eqNuHatRecursion}, the left side of \eqref{eqTreeTreeCoupling} remains unchanged when this modification occurs. Therefore, \eqref{eqTreeTreeCoupling} follows.

It only remains to prove \eqref{eqGraphTreeCoupling}. For this, we begin with two claims. These claims use the mapping $\Phi$ from graph nodes to tree nodes defined in Algorithm \ref{algSimulConstruction} in Appendix \ref{secAppSimultaneous}. Claim \ref{clmCouplingCopies} states that tree nodes that do not map back to graph nodes do not contribute to the right side of \eqref{eqGraphTreeCoupling}. Claim \ref{clmCouplingNuIso} states that a tree node that does map back to a graph node contributes to the right side of \eqref{eqGraphTreeCoupling} the same value that the corresponding graph node contributes to the left side of \eqref{eqGraphTreeCoupling}. Taken together, these claims will allow us to prove the lemma.

\begin{clm} \label{clmCouplingCopies}
If $\tau_S > m$, $\i \in \hat{A}_j$ for some $j \in \{0,1,\ldots,m\}$, and $\Phi^{-1}(\i) = \emptyset$, then $U_{\i} \hat{\nu}_{\phi}(\i) = 0$.
\end{clm}
\begin{proof}
We begin with some notation. We denote $\i$ by $\i = (i_1,i_2,\ldots,i_j)$, and for $l \leq j$, we let $\i | l = (i_1,i_2,\ldots,i_l)$, with $\i | 0 = \phi$ by convention. Define $l^* = \max \{ l \in \{0,1,\ldots,j\} : \Phi^{-1} ( \i | l ) \neq \emptyset \}$. Note the set over which the maximum is taken is nonempty, since $\Phi^{-1} ( \i | 0 ) = \Phi^{-1} ( \phi ) = s$; furthermore, since $\Phi^{-1}(\i | j) = \Phi^{-1}(\i) = \emptyset$ by assumption, $l^* < j$. In words, $\i | l^*$ is the youngest ancestor of $\i$ that maps to a node in the tree; we let $v' = \Phi^{-1}(\i | l^* )$ denote this node.

We observe $\Phi^{-1}(\i|l) \neq \emptyset\ \forall\ l \in \{0,1,\ldots,l^*-1\}$. To see this, suppose instead that $\Phi^{-1}(\i|l) = \emptyset$ for some such $l$. Then, from the second inner for loop in Algorithm \ref{algSimulConstruction}, the offspring $\i | (l+1)$ was born without adding a node to the graph, which implies $\Phi^{-1}(\i|(l+1)) = \emptyset$. Repeating this argument eventually gives $\Phi^{-1}(\i|l^*) = \emptyset$, a contradiction.

Now suppose $U_{\i} \hat{\nu}_{\phi}(\i) > 0$; we seek a contradiction. First, by \eqref{eqNuHatRecursion}, $U_{\i} \hat{\nu}_{\phi}(\i) > 0$ implies
\begin{equation}\label{eqCouplingCopiesRecusionAssumption}
U_{\i|0} = U_{\i|1} = \cdots = U_{\i} = 1
\end{equation}
which further implies $U_{\Phi^{-1}(\i|l)} = U_{\i|l} = 1\ \forall\ l \in \{0,1,\ldots,l^*\}$, i.e.\ the graph $H_n$ contains a path of length $l^*$ from $s = \Phi^{-1}(\i|0)$ to $v' = \Phi^{-1}(\i|l^*)$ that avoids $K_n$.

Next, note that $\Phi^{-1}(\i|l^*) \neq \emptyset$, $\Phi^{-1}(\i | (l^*+1 )) = \emptyset$ implies that, during the $(l^*+1)$-th iteration of Algorithm \ref{algSimulConstruction}, an outstub of $v'$ was paired with an instub of some $v \in V_n$ that already belonged to the graph, and so a copy of $v$ (namely, $\i | (l^*+1)$) was added to the tree. Consider the following cases for the labels of these nodes at the moment of pairing:
\begin{itemize}
\item If $g(v') = A$ or $g(v) = A$, we have a contradiction, since by assumption, both $v'$ and $v$ already belonged to the graph at the moment of pairing.
\item If $g(v') = B$ or $g(v) = B$, $U_{\i | l^*} = U_{v'} = 0$ or $U_{\i | (l^*+1)} = U_v = 0$, contradicting \eqref{eqCouplingCopiesRecusionAssumption}.
\item If $g(v') = D$, $g(v) \in \{C,D\}$, then $\tau_S = l^* \leq m$ in Algorithm \ref{algSimulConstruction}, a contradiction.
\end{itemize}
The only remaining case is $g(v') = C$ at the moment of pairing. But this contradicts the earlier statement that the graph contains a path from $s$ to $v'$ of length $l^*$ that avoids $K_n$ (since this path was present at start of the $(l^*+1)$-th iteration, it was present at the moment of pairing). 
\end{proof}

\begin{clm} \label{clmCouplingNuIso}
If $\tau_S> m$, then $U_v \nu_s^{(m)}(v) = U_{\Phi(v)} \hat{\nu}_{\phi}(\Phi(v))\ \forall\ v \in \cup_{j=0}^m A_j$. 
\end{clm}
\begin{proof}
We proceed by induction. For the base of induction, we note $A_0 = \{ s \}$, so the statement only needs to be verified for $v = s$. But this is immediate, since $\Phi(s) = \phi$ and $U_s = U_{\phi} = 1$ in Algorithm \ref{algSimulConstruction}, and since $\nu_s^{(0)}(s) = \hat{\nu}_{\phi}(\phi) = 1$ by \eqref{eqNuMDefn} and \eqref{eqNuHatRecursion}.

Now assume $\tau_S > m$ and let $v \in \cup_{j=0}^m A_j$. We consider two cases. 

First, if $v \in A_j$ for some $j \in \{0,1,\ldots,m-1\}$, we can use the inductive hypothesis to write
\begin{equation}
U_v \nu_s^{(m)} ( v ) = U_v \left( \nu_s^{(m)}(v) - \nu_s^{(m-1)}(v) \right) + U_v \nu_s^{(m-1)}(v) = U_v e_s^{\trans} (1-\alpha_n)^m \tilde{Q}^m e_v + U_{\Phi(v)} \hat{\nu}_{\phi}(\Phi(v)) 
\end{equation}
and so it suffices to show $U_v e_s^{\trans} \tilde{Q}^m e_v = 0$. Clearly, this holds when $U_v = 0$. If instead $U_v = 1$, suppose $e_s^{\trans} \tilde{Q}^m e_v > 0$. First, note that $U_v = 1$ and $v \in A_j, j < m$ imply $g(v) \in \{C,D\}$ at the start of the $m$-th iteration of Algorithm \ref{algSimulConstruction}. Furthermore, $e_s^{\trans} \tilde{Q}^m e_v > 0$ implies there exists a path of length $m$ from $s$ to $v$, with every node $w$ along the path satisfying $U_w = 1$. Let $v'$ be the node immediately preceding $v$ on this path, so that an outstub of $v'$ was paired with instub of $v$ during the $m$-th iteration. Then we have $e_s^{\trans} \tilde{Q}^{m-1} e_{v'} > 0$, which implies $g(v') = D$ at the start of the $m$-th iteration of Algorithm \ref{algSimulConstruction}. But $g(v') = D, g(v) \in \{C,D\}$ contradicts $\tau_S > m$ in Algorithm \ref{algSimulConstruction}. Therefore, we must have $e_s^{\trans} \tilde{Q}^m e_v = 0$.

Now suppose $v \in A_m$. Then $U_v \nu_s^{(m-1)}(v) = 0$ (else, $v$ is at most $m-1$ steps from $s$, contradicting $v \in A_m$), so we aim to show $U_ve_s^{\trans} (1-\alpha_n)^m \tilde{Q}^m = U_{\Phi(v)} \hat{\nu}_{\phi}(\Phi(v))$. Since $U_v = U_{\Phi(v)}$ in Algorithm \ref{algSimulConstruction}, this is trivial when $U_v = 0$, and when $U_v = 1$, it suffices to show
\begin{equation}\label{eqCouplingFinalCase}
e_s^{\trans} (1-\alpha_n)^m \tilde{Q}^m = \hat{\nu}_{\phi}(\Phi(v)) .
\end{equation}
Towards this end, let $v' \in \cup_{j=0}^{m-1} A_j$ be the first node whose outstub was paired with an instub of $v$ during the $m$-th iteration (which occurs by assumption $v \in A_m$); by the inductive hypothesis,
\begin{equation}\label{eqCouplingIndHypFinalCase}
U_{v'} \nu_s^{(m-1)} ( v' ) = U_{\Phi(v')} \hat{\nu}_{\phi} ( \Phi(v') ) .
\end{equation}
Now since $D_{v'} = D_{\Phi(v')}$, and since $\Phi(v)$ is an offspring of $\Phi(v')$, we can use \eqref{eqNuHatRecursion} to obtain
\begin{equation}\label{eqCouplingLowerBound}
\frac{(1-\alpha_n) U_{v'} \nu_s^{(m-1)} ( v' )}{ D_{v'} } = \frac{(1-\alpha_n) U_{\Phi(v')} \hat{\nu}_{\phi}(\Phi(v')) }{ D_{\Phi(v')} } = \hat{\nu}_{\phi}(\Phi(v) ) .
\end{equation}
Next, observe the left side of \eqref{eqCouplingLowerBound} is at most $e_s^{\trans} (1-\alpha_n)^m \tilde{Q}^m$ by \eqref{eqNuMDefn}, so we must show this inequality is actually an equality. Suppose instead that the inequality is strict. Then, later in the $m$-th iteration, we must have paired an outstub of some $v''$ s.t.\ $g(v'') = D$ with another instub of $v$. But $g(v) \in \{C,D\}$ after the $v'$ outstub was paired with the $v$ instub, and $g(v'') = D, g(v) \in \{C,D\}$ contradicts $\tau_S > m$ in Algorithm \ref{algSimulConstruction}.
\end{proof}

Having established these claims, we turn to the proof of \eqref{eqGraphTreeCoupling}. Assume $\tau_S > m$. Observe that by Lines \ref{algSimulAddToAm}-\ref{algSimulAddToHatAm} of Algorithm \ref{algSimulConstruction}, $\{ \Phi(v) : v \in A_j \} \subset \hat{A}_j$, so the right side of \eqref{eqGraphTreeCoupling} satisfies
\begin{equation}
\sum_{j=0}^m \sum_{\i \in \hat{A}_j} U_{\i} \hat{\nu}_{\phi}(\i) = \sum_{j=0}^m \left( \sum_{\i \in \hat{A}_j : \Phi^{-1}(\i) = \emptyset} U_{\i} \hat{\nu}_{\phi}(\i) + \sum_{v \in A_j} U_{\Phi(v)} \hat{\nu}_{\phi}(\Phi(v)) \right) .
\end{equation}
Now since $U_{\i} \hat{\nu}_{\phi}(\i) \geq 0$ by definition, Claim \ref{clmCouplingCopies} implies
\begin{equation}
\sum_{\i \in \hat{A}_j : \Phi^{-1}(\i) = \emptyset} U_{\i} \hat{\nu}_{\phi}(\i)  = 0\ \forall\ j \in \{0,1,\ldots,m\} .
\end{equation}
Furthermore, since $\nu_s^{(m)}(v) = 0\ \forall\ v \notin \cup_{j=0}^m A_j$ (which holds by \eqref{eqNuMDefn}), Claim \ref{clmCouplingNuIso} implies 
\begin{equation}
\nu_s^{(m)} ( V_n \setminus K_n ) = \sum_{j=0}^m \sum_{v \in A_j} U_v \nu_s^{(m)}(v) = \sum_{j=0}^m \sum_{v \in A_j} U_{\Phi(v)} \hat{\nu}_{\phi}(\Phi(v)) .
\end{equation}
Finally, combining the previous three equations yields \eqref{eqGraphTreeCoupling}.

%% file: proofCouplingBroken.tex
\subsection{Proof of Lemma \ref{lemCouplingFailure}} \label{secProofLemCouplingFailure}

We begin with some initial definitions that will be used throughout the proof. Specifically, let $\zeta_n = \E_n [ D_{\i} ]$ and $\lambda_n = \E_n [ N_{\i} U_{\i} ]$, where $(N_{\i}, D_{\i}, U_{\i}) \sim f_n$ are the attributes for a non-root node in the tree. Then, conditioned on $\Omega_n$,
\begin{gather}
\zeta_n = \frac{1}{L_n} \sum_{h=1}^n N_h D_h = \frac{ \eta_2 ( 1 + O(n^{-\gamma}) ) }{ \eta_1 ( 1 + O(n^{-\gamma}) ) } = \zeta ( 1 + O(n^{-\gamma}) ) , \\
\lambda_n = \frac{1}{L_n} \sum_{h=1}^n N_h^2 U_h = \frac{ \eta_3 ( 1 + O(n^{-\gamma}) ) }{ \eta_1 ( 1 + O(n^{-\gamma}) ) } = \lambda ( 1 + O(n^{-\gamma}) ) .
\end{gather}

Similarly, let $\zeta_n^* = \E_n [ D_{\phi} ]$ and $\lambda_n^* = \E_n [ N_{\phi} ]$, where $(N_{\phi}, D_{\phi}) \sim f_n^*$ are the attributes for the root node of the tree, so that conditioned on $\Omega_n$,
\begin{gather}
\zeta_n^* = \frac{1}{\sum_{h=1}^n U_h} \sum_{h=1}^n D_h U_h = \zeta^* ( 1 + O ( n^{-\gamma} ) ) , \quad \lambda_n^* = \frac{1}{\sum_{h=1}^n U_h} \sum_{h=1}^n N_h U_h = \lambda^* ( 1 + O ( n^{-\gamma} ) )  .
\end{gather}

We now explain our approach for bounding $\P [ \tau_G \leq m | U_s = 1 ]$. First, observe that, conditioned on $U_s = 1$, the graphs in Algorithms \ref{algGraphConstruction} and \ref{algSimulConstruction} (the graph and simultaneous constructions, respectively) are constructed by the same procedure until $\tau_G = m$ or $\tau_S = m$; further, $\tau_G$ is assigned in Algorithm \ref{algGraphConstruction} by the same procedure $\tau_S$ is assigned in Algorithm \ref{algSimulConstruction}. This implies 
\begin{equation}
\P [ \tau_G \leq m | U_s = 1 ] = \P [ \tau_S \leq m ] .
\end{equation}
Next, for $i \in \{0,1\}$, define 
\begin{equation}
E_i = \{ g(e) = i \textrm{ at the moment $\tau_S$ is assigned in Algorithm \ref{algSimulConstruction}} \} .
\end{equation}
In other words, $E_0$ is the event that the coupling breaks because a paired instub was sampled, while $E_1$ is the event that the coupling breaks because an unpaired instub that forms an edge $v' \rightarrow v$ s.t.\ $g(v') = D, g(v) \in \{C,D\}$ was sampled. Furthermore, for $l \in \{1,2,\ldots,m\}$, define
\begin{equation}\label{eqZhatLdefn}
\hat{Z}_l = \sum_{\i \in \hat{A}_{l-1}} D_{\i} ,
\end{equation}
which is the total number of outstubs in generation $l-1$ of the tree; note $\hat{Z}_l = | \hat{A}_l |$. Finally, let $\{ y_n : n \in \N \}$ be a sequence tending to infinity (which we will choose later), and let
\begin{equation}
F_m = \left\{  \max_{1 \leq l \leq m} \frac{\hat{Z}_l}{\zeta^{l-1}} \leq \zeta^* y_n \right\} .
\end{equation}
We can then use the previous four equations to write
\begin{equation}\label{eqCouplingFailureApproach}
\P [ \tau_G \leq m | U_s = 1 ]  \leq O ( n^{-\delta} ) + \P [ F_m^C | \Omega_n ] + \sum_{i=0}^1 \sum_{l=1}^m \P [ \tau_S = l  , E_i , F_m | \Omega_n ] .
\end{equation}
where we have also used $\P[ \Omega_n^C ] = O(n^{-\delta})$ by Assumption \ref{assDegSeq}. In the remainder of this appendix, we bound each term in \eqref{eqCouplingFailureApproach}.

\subsubsection{$\P [ F_m^C | \Omega_n ]$ bound}

First, note that $\{ D_{\i} \}_{\i \in \hat{A}_{l-1} }$ are identically distributed and independent of $\hat{Z}_{l-1} = | \hat{A}_{l-1} |$. Hence,
\begin{equation}
\E_n [ \hat{Z}_l ] = \E_n [ \E_n [ \hat{Z}_l | \hat{Z}_{l-1} ] ] = \E_n [ \hat{Z}_{l-1} \E_n [ D_{\i} | \hat{Z}_{l-1} ] ] = \E_n [ \hat{Z}_{l-1} ] \E_n [ D_{\i} ] = \E_n [ \hat{Z}_{l-1} ] \zeta_n ,
\end{equation}
and so applying recursively gives
\begin{equation}
\E_n [ \hat{Z}_l ] = \E_n [ \hat{Z}_1 ] \zeta_n^{l-1} = \E_n [  D_{\phi} ] \zeta_n^{l-1} = \zeta_n^* \zeta_n^{l-1} .
\end{equation}
Now let $X_l = \hat{Z}_l / ( \zeta_n^* \zeta_n^{l-1} )$, so that $\E_n [ X_l ] = 1$. Furthermore, define
\begin{equation}
\mathcal{G}_l = \sigma(  \{ N_h, D_h, U_h : 1 \leq h \leq n \}\cup \{ D_{\i} : \i \in \hat{A}_j , 0 \leq j < l \} ). 
\end{equation}
Then for $j > 0$,
\begin{align}
\E [ X_{l+j} | \mathcal{G}_l ] & = \frac{\E [ \hat{Z}_{l+j} | \mathcal{G}_l ] }{ \zeta_n^* \zeta_n^{l+j-1} } = \frac{\E [ \hat{Z}_{l+j-1} | \mathcal{G}_l ] \E [  D_{\i} | \mathcal{G}_l ] }{ \zeta_n^* \zeta_n^{l+j-1} } = \frac{\E [ \hat{Z}_{l+j-1} | \mathcal{G}_l ]  }{ \zeta_n^* \zeta_n^{l+j-2} } = \E [ X_{l+j-1} | \mathcal{G}_l ] \\
& = \cdots = \E[ X_l | \mathcal{G}_l ] = X_l ,
\end{align}
so $\{ X_l : l \in \N \}$ is a martingale with respect to $\{ \mathcal{G}_l : l \in \N \}$. This implies, by Doob's inequality,
\begin{equation}
\P_n \left[  \max_{1 \leq l \leq m} X_l > \frac{  y_n}{ ( 1 + O (n^{-\gamma}))^{m} }  \right] \leq \frac{ ( 1 + O (n^{-\gamma}))^{m}}{  y_n } ,
\end{equation}
where we have used $\E_n [ X_m ] = 1$. Using this bound, we can obtain
\begin{align}
\P [ F_m^C | \Omega_n ] & = \P \left[  \max_{1 \leq l \leq m} \frac{\hat{Z}_l}{\zeta^{l-1}} > \zeta^* y_n  \middle| \Omega_n \right] = \P \left[ \max_{1 \leq l \leq m} \frac{X_l \zeta_n^* \zeta_n^{l-1}}{\zeta^* \zeta^{l-1} } > y_n \middle|  \Omega_n \right] \\
& = \P \left[ \max_{1 \leq l \leq m} X_l  ( 1 + O ( n^{-\gamma} ) )^l >  y_n \middle|  \Omega_n \right]  \leq \P \left[ \max_{1 \leq l \leq m} X_l  > \frac{y_n}{ ( 1 + O ( n^{-\gamma} ) )^{m}} \middle|  \Omega_n \right] \\
& = \frac{1}{\P[\Omega_n]} \E \left[ 1(\Omega_n) \P_n \left[  \max_{1 \leq l \leq m} X_l  > \frac{y_n}{( 1 + O ( n^{-\gamma} ) )^{m}}    \right] \right] \\
& \leq  \frac{1}{\P[\Omega_n]} \E \left[ 1(\Omega_n) \frac{  ( 1 + O (n^{-\gamma}))^{m}}{  y_n } \right] = \frac{ ( 1 + O (n^{-\gamma}))^{m}}{ y_n } = O ( y_n^{-1} ),
\end{align}
where in the third line we used the tower property and the fact that $1(\Omega_n)$ is fixed given the degree sequence, and where the final equality holds by the assumption $m = O(n^{\gamma})$ in the statement of the lemma, since then $( 1 + O (n^{-\gamma}))^m = ( 1 + \frac{O(1)}{m} )^m = e^{O(1)} = O(1)$. 

\subsubsection{$\P [ \tau_S = l  , E_0 , F_m | \Omega_n ]$ bound} \label{appCouplingFailureE0Fm}

Let $l \in \{1,2,\ldots,m\}$. We observe
\begin{equation}\label{eqE0proofTower}
\P [ \tau_S = l , E_0, F_m | \Omega_n ] =   \E \left[ 1 \left( F_m  \right) \P_n \left[  \tau_S = l , E_0  \middle| \{ \hat{Z}_{j} \}_{j=1}^{m+1}  \right] \middle| \Omega_n \right]
\end{equation}
which holds because $1(\Omega_n)$ and $1(F_m)$ are fixed given the degree sequence and $\{ \hat{Z}_{j} \}_{j=1}^{m}$. Next, observe $\{ \tau_S = l , E_0 \}$ occurs if and only if, during the $l$-th iteration, we sample an instub that has already been paired while attempting to pair an outstub belonging to a node $v' \in A_{l-1}$. We aim to bound the probability of this event.

Consider any such outstub. Since we sample instubs uniformly from the set of all $L_n$ instubs, the probability of sampling a paired instub is the fraction of paired instubs at the moment we attempt to pair the outstub under consideration. This fraction is clearly bounded by the fraction of paired instubs at the end of the $l$-th iteration. Furthermore, since each time we pair an instub of $v \in V$ in the graph, we also add a node to the tree with the same attributes as $v$, the numerator of this fraction is further bounded by the number of nodes in the tree at the end of the $l$-th iteration, which by definition is
\begin{equation}\label{eqNumberOfNodesInTree}
\frac{1}{L_n} \sum_{j=1}^{l+1} \hat{Z}_{j} .
\end{equation}

Now consider the number of such outstubs. By definition, this is $\sum_{v' \in A_{l-1} } D_{v'}$. Furthermore, since each time we add a node to $A_{l-1}$ in the graph, we also add a node with the same attributes to $\hat{A}_{l-1}$ in the tree, we have the bound
\begin{equation}
\sum_{v' \in A_{l-1} } D_{v'} \leq \sum_{\i \in \hat{A}_{l-1}} D_{\i} \triangleq \hat{Z}_l .
\end{equation}

Combining these arguments, letting $\textrm{Bin}$ denote a binomial random variable, and using Markov's inequality, we can write
\begin{align}
\P_n \left[ \tau_S = l, E_0 \middle| \{ \hat{Z}_{j} \}_{j=1}^{m+1} \right] & \leq \P_n \left[ \textrm{Bin} \left( \hat{Z}_l , \frac{\sum_{j=1}^{l+1} \hat{Z}_{j}}{L_n} \right) \geq 1 \middle|  \{ \hat{Z}_{j} \}_{j=1}^{m+1} \right] \\
& \leq \E_n \left[ \textrm{Bin} \left( \hat{Z}_l , \frac{\sum_{j=1}^{l+1} \hat{Z}_{j}}{L_n} \right)  \middle|  \{ \hat{Z}_{j} \}_{j=1}^{m+1} \right] = \hat{Z}_l \frac{\sum_{j=1}^{l+1} \hat{Z}_{j}}{L_n} .  \label{eqE0proofTail}
\end{align}

Next, we recognize $1(F_m) \hat{Z}_l \leq \zeta^* \zeta^{l-1} y_n$ by definition of $F_m$, and we combine \eqref{eqE0proofTower} and \eqref{eqE0proofTail},
\begin{equation}
\P [ \tau_S = l , E_0, F_m | \Omega_n ] \leq \E \left[ 1(F_{m}) \hat{Z}_l \frac{\sum_{j=1}^{l+1} \hat{Z}_{j}}{L_n}  \middle| \Omega_n \right] \leq \zeta^* \zeta^{l-1} y_n \sum_{j=1}^{l+1} \E \left[ \frac{\hat{Z}_{j}}{L_n}  \middle| \Omega_n \right] .
\end{equation}

Furthermore, by definition of $\Omega_n$, we have
\begin{equation}\label{eqNumberOfNodesInTreeExp}
\E \left[ \frac{\hat{Z}_{j}}{L_n}  \middle| \Omega_n \right] = \E \left[ \frac{ \E_n [ \hat{Z}_{j} ] }{L_n}  \middle| \Omega_n \right] = \E \left[ \frac{ \zeta_n^* \zeta_n^{j-1} }{L_n}  \middle| \Omega_n \right] = \frac{\zeta^* \zeta^{j-1}}{ n \eta_1 } ( 1 + O(n^{-\gamma}) )^j = O \left( \frac{\zeta^{j-1} }{ n } \right) ,
\end{equation}
where $( 1 + O(n^{-\gamma}) )^j = O(1)$ again follows from $m = O(n^{-\gamma})$. We have therefore shown
\begin{equation}
\P [ \tau_S = l , E_0, F_m | \Omega_n ] = O \left( \frac{y_n}{n} \zeta^{l-1} \sum_{j=0}^l \zeta^j \right) .
\end{equation}

\subsubsection{$\P [ \tau_S = l  , E_1 , F_m | \Omega_n ]$ bound}

We will use the same approach to bound this term as we used to bound the previous one. Observe $\{ \tau_S = l, E_1 \}$ occurs if and only if, during the $l$-th iteration, we sample an instub belonging to $v$ s.t.\ $g(v) \in \{C,D\}$ while attempting to pair an outstub belonging to a node $v' \in A_{l-1}$ s.t.\ $g(v') = D$. The key step in the derivation will be bounding the number of such instubs and outstubs.

First, the number of such outstubs is clearly bounded the number of \textit{all} outstubs paired during the $l$-th iteration. As we argued previously, this is further bounded by $\hat{Z}_l$.

Next, for $j \in \{1,2,\ldots,m+1\}$, define $\hat{V}_j = \sum_{\i \in \hat{A}_{j-1} } N_{\i} U_{\i}$. Similar to the previous argument, the number of such instubs while pairing any such outstub is bounded by the number of instubs belonging to $U_v = 1$ nodes in the graph at the end of the $l$-th iteration. Since each time we add a node to the graph, we also add a node to the tree with the same attributes, the former quantity is bounded by the same quantity computed on the tree, i.e.\
\begin{equation}
\sum_{j=1}^{l+1} \sum_{\i \in \hat{A}_{l-1} } N_{\i} U_{\i} = \sum_{j=1}^{l+1} \hat{V}_j .
\end{equation}

Hence, as in the analysis of $\P [ \tau_S = l  , E_0 , F_m | \Omega_n ]$,
\begin{align}
\P [ \tau_S = l , E_1, F_m | \Omega_n ] & = \E \left[ 1(F_m) \P_n \left[  \tau_S = l , E_1  \middle| \{ \hat{Z}_j \}_{j=1}^m , \{ \hat{V}_j \}_{j=1}^{l+1}  \right] \middle| \Omega_n \right] \\
& \leq \E \left[ 1(F_m) \hat{Z}_{l} \frac{\sum_{j=1}^{l+1} \hat{V}_j}{L_n} \middle| \Omega_n \right] \leq \zeta^* \zeta^{l-1} y_n \sum_{j=1}^{l+1}  \E \left[ \frac{ \E_n [ \hat{V}_j ] }{ L_n } \middle| \Omega_n \right] .
\end{align}

Our final step is to compute $\E_n [ \hat{V}_j ]$. For $j > 1$, we have
\begin{equation}
\E_n [ \hat{V}_j ] = \E_n [ \hat{Z}_{j-1} ] \E_n [ N_{\i} U_{\i} ] = \zeta_n^* \zeta_n^{j-2} \lambda_n ,
\end{equation}
where the first inequality holds since $| \hat{A}_{j-1} | = \hat{Z}_{j-1}$ and since $\{ N_{\i} U_{\i} : \i \in \hat{A}_{l-1} \}$ are identically distributed and independent of $\hat{Z}_{j-1}$; the second inequality follows from previous derivations. Therefore,
\begin{equation}
\E \left[ \frac{ \E_n [ \hat{V}_j ] }{ L_n } \middle| \Omega_n \right] = \E \left[ \frac{\zeta_n^* \zeta_n^{j-2} \lambda_n}{L_n} \middle| \Omega_n \right] = \frac{ \zeta^* \zeta^{j-2} \lambda }{ n \eta_1 } ( 1 + O(n^{-\gamma}) )^j = O \left( \frac{ \zeta^{j-2} }{ n } \right) .
\end{equation}
For $j = 1$, since $\hat{A}_0 = \{ \phi \}$ with $U_{\phi} = 1$, we simply have $\E_n [ \hat{V}_1 ] = \E_n [ N_{\phi} ] = \lambda_n^*$, so
\begin{equation}
\E \left[ \frac{ \E_n [ \hat{V}_1 ] }{ L_n } \middle| \Omega_n \right] = \E \left[ \frac{\lambda_n^*}{L_n} \middle| \Omega_n \right] = \frac{ \lambda^* }{ n \eta_1 } ( 1 + O(n^{-\gamma}) ) = O \left( \frac{ 1 }{ n } \right) .
\end{equation}

Combining previous arguments, we obtain
\begin{equation}
\P [ \tau_S = l , E_1, F_m | \Omega_n ] = O \left( \frac{ y_n }{ n } \zeta^{l-1} \sum_{j=0}^{l-1} \zeta^j \right) .
\end{equation}

\subsubsection{Overall bound}

Combining the bounds from the previous sections, we obtain
\begin{equation}
\P [ \tau_G \leq m | U_s = 1 ] = O \left( n^{-\delta} + y_{n}^{-1} + \frac{y_n}{n} \sum_{l=1}^m \zeta^{l-1} \sum_{j=0}^l \zeta^j  \right) .
\end{equation}

By Assumption \ref{assDegSeq}, we have $\zeta > 1$, which implies
\begin{equation}
\sum_{l=1}^m \zeta^{l-1} \sum_{j=0}^l \zeta^j = \sum_{l=1}^m \zeta^{l-1} \frac{\zeta^{l+1}-1}{\zeta-1} \leq \frac{1}{\zeta-1} \sum_{l=1}^m \zeta^{2l} = \frac{\zeta^2 ( \zeta^{2m} - 1 )}{(1-\zeta)^2} \leq \left( \frac{\zeta}{\zeta-1} \right)^2 \zeta^{2m} = O ( \zeta^{2m} ) .
\end{equation}

We thus obtain
\begin{equation}
\P [ \tau_G \leq m | U_s = 1 ]= O \left( n^{-\delta} + y_n^{-1} + y_n \zeta^{2m} / n \right) 
\end{equation}

Finally, we choose $y_n$ to minimize the bound. This yields
\begin{equation}
\P [ \tau_G \leq m | U_s = 1 ] =  O \left( n^{-\delta} + \zeta^m / \sqrt{n} \right) ,
\end{equation}
which is what we set out to prove.

\subsubsection{Remark}

Towards bounding $\P [ \tau_S = l  , E_1 , F_m | \Omega_n ]$, we bounded the number of outstubs belonging to $v' \in A_{l-1}$ s.t.\ $g(v') = D$ by $\hat{Z}_l$. We can in fact obtain a tighter bound, suggesting that the bound above was unnecessarily loose. However, we show here that this modified approach ultimately yields the same result. For this, suppose in Assumption \ref{assDegSeq} we add the event
\begin{equation}
\Omega_{n,7} =\left\{ \left| \frac{\sum_{h=1}^n N_h D_h U_h}{n} - \eta_4 \right| \leq n^{-\gamma}  \right\} ,
\end{equation}
and we set $\zeta' = \eta_4 / \eta_1$. Let $\{ y_n' : n \in \N \}$ be a sequence tending to infinity, and define
\begin{equation}
\hat{Z}_l' = \sum_{\i \in \hat{A}_{l-1}} \left(  \prod_{j=0}^{l-1} U_{ \i | j }   \right) D_{\i} , \quad F_m' = \left\{  \max_{1 \leq l \leq m} \frac{\hat{Z}_l'}{ (\zeta')^{l-1}} \leq \zeta^* y_n' \right\} .
\end{equation}
Using the same approach we used to bound $\P[ F_m^C | \Omega_n ]$, it is possible to show
\begin{equation}
\P \left[ (F_m')^C \middle| \Omega_n \right] = O \left( (y_n')^{-1} \right).
\end{equation}
Then, when analyzing $\{ \tau_S = l, E_1 \}$, we could bound the number of outstubs belonging to $v' \in A_{l-1}$ s.t.\ $g(v') = D$ by $\hat{Z}_l'$, which is tighter  than $\hat{Z}_l$ (used above). This would ultimately yield
\begin{equation}
\P [ \tau_S = l, E_1, F_m' | \Omega_n ] = O \left( \frac{y_n'}{n} (\zeta' )^{l-1} \sum_{j=0}^{l-1} \zeta^j \right)  ,
\end{equation}
which would imply
\begin{equation}
\P [ \tau_S \leq m, E_1, F_m' | \Omega_n ] = O \left( \frac{y_n'}{n} \sum_{l=1}^m (\zeta' )^{l-1} \sum_{j=0}^{l-1} \zeta^j \right) = O \left( \frac{y_n'}{n} ( \zeta \zeta' )^m \right) .
\end{equation}
Using this approach, we could write
\begin{align}
& \P [ \tau_G \leq m | U_s = 1 ] \\
& \quad = O ( n^{-\delta} ) + \P \left[ F_m^C \middle| \Omega_n \right] + \P [ \tau_S \leq m , E_0, F_m | \Omega_n ] + \P \left[ (F_m')^C \middle| \Omega_n \right] + \P [ \tau_S \leq m , E_1, F_m' | \Omega_n ] \\
& \quad = O \left( n^{-\delta} + y_n^{-1} + \frac{y_n}{n} \zeta^{2m} + ( y_n' )^{-1} + \frac{y_n'}{n} ( \zeta \zeta' )^m \right) = O \left( n^{-\delta} + \frac{ \zeta^m }{\sqrt{n}} + \frac{ ( \zeta \zeta' )^{m/2} }{\sqrt{n}} \right) ,
\end{align}
where in the final step we chose $y_n, y_n'$ to minimize the bound. However, since $\zeta' \leq \zeta$ by definition, this bound is ultimately $O( n^{-\delta} + \zeta^m / \sqrt{n} )$, which is the same bound we obtained above.

%% file: proofTailBound.tex
\subsection{Proof of Lemma \ref{lemTreeTailBound}} \label{secProofLemTreeTailBound}

For $j \in \{1,2,\ldots,m\}$, let $X_j = \sum_{\i \in \hat{A}_j} U_{\i} \hat{\mu}_{\phi}(\i)$. Additionally, for each $n \in \N$, we define
\begin{equation}
\hat{p}_n = \frac{\sum_{h=1}^n U_h N_h}{L_n}   .
\end{equation}
Note that, by Assumption \ref{assDegSeq}, we have $| \hat{p}_n - p | < n^{-\gamma}$ when $\Omega_n$ holds.

Before proceeding, we present some intermediate results required for our analysis.

\begin{clm} \label{clmMartingale}
For any $i,j \in \N$ s.t.\ $j \geq i$, $\E_n [ X_j | X^i ] = ( (1-\alpha_n) \hat{p}_n )^{j-i}  X_i$, where $X^i = \{ X_l \}_{l=1}^i$.
\end{clm}
\begin{proof}
We first observe
\begin{equation}\label{eqMuHatGenRecursion}
X_j = \sum_{\i \in \hat{A}_j} U_{\i} \hat{\mu}_{\phi}(\i) = \sum_{\i \in \hat{A}_j} \prod_{l=0}^{j-1} \frac{(1-\alpha_n) U_{\i | l} }{  D_{\i | l} } U_{\i} = \sum_{\i \in \hat{A}_{j-1}} \prod_{l=0}^{j-1} \frac{(1-\alpha_n) U_{\i | l} }{  D_{\i | l} } \sum_{k=1}^{D_{\i}} U_{ (\i,k) } ,
\end{equation}
where the first equality follows from \eqref{eqMuHatRecursion} and the second follows since, by Algorithm \ref{algTreeConstruction},
\begin{equation}
\hat{A}_j = \left\{ (\i,k) : \i \in \hat{A}_{j-1}, U_{\i} = 1, k \in \{ 1,2,\ldots,D_{\i} \} \right\} .
\end{equation}
Next, let $\i \in \hat{A}_{j-1}$ s.t.\ $U_{\i} = 1$. For each $k \in \{1,2,\ldots,D_{\i} \}$, observe
\begin{align}
& \E \left[ U_{(\i,k)} \middle| \{ N_h, D_h, U_h : 1 \leq h \leq n \} \cup \{ U_{\i'} , D_{\i'} : \i' \in \hat{A}_s, s < j   \}   \right] \\
& \quad = \E \left[ U_{(\i,k)} \middle| \{ N_h, D_h, U_h : 1 \leq h \leq n \} \right]  = \frac{\sum_{h=1}^n U_h N_h}{L_n}  = \hat{p}_n , \label{eqTreeUnknownSampleProb}
\end{align}
which follows since in Algorithm \ref{algTreeConstruction}, the attributes $(N_{(\i,k)}, D_{(\i,k)}, U_{(\i,k)})$ are sampled from $f_n$, independent of the attributes of nodes in previous generations. Combining \eqref{eqMuHatGenRecursion} and \eqref{eqTreeUnknownSampleProb} gives

\begin{align}
& \E \left[ X_j \middle| \{ N_h, D_h, U_h : 1 \leq h \leq n \} \cup \{ U_{\i} , D_{\i} : \i \in \hat{A}_s, s < j   \}   \right] = \sum_{\i \in \hat{A}_{j-1}} \prod_{l=0}^{j-1} \frac{(1-\alpha_n) U_{\i | l} }{  D_{\i | l} } \sum_{k=1}^{D_{\i}} \hat{p}_n \\
& \quad = \sum_{\i \in \hat{A}_{j-1}} \prod_{l=0}^{j-2} \frac{(1-\alpha_n) U_{\i | l} }{  D_{\i | l} } \frac{(1-\alpha_n) U_{\i } }{  D_{\i } } ( D_{\i} \hat{p}_n ) = (1-\alpha_n) \hat{p}_n \sum_{\i \in \hat{A}_{j-1}} \prod_{l=0}^{j-2} \frac{(1-\alpha_n) U_{\i | l} }{  D_{\i | l} } U_{\i} \\
& \quad\quad = (1-\alpha_n) \hat{p}_n \sum_{\i \in \hat{A}_{j-1}} \hat{\mu}_{\phi}(\i) U_{\i} = (1-\alpha_n) \hat{p}_n X_{j-1} \label{eqRecursionSetup} .
\end{align}
Note that $X^i$ is a function of $\{ U_{\i} , D_{\i} : \i \in \hat{A}_s, s < j \}$, so we can also write
\begin{equation}
\E \left[ X_j \middle| \{ N_h, D_h, U_h : 1 \leq h \leq n \} \cup \{ U_{\i} , D_{\i} : \i \in \hat{A}_s, s < j   \} \cup  X^{i}   \right] = (1-\alpha_n) \hat{p}_n X_{j-1} .
\end{equation}
Then, taking conditional expectation with respect to $\{ N_h, D_h, U_h : 1 \leq h \leq n \} \cup  X^i $ on both sides,
\begin{equation}
\E_n \left[ X_j \middle|  X^{i}  \right] = (1-\alpha_n) \hat{p}_n \E_n \left[ X_{j-1} | X^i \right] ,
\end{equation}
and so applying recursively gives
\begin{equation}
\E_n [ X_j | X^i ] = ( (1-\alpha_n) \hat{p}_n )^{j-i} X_i ,
\end{equation}
which completes the proof.
\end{proof}

\begin{clm} \label{clmAzumaLemma}
Let $Z$ be a random variable satisfying $\E [ Z ] = 0$ and $a \leq Z \leq b$ a.s. Then 
\begin{equation}
\E \left[ e^{\lambda Z} \right] \leq e^{\lambda^2 (b-a)^2/8}\ \forall\ \lambda > 0 .
\end{equation}
\end{clm}
\begin{proof}
See, for example, Lemma 5.1 in \cite{dubhashi2009concentration}.
\end{proof}

\begin{clm} \label{clmAzumaApplication}
For any $j \in \N$ and any $c_j > 0$, define $Y_j =  c_j ( X_j - (1-\alpha_n) \hat{p}_n X_{j-1} ) )$. Then
\begin{equation}
\E_n [ \exp ( \lambda Y_j ) | X_{j-1} ] \leq \exp \left( \frac{\lambda^2}{8} ( c_j ( 1-\alpha_n )^j )^2 \right) .
\end{equation}
\end{clm}
\begin{proof}
Note $\E_n [ Y_j | X_{j-1} ] = 0$ by Claim \ref{clmMartingale}. Furthermore, $X_j \in [ 0 , (1-\alpha_n) X_{j-1} ]$ by \eqref{eqMuHatRecursion}, so
\begin{equation}
Y_j \leq  c_j ( 1- \alpha_n ) ( 1 - \hat{p}_n ) X_{j-1} \triangleq b_j , \quad Y_j \geq - c_j ( 1- \alpha_n ) \hat{p}_n  X_{j-1} \triangleq a_j .
\end{equation}
Therefore, applying Claim \ref{clmAzumaLemma} gives
\begin{equation}
\E_n [ \exp ( \lambda Y_j ) | X_{j-1} ] \leq \exp \left( \frac{\lambda^2}{8} ( c_j ( 1-\alpha_n ) X_{j-1} )^2 \right) ,
\end{equation}
and using $X_{j-1} \leq (1-\alpha_n)^{j-1}$ (which again follows from \eqref{eqMuHatRecursion}) completes the proof.
\end{proof}

We now turn to the proof of the lemma. First, we write
\begin{equation}\label{eqWBPtailTwoCases}
\P \left[ \alpha_n \sum_{j=1}^{m-1} X_j + X_{m} \geq \epsilon \right] \leq \P \left[ \alpha_n \sum_{j=1}^{m-1} X_j + X_{m} \geq \epsilon \middle| {\Omega}_n \right] + \P [ {\Omega}_n^C ] .
\end{equation}
Recall $\P [ {\Omega}_n^C ] = O(n^{-\delta})$ by Assumption \ref{assDegSeq}, so it remains to bound the first summand. First,
\begin{align}\label{eqWBPtailPtoPn}
\P \left[ \alpha_n \sum_{j=1}^{m-1} X_j + X_{m} \geq \epsilon \middle| \Omega_n \right] = \frac{1}{ \P [ {\Omega}_n ] } \E \left[ 1 \left( {\Omega}_n \right) \P_n \left[ \alpha_n \sum_{j=1}^{m-1} X_j + X_{m} \geq \epsilon \right] \right]  .
\end{align}

For the term inside the expectation, we have
\begin{align}\label{eqPnBoundTwoTerms}
\P_n \left[ \alpha_n \sum_{j=1}^{m-1} X_j + X_{m} > \epsilon \right] & \leq \P_n \left[ \alpha_n \sum_{j=1}^{m-1} X_j  > \frac{\epsilon}{2} \right]  + \P_n \left[ X_{m} > \frac{\epsilon}{2} \right] .
\end{align}

For the second summand, we use Markov's inequality to write
\begin{equation}
 \P_n \left[ X_{m} > \frac{\epsilon}{2} \right] \leq \frac{2 \E_n [ X_m ]}{\epsilon} =  \frac{2 (1-\alpha_n)^m \hat{p}_n^m }{\epsilon} < \frac{2 \hat{p}_n^m}{\epsilon} .
\end{equation}

Recall that $\hat{p}_n \leq p + n^{-\gamma}$ when $\Omega_n$ holds. Therefore, by assumption $m = O(n^{\gamma})$, we obtain
\begin{equation}\label{eqSecondSummandComplete}
 1 (\Omega_n) \P_n \left[ X_{m} > \frac{\epsilon}{2} \right]  < \frac{2 ( p + n^{-\gamma} )^m }{ \epsilon } = p^m \frac{ 2 ( 1 + \frac{1/p}{n^{\gamma}} )^m }{ \epsilon } = O \left( p^m \right) .
\end{equation}

We next consider the first summand in \eqref{eqPnBoundTwoTerms}. First, we use the Chernoff bound to write
\begin{equation} \label{eqChernoffExpression}
\P_n \left[ \alpha_n \sum_{j=1}^{m-1} X_j  > \epsilon \right] \leq \min_{\lambda > 0} e^{ - \lambda \epsilon } \E_n \left[ \prod_{j=1}^{m-1} \exp \left( \lambda \alpha_n X_j \right) \right] .
\end{equation}

To analyze \eqref{eqChernoffExpression}, we require a definition: for $j = 0, 1, \ldots , m-1$, let
\begin{equation}\label{eqCjDefn}
c_j = \E_n \left[ \alpha_n \sum_{ i=0 }^{ m-j-1 } X_i  \right] =  \alpha_n \sum_{ i=0 }^{ m-j-1 } ( (1-\alpha_n) \hat{p}_n )^i = \frac{ \alpha_n ( 1 - ( (1-\alpha_n) \hat{p}_n )^{m-j} ) }{ 1 - (1-\alpha_n) \hat{p}_n } ,
\end{equation}
where we have used Claim \ref{clmMartingale} and $X_0 = \sum_{ \i \in \hat{A}_0 } U_{ \i } \mu_{\phi}(\i) = U_{\phi} \mu_{\phi}(\phi) = 1$ by definition. From \eqref{eqCjDefn}, it is straightforward to show the following:
\begin{equation}\label{eqCjObservations}
c_0 - \alpha_n = \E_n \left[ \alpha_n \sum_{i=1}^{m-1} X_i \right] , \quad c_{m-1} = \alpha_n , \quad c_{j} = \alpha_n + (1-\alpha_n) \hat{p}_n c_{j+1} 
\end{equation}

Now for $j \in \{1,2,\ldots,m-1\}$, we use Claim \ref{clmAzumaApplication} and \eqref{eqCjObservations} to obtain
\begin{align}
\E_n [ \exp ( \lambda c_j X_j ) | X^{j-1} ] & = \E_n [ \exp ( \lambda c_j ( X_j - (1-\alpha_n) \hat{p}_n X_{j-1} ) ) | X^{j-1} ] \exp ( \lambda c_j (1-\alpha_n) \hat{p}_n X_{j-1} ) \\
& \leq \exp \left( \frac{\lambda^2}{8} ( c_j ( 1-\alpha_n )^j )^{2} \right) \exp ( \lambda (c_{j-1} -\alpha_n) X_{j-1} ) .
\end{align}

We can then recursively apply to the expectation in \eqref{eqChernoffExpression}, i.e.\
\begin{align}
& \E_n \left[ \prod_{j=1}^{m-1} \exp \left( \lambda \alpha_n X_j \right) \right] = \E_n \left[ \prod_{j=1}^{m-2} \exp \left( \lambda \alpha_n X_j \right)  \exp \left( \lambda c_{m-1} X_{m-1} \right) \right] \\
& \quad = \E_n \left[ \prod_{j=1}^{m-2} \exp \left( \lambda \alpha_n X_j \right) \E_n \left[ \exp \left( \lambda c_{m-1} X_{m-1} \right) \middle| X^{m-2} \right] \right] \\
& \quad \leq \E_n \left[ \prod_{j=1}^{m-2} \exp \left( \lambda \alpha_n X_j \right) \exp (  \lambda ( c_{m-2} - \alpha_n ) X_{m-2} ) \right] \exp \left( \frac{\lambda^2}{8} \left( c_{m-1} ( 1-\alpha_n )^{m-1} \right)^2 \right) \\
& \quad = \E_n \left[ \prod_{j=1}^{m-3} \exp \left( \lambda \alpha_n X_j \right) \E_n \left[ \exp \left( \lambda c_{m-2} X_{m-2} \right) \middle| X^{m-3} \right] \right] \exp \left( \frac{\lambda^2}{8} \left( c_{m-1} ( 1-\alpha_n )^{m-1} \right)^2 \right) \\
& \quad\quad \vdots \\
& \quad  \leq \E_n [ \exp ( \lambda ( c_0 - \alpha_n ) X_0 ) ] \exp \left( \frac{\lambda^2}{8} \sum_{j=1}^{m-1} \left( c_j ( 1-\alpha_n )^j \right)^2 \right) \\
& \quad = \exp \left( \lambda \E_n  \left[ \alpha_n \sum_{i=1}^{m-1} X_i \right] + \frac{\lambda^2}{8} \sum_{j=1}^{m-1} \left( c_j ( 1-\alpha_n )^j \right)^2 \right) ,
\end{align}
where we have used $X_0 = 1$ and \eqref{eqCjObservations} for the final equality. Substituting into \eqref{eqChernoffExpression} yields
\begin{equation}\label{eqChernoffAfterAzuma}
\P_n \left[ \alpha_n \sum_{j=1}^{m-1} X_j  > \frac{\epsilon}{2} \right] \leq \min_{\lambda > 0 } \exp \left( - \lambda \left( \frac{\epsilon}{2} - \E_n \left[ \alpha_n \sum_{i=1}^{m-1} X_i \right] \right) + \frac{\lambda^2}{8} \sum_{j=1}^{m-1} \left( c_j ( 1-\alpha_n )^j \right)^2 \right) .
\end{equation}
It is straightforward to show the global minimizer of \eqref{eqChernoffAfterAzuma} is 
\begin{equation}
\lambda^* = \frac{4 \left( \frac{\epsilon}{2} - \E_n \left[ \alpha_n \sum_{i=1}^{m-1} X_i \right] \right) }{ \sum_{j=1}^{m-1} \left( c_j ( 1-\alpha_n )^j \right)^2 } ,
\end{equation}
which is positive when $\E_n \left[ \alpha_n \sum_{i=1}^{m-1} X_i \right] < \frac{\epsilon}{2}$. Substituting into \eqref{eqChernoffAfterAzuma},
\begin{equation}\label{eqChernoffAfterMin}
\P_n \left[ \alpha_n \sum_{j=1}^{m-1} X_j  > \frac{\epsilon}{2} \right] \leq \exp \left( - \frac{ 2 \left( \frac{\epsilon}{2} - \E_n \left[ \alpha_n \sum_{i=1}^{m-1} X_i \right] \right)^2 }{ \sum_{j=1}^{m-1} \left( c_j ( 1-\alpha_n )^j \right)^2 } \right) .
\end{equation}

We now derive bounds for the denominator and numerator in the exponential in \eqref{eqChernoffAfterMin}. To (coarsely) approximate the denominator, we have
\begin{align}
& c_j < \frac{ \alpha_n  }{ 1 - (1-\alpha_n) \hat{p}_n } < \frac{ \alpha_n }{ 1 - \hat{p}_n }  , \quad \sum_{j=1}^{m-1} \left(  1-\alpha_n  \right)^{2j} < \sum_{j=0}^{\infty} ( 1 - \alpha_n )^j  = \frac{1}{ \alpha_n } \\
& \Rightarrow \sum_{j=1}^{m-1} \left( c_j ( 1-\alpha_n )^j \right)^2 <  \frac{ \alpha_n }{ (1 - \hat{p}_n)^2 }  < \frac{ \alpha_n }{ ( 1 - p - n^{-\gamma} )^2 } , \label{eqDenominatorBound}
\end{align}
where the final inequality holds assuming $\Omega_n$ and $n$ is sufficiently large (so that $p+n^{-\gamma} < 1$). For the numerator, first observe that, when $\Omega_n$ holds, we have
\begin{equation}
\E_n \left[ \alpha_n \sum_{i=1}^{m-1} X_i \right] = \frac{ \alpha_n ( (1-\alpha_n) \hat{p}_n - ( (1-\alpha_n) \hat{p}_n )^{m} ) }{ 1 - (1-\alpha_n) \hat{p}_n } < \frac{ \alpha_n }{ 1 - ( p + n^{-\gamma} ) } ,
\end{equation}
and so $\E_n [ \alpha_n \sum_{i=1}^{m-1} X_i ] < \epsilon/2$ for $n$ sufficiently large (as required), assuming $\alpha_n \rightarrow 0$ as $n \rightarrow \infty$. Therefore, when $\Omega_n$ holds, $\alpha_n \rightarrow 0$, and $n$ is large,
\begin{equation}\label{eqNumeratorBound}
\left( \frac{\epsilon}{2} - \E_n \left[ \alpha_n \sum_{i=1}^{m-1} X_i \right] \right)^2 > \left( \frac{\epsilon}{2} - \frac{\alpha_n}{1 - (p+n^{-\gamma}) } \right)^2 = \frac{\epsilon^2}{4} - \frac{\alpha_n}{1 - (p+n^{-\gamma}) } \left( \epsilon - \frac{\alpha_n}{1 - (p+n^{-\gamma}) } \right) .
\end{equation}

Thus, under these assumptions, \eqref{eqDenominatorBound} and \eqref{eqNumeratorBound} give
\begin{equation}
\frac{ 2 \left( \frac{\epsilon}{2} - \E_n \left[ \alpha_n \sum_{i=1}^{m-1} X_i \right] \right)^2 }{ \sum_{j=1}^{m-1} \left( c_j ( 1-\alpha_n )^j \right)^2 } > \frac{ ( 1 - p - n^{-\gamma} )^2  \epsilon^2 }{ 2 \alpha_n } - 2 ( 1 - p - n^{-\gamma} ) \left( \epsilon - \frac{\alpha_n}{1 - (p+n^{-\gamma}) } \right) .
\end{equation}

To summarize, we have shown that for $n$ sufficiently large, assuming $\alpha_n \rightarrow 0$ and $\Omega_n$ holds,
\begin{align}
\P_n \left[ \alpha_n \sum_{j=1}^{m-1} X_j  > \frac{\epsilon}{2} \right] & \leq \exp \left( - \frac{ ( 1 - p - n^{-\gamma} )^2  \epsilon^2 }{ 2 \alpha_n } \right) \exp \left(2 ( 1 - p - n^{-\gamma} ) \left( \epsilon - \frac{\alpha_n}{1 - (p+n^{-\gamma}) } \right)  \right) \\
& = O \left( \exp \left( - \frac{ ( (1-p) \epsilon )^2 }{ 2 \alpha_n } \right) \right) \label{eqFirstSummandComplete}
\end{align}
where the equality holds because the second exponential term is $O(1)$ for $p \in (0,1)$. Finally, we combine \eqref{eqWBPtailTwoCases}, \eqref{eqWBPtailPtoPn},  \eqref{eqPnBoundTwoTerms}, \eqref{eqSecondSummandComplete}, and \eqref{eqFirstSummandComplete} to obtain
\begin{equation}
\P \left[ \alpha_n \sum_{j=1}^{m} X_j + X_m  > \epsilon \right] = O \left( n^{-\delta} + p^m + e^{ - ((1-p) \epsilon )^2 / ( 2 \alpha_n ) } \right) ,
\end{equation}
which completes the proof.

%% file: simultaneousConstruction.tex
\subsection{Simultaneous construction of graph and tree} \label{secAppSimultaneous}

For the proofs of Lemmas \ref{lemMuCoupling} and \ref{lemCouplingFailure}, we use Algorithm \ref{algSimulConstruction}, which simultaneously constructs a graph and a tree. Algorithm \ref{algSimulConstruction} uses similar notation as Algorithms \ref{algGraphConstruction} and \ref{algTreeConstruction} in Appendix \ref{secCoupling}. However, there are some differences, which we explain before presenting the algorithm.
\begin{itemize}
\item In Algorithm \ref{algGraphConstruction}, we chose $s \sim V_n$ uniformly, which is the standard DCM construction. In Algorithm \ref{algSimulConstruction}, we instead choose $s \sim V_n \setminus K_n$ uniformly. This is because in the statement of Lemma \ref{lemMuCoupling} involves $\mu_s^{(m)}(V_n \setminus K_n)$, conditioned on $U_s = 1$ (i.e.\ $s \in V_n \setminus K_n$); similarly, the statement of Lemma \ref{lemCouplingFailure} involves $\{ \tau_G \leq m \}$, conditioned on $U_s = 1$.
\item Algorithm \ref{algSimulConstruction} uses a function $\Phi : V_n \rightarrow \mathcal{U}$, where $\mathcal{U} = \cup_{j=0}^{\infty} \N^j$ and $\N^0 = \{ \phi \}$ by convention. The function $\Phi$ will be used to map nodes in the graph (which have labels in the set $V_n$, as in Algorithm \ref{algGraphConstruction}) to nodes in the tree (which have labels in the set $\mathcal{U}$, as in Algorithm \ref{algTreeConstruction}).
\item The variable $\tau_S$ in Algorithm \ref{algSimulConstruction} denotes the first iteration at which events that break the coupling occur (analogous to $\tau_G$ in Algorithm \ref{algGraphConstruction}). Once these events occur, the simultaneous construction terminates, and the graph and tree constructions are continued separately using Algorithms \ref{algGraphConstruction} and \ref{algTreeConstruction}, respectively.
\end{itemize}

For illustrative purposes, we include an example of the simultaneous construction in Figure \ref{figSimulExample}. The basic idea is as follows. Whenever a new node is added to the graph, (which occurs when outstub $(v',j)$ is paired with an instub belonging to $v \in V_n$ s.t.\ $g(v) = A$) a new offspring (with the same attributes as $v$) is added to the tree, and a map between the graph node and tree offspring is defined. In particular, Figure \ref{figSimulExample} has the following mapping:
\begin{gather}
\Phi(s) = \phi, \quad \Phi(1) = (1), \quad \Phi(2) = (2), \quad \Phi(3) = (1,2), \\
 \quad \Phi(4) = (1,3),  \quad \Phi(5) = (2,1), \quad \Phi(6) = (2,3), \quad \Phi(7) = (2,4) .
\end{gather}

 If an edge is added between two nodes already in the graph (which occurs when outstub $(v',j)$ is paired with an instub belonging to $v \in V_n$ s.t.\ $g(v) \in \{B,C,D\}$), a new offspring with the same attributes as $v$ is added to the tree. This is illustrated in Figure \ref{figSimulExample} by the following examples:
\begin{itemize}
\item Node 1 in the graph adds an edge to itself; offspring $(1,1)$ in the tree has the attributes of 1
\item Node 1 in the graph adds an edge to 2; offspring $(1,4)$ in the tree has the attributes of 2
\item Node 2 in the graph adds a multi-edge to 5; offspring $(2,2)$ in the tree has the attributes of 5
\end{itemize}
These offspring can be thought of as copies of nodes already in the tree: $(1,1)$, $(1,4)$, and $(2,2)$ are copies of $(1)$, $(2)$, and $(2,1)$, respectively. Furthermore, note that for $\i \in \{ (1,1), (1,4), (2,2) \}$, $\Phi^{-1}(\i) = \emptyset$. In other words, copies of nodes in the tree do not map back to nodes in the graph. This implies that we may have more nodes in the tree than in the graph. For this reason, after pairing all outstubs belonging to all $v' \in A_{m-1}$ (which map to nodes in the tree), we must separately add offspring to nodes $\i \in \hat{A}_{m-1}$ s.t.\ $\Phi^{-1}(\i) = \emptyset$ (which do \textit{not} map to nodes in the tree). This is done in Lines \ref{algSimulReturnToUnassignedBegin}-\ref{algSimulReturnToUnassignedEnd} in Algorithm \ref{algSimulConstruction}.


\begin{figure}
\centering
\begin{subfigure}{.5\textwidth}
  \centering
  \includegraphics[height=1.5in]{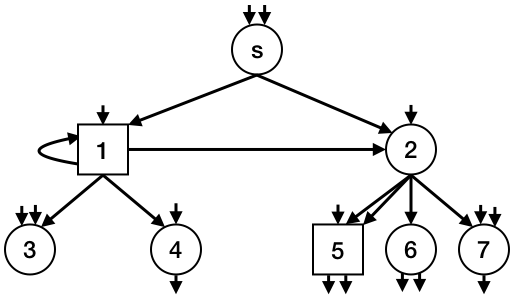}
  \caption{Graph}
  \label{figSimulGraphExample}
\end{subfigure}%
\begin{subfigure}{.5\textwidth}
  \centering
  \includegraphics[height=1.5in]{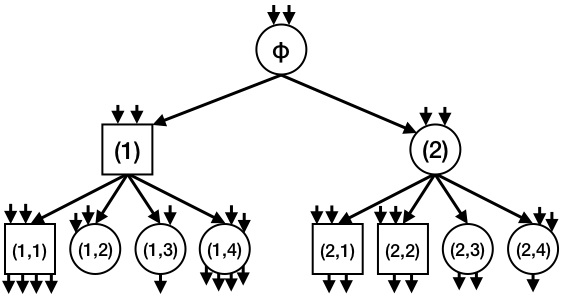}
  \caption{Tree}
  \label{figSimulTreeExample}
\end{subfigure}
\caption{Simultaneous construction example}
\label{figSimulExample}
\end{figure}

\begin{algorithm}[H]
\DontPrintSemicolon
\caption{Simultaneous Construction} \label{algSimulConstruction}

Choose $s$ from $V_n \setminus K_n$ uniformly, set $g(s) = D$, set $A_0 = \{ s \}$

Set $g(e) = 1\ \forall\ e \in S$, set $g(v) = A\ \forall\ v \in V_n \setminus \{ s \}$

Set $(N_{\phi},D_{\phi},U_{\phi}) = (N_s,D_s,U_s)$, set $\hat{A}_0 = \{ \phi \}$

Set $\Phi(s) = \phi$, set $\tau_S = \infty$

\For{$m=1$ \KwTo $\infty$}{

Set $A_m = \hat{A}_m = \emptyset$

\For{$v' \in A_{m-1}$}{

Let $\i = \Phi(v')$

\For{$j = 1$ \KwTo $D_{v'}$}{ 


// find instub for pairing, check if failure has occurred

Uniformly sample instub $e$, denote instub node by $v$

\If{$g(e) = 0$ or $g(e) = 1, g(v') = D, g(v) \in \{C,D\}$}{ \label{algSimulBreak}

Set $\tau_S = m$

Continue constructing graph as in Algorithm \ref{algGraphConstruction}

Continue constructing tree as in Algorithm \ref{algTreeConstruction}

\bf{return} }

// update graph, tree, and map

Pair $(v',j)$ with $e$, set $g(e) = 0$

\lIf{$g(v) = A$}{set $A_m = A_m \cup \{ v \}$, set $\Phi(v) = (\i,j)$} \label{algSimulAddToAm}

Add offspring $( \i ,j)$ to $\i$, set $(N_{( \i ,j)} , D_{( \i ,j)} , U_{( \i ,j)}) = (N_v, D_v, U_v)$, set $\hat{A}_m = \hat{A}_m \cup \{ ( \i ,j) \}$ \label{algSimulAddToHatAm} \label{algSimulOffspringAttr}

// update node label in graph

\lIf{$U_v = 0, g(v) = A$}{set $g(v) = B$}
\lElseIf{$U_v = 1, g(v) = A, g(v') = B$}{set $g(v) = C$}
\lElseIf{$U_v = 1, g(v) = A, g(v') \in \{C,D\}$}{set $g(v) = g(v')$}





}

\lIf{$g(e') = 0\ \forall\ e' \in I_n$}{{\bf{return}}}

}

// generate offspring for tree nodes not mapped to a graph node

\For{$\i \in \hat{A}_{m-1}$ s.t.\ $\Phi^{-1}(\i) = \emptyset$}{ \label{algSimulReturnToUnassignedBegin}

\For{$j = 1$ \KwTo $D_{\i}$}{

Add offspring $( \i ,j)$ to $\i$, sample $(N_{( \i ,j)} , D_{( \i ,j)} , U_{( \i ,j)})$ from $f_n$, set $\hat{A}_m = \hat{A}_m \cup \{ (\i,j) \}$




}

} \label{algSimulReturnToUnassignedEnd}

}

\end{algorithm}

%% file: shortProofs.tex
\section{Other proofs} \label{appOtherProofs}

\subsection{Proof of Claim \ref{clmChoiceOfAlpha}} \label{appProofChoiceOfAlpha}

First, note that $\forall\ n \in \N$, $\forall\ l \in \N$, and $\forall\ l' \leq l$, we have (a.s.)
\begin{equation}\label{eqPiSconcentrateLstep}
\pi_s(V_{n,s}(l)) \geq  \alpha_n \sum_{j=0}^l (1-\alpha_n)^j \left( e_s^{\trans} P^j 1_n \right) = \alpha_n \sum_{j=0}^l (1-\alpha_n)^j = 1 - (1-\alpha_n)^{l+1} \geq 1 - (1-\alpha_n)^{l'} ,
\end{equation}
where the first inequality follows from \eqref{eqPprPowerIter} and by definition of $V_{n,s}(l)$, and the first equality holds since $P^j$ is row stochastic (the remaining steps are simple manipulations). Therefore, when $\alpha_n = \frac{\rho \log(1/\tau) \log \zeta }{ \log n }$, we can define $c = \rho \log \zeta$ and use \eqref{eqPiSconcentrateLstep} to write
\begin{align}
\liminf_{n \rightarrow \infty} \pi_s \left( V_{n,s} \left( \ceil*{ \tfrac{\log(1/\tau)}{\alpha_n} } \right) \right) & \geq  1 - \lim_{n \rightarrow \infty} \left(1 + \frac{\log (1/\tau) }{ \log (n) / c } \right)^{\log (n) / c} = 1 - \tau\ a.s.,
\end{align}
which is the desired bound. If instead $\alpha_n = \alpha$ is a constant, we have more simply
\begin{equation}
\liminf_{n \rightarrow \infty} \pi_s \left( V_{n,s} \left( \ceil*{ \tfrac{ \log(1/\tau) }{ \log(1/(1-\alpha)) } } \right) \right) \geq 1 - (1-\alpha)^{ \frac{ \log(1/\tau) }{ \log(1/(1-\alpha)) } }  = 1 -  (1-\alpha)^{ \log_{(1-\alpha)} (\tau) } = 1 - \tau\ a.s. ,
\end{equation}
which is the other desired bound. Next, to bound the expected size of $V_{n,s}(l)$, we use the analysis of Appendix \ref{appCouplingFailureE0Fm}. First, for $l \in \N$, the argument preceding \eqref{eqNumberOfNodesInTree} in Appendix \ref{appCouplingFailureE0Fm} implies
\begin{equation}
| V_{n,s}(l) | \leq \sum_{j=0}^l \hat{Z}_j ,
\end{equation}
where $\hat{Z}_j$ is defined in \eqref{eqZhatLdefn}. Furthermore, by  \eqref{eqNumberOfNodesInTreeExp} in Appendix \ref{appCouplingFailureE0Fm}, we have for $j \in \N$,
\begin{equation}
\E \left[ \hat{Z}_{j} \middle| \Omega_n \right] = O \left( \zeta^{j-1}  \right) ,
\end{equation}
while $\hat{Z}_0 = 1$ by definition. Combining the previous two equations gives for $l \in \N$,
\begin{equation}
\E \left[ V_{n,s}(l) | \Omega_n \right] = O \left( 1 + \sum_{j=0}^{l-1} \zeta^{j} \right) = O \left( \zeta^l \right) .
\end{equation}
Therefore, when $\alpha_n = \frac{\rho \log(1/\tau) \log \zeta }{ \log n }$, we have
\begin{equation}
\E \left[ \left| V_{n,s} \left( \ceil*{ \tfrac{\log(1/\tau)}{\alpha_n} } \right) \right| \middle| \Omega_n \right] = O \left( \zeta^{ \log(1/\tau) / \alpha_n} \right) = O \left( \zeta^{ \log_{\zeta} ( n^{1/\rho} ) } \right) = O \left( n^{1/\rho} \right) .
\end{equation}
Similarly, if $\alpha_n = \alpha$ is a constant,
\begin{equation}
\E \left[ \left| V_{n,s} \left( \ceil*{ \tfrac{ \log(1/\tau) }{ \log(1/(1-\alpha)) } } \right) \right| \middle| \Omega_n \right] = O \left( \zeta^{ \log(1/\tau) / \log(1/(1-\alpha)) } \right) = O(1) .
\end{equation}

\subsection{Proof of Theorem \ref{thmSublinear}} \label{appProofOfThmSublinear}

First, we observe
\begin{align}
\E [ \Delta_{\psi_{n,\kappa}}(\epsilon) ] & = \E \left[ | K_n | + | \{ v \in V_n \setminus K_n : B_v(\epsilon) \textrm{ holds} \} | \right]  \\
& = n^{\kappa} + \sum_{v \in V_n} \E \left[ 1 ( B_v(\epsilon) , U_v = 1 ) \right] \\
& = n^{\kappa} +  n \E \left[ 1 ( B_s(\epsilon), U_s = 1 ) \right] \\
& \leq  n^{\kappa} + n \P [ B_s(\epsilon) | U_s = 1 ] \\
& = n^{\kappa} + n O \left( n^{-c(\epsilon) } \right)  \\
& = O \left( n^{ \max \{ \kappa, 1-c(\epsilon) \} } \right)  ,
\end{align}
where the successive steps hold by definition of $\Delta_{\psi_{n,\kappa}}(\epsilon)$, by definition of $\psi_{n,\kappa}$, since $1(B_v(\epsilon) , U_v = 1 )$ are identically distributed before the degree sequence is realized, since $\P [ U_v = 1 ] \leq 1$, and by Lemma \ref{lemMainTailBound}, respectively. Hence, by Markov's inequality,
\begin{equation}
\P \Big[ \Delta_{\psi_{n,\beta}} ( \epsilon ) \geq C n^{\bar{c}} \Big] \leq \frac{ \E [ \Delta_{\psi_{n,\kappa}}(\epsilon) ] }{ C n^{\bar{c}} }  = O \left( n^{ \max \{ \kappa, 1-c(\epsilon) \} - \bar{c} } \right) .
\end{equation}

%% file: experimentDetails.tex
\section{Experimental details} \label{appExpDetails}

\subsection{Dataset details} \label{appDatasets}

The following table shows details of the datasets used for experiments in Section \ref{secAlgorithms}. All datasets are available from the Stanford Network Analysis Platform \cite{snapnets}. The $\alpha_n$ values shown are used for all experiments conducted on the corresponding graph. We note that, while these are smaller than $\alpha_n$ values typically used, they are the same order of magnitude ($\alpha_n = 0.15$ is a common choice in the literature).

\begin{center}
\begin{tabular}{|l|l|c|c|c|}
\hline 
\textbf{Dataset} & \textbf{Description} & $n$ & $|E_n|$ & $\alpha_n = 1/  \log n$ \\ 
\hline 
soc-LiveJournal1 & Blogging social network & 4847571 & 68993773 & 0.065 \\ 
\hline 
soc-pokec & Slovakian social network & 1632803 & 30622564 & 0.070 \\ 
\hline 
web-Google & Partial web graph from Google & 875713 & 5105039 & 0.073 \\ 
\hline 
web-BerkStan & berkley.edu, stanford.edu web graph & 685230 & 7600595 & 0.074 \\ 
\hline 
web-Stanford & stanford.edu web graph & 281903 & 2312497 & 0.080 \\ 
\hline 
\end{tabular} 
\end{center}

\subsection{Scheme to bound estimation error} \label{appEstErrBound}

To bound $\| \pi_v - ( \alpha_n e_v^{\trans}  + \sum_{k \in K_n} \beta_v(k) \pi_k ) \|_1$, where $\beta_v(k)$ are defined in \eqref{eqChoiceOfBetaV}, we employ a power iteration scheme: we initialize $x_v^{(0)} = e_v^{\trans}$, and given $x_v^{(i-1)}$ for $i \geq 1$, we set
\begin{equation}
x_v^{(i)} = \alpha_n e_v^{\trans} + (1-\alpha_n) x_v^{(i-1)} \tilde{P} ,
\end{equation}
where $\tilde{P}$ is defined in \eqref{eqModifiedMCtransition}. We claim
\begin{equation}\label{eqNaiveIter}
x_v^{(i)} = \alpha_n \mu_v^{(i-1)} + (1-\alpha_n)^i e_v^{\trans} \tilde{P}^i ,
\end{equation}
where $\mu_v^{(i-1)} = e_v^{\trans} \sum_{j=0}^{i-1} (1-\alpha_n)^j \tilde{P}^j$ (as in Appendix \ref{secMainProof}-\ref{appMainProofs}). \eqref{eqNaiveIter} is easily proven inductively: the base of induction holds by definition; assuming true for $i-1$, we have
\begin{align}
x_v^{(i)} & = \alpha_n e_v^{\trans} + (1-\alpha_n) \left( \alpha_n \mu_v^{(i-2)} + (1-\alpha_n)^{i-1} e_v^{\trans} \tilde{P}^{i-1} \right) \tilde{P} \\
& = \alpha_n e_v^{\trans} + \alpha_n \sum_{j=1}^{i-1} (1-\alpha_n)^j \tilde{P}^j + ( 1-\alpha_n )^i e_v^{\trans} \tilde{P}^i = \alpha_n \mu_v^{(i-1)} + ( 1-\alpha_n )^i e_v^{\trans} \tilde{P}^i
\end{align}
as claimed. Now by Lemma \ref{lemL1bound} in Appendix \ref{secMainProof}, for any $i \in \N$ we obtain the following bound:
\begin{align}
\left\| \pi_v - \left( \alpha_n e_v^{\trans} + \sum_{k \in K_n} \beta_v(k) \pi_k \right) \right\|_1 & \leq \alpha_n  \mu_v^{(i-1)}(V_n \setminus K_n)  +  (1-\alpha_n)^i e_v^{\trans} \tilde{P}^i e_{V_n \setminus K_n} -\alpha_n  \label{eqNaiveIterErrorBoundpre} \\
& = x_v^{(i)} ( V_n \setminus K_n ) - \alpha_n . \label{eqNaiveIterErrorBound} 
\end{align}

From this bound, we can prove two other claims from Section \ref{secAlgorithms}. First, we note
\begin{align}
x_v^{(i)} ( V_n \setminus K_n ) & = \alpha_n e_v^{\trans} \sum_{j=0}^{i-1} (1-\alpha_n)^j \tilde{P}^j + (1-\alpha_n)^i e_v^{\trans} \tilde{P}^i e_{V_n \setminus K_n}  \leq \alpha_n \sum_{j=0}^{i-1} (1-\alpha_n)^j + (1-\alpha_n)^i = 1
\end{align}
where the inequality follows since $\tilde{P}$ is nonnegative with row sums bounded by 1. Hence, from \eqref{eqNaiveIterErrorBound}, the estimation error is bounded by $(1-\alpha_n)$ (as claimed in Section \ref{secAlgorithms}). Next, suppose $v \in V_{n,0}$, with $V_{n,0}$ given by \eqref{eqVn0defn}. Then $e_v^{\trans} \tilde{P}^j e_{V_n \setminus K_n} = 0$ by definition, so $x_v^{(i)}(V_n \setminus K_n) = \alpha_n$, and the estimation error is zero (as claimed in Section \ref{secAlgorithms}).

We can also bound the gap in the inequality \eqref{eqNaiveIterErrorBoundpre}: use \eqref{eqL1BoundNoM} in Appendix \ref{secProofLemL1bound} and \eqref{eqNaiveIterErrorBoundpre} to write
\begin{align}
& \left\| \pi_v - \left( \alpha_n e_v^{\trans} + \sum_{k \in K_n} \beta_v(k) \pi_k \right) \right\|_1 - \left( x_v^{(i)} ( V_n \setminus K_n ) - \alpha_n \right) \\
& \quad = \left( \alpha_n e_v^{\trans} \sum_{j=0}^{\infty} (1-\alpha_n)^i \tilde{P}^i e_{V_n \setminus K_n} - \alpha_n \right) - \left( x_v^{(i)} ( V_n \setminus K_n ) - \alpha_n \right) \\
& \quad = \alpha_n \mu_v^{(i-1)} ( V_n \setminus K_n ) + \alpha_n e_v^{\trans} \sum_{j=i}^{\infty} (1-\alpha_n)^i \tilde{P}^i e_{V_n \setminus K_n} - x_v^{(i)} ( V_n \setminus K_n )  \\
& \quad = \alpha_n e_v^{\trans} \sum_{j=i}^{\infty} (1-\alpha_n)^i \tilde{P}^i e_{V_n \setminus K_n} -  (1-\alpha_n)^i e_v^{\trans} \tilde{P}^i e_{V_n \setminus K_n} \geq - (1-\alpha_n)^i
\end{align}
where the inequality holds by dropping a nonnegative term and since $e_v^{\trans} \tilde{P}^i e_{V_n \setminus K_n} \leq 1$. Hence, if we let $i^* \geq \log_{ ( 1-\alpha_n) } ( tol ) $ for some desired tolerance $tol$, the bound $x_v^{(i^*)}(V_n \setminus K_n) - \alpha_n$ is tight within additive error $tol$. (For all experiments, we set $tol = 0.05$.)

To bound average error across $V_n \setminus K_n$, we instead use the iteration
\begin{equation}
x_{ V_n \setminus K_n }^{(0)} = \frac{e_{V_n \setminus K_n}^{\trans}}{ | V_n \setminus K_n |} , \quad x_{ V_n \setminus K_n }^{(i)} = \alpha_n \frac{e_{V_n \setminus K_n}^{\trans} }{ | V_n \setminus K_n |} + ( 1 - \alpha_n ) x_{ V_n \setminus K_n }^{(i-1)} \tilde{P} .
\end{equation}
Note $x_{ V_n \setminus K_n }^{(i)} = \frac{1}{ | V_n \setminus K_n | } \sum_{v \in V_n \setminus K_n} x_v^{(i)}$ when $i = 0$ by definition; assuming true for general $i-1$,
\begin{align}
x_{ V_n \setminus K_n }^{(i)} & = \alpha_n \frac{e_{V_n \setminus K_n}^{\trans} }{ | V_n \setminus K_n |} + ( 1 - \alpha_n ) \left( \frac{1}{ | V_n \setminus K_n | } \sum_{v \in V_n \setminus K_n} x_v^{(i-1)} \right) \tilde{P} \\
& = \frac{1}{|V_n \setminus K_n|} \sum_{v \in V_n \setminus K_n} \left( \alpha_n e_v^{\trans} + (1-\alpha_n) x_v^{(i-1)} \tilde{P} \right) = \frac{1}{|V_n \setminus K_n|} \sum_{v \in V_n \setminus K_n} x_v^{(i)} ,
\end{align}
i.e.\ $x_{ V_n \setminus K_n }^{(i)} = \frac{1}{ | V_n \setminus K_n | } \sum_{v \in V_n \setminus K_n} x_v^{(i)}\ \forall\ i \in \N$. It follows from above that
\begin{equation}
\frac{1}{|V_n \setminus K_n|} \sum_{v \in V_n \setminus K_n} \left\| \pi_v - \left( \alpha_n e_v^{\trans} + \sum_{k \in K_n} \beta_v(k) \pi_k \right) \right\|_1 \leq x_{V_n \setminus K_n}^{(i)} ( V_n \setminus K_n ) - \alpha_n ,
\end{equation}
which is the average error bound we compute for Figure \ref{figAvgErrPlots}. The argument above also implies this bound is tight within $tol$ when $i^* \geq \log_{ ( 1-\alpha_n) } ( tol )$.

\subsection{Details on Figure \ref{figAllErrPlots} experiment} \label{appAllErrorExp}

In addition to the histograms of $l_1$ error shown in Figure \ref{figAllErrPlotsHist}, we include a more detailed set of plots for the same experiment. Specifically, we estimate the error $ | \alpha_n e_v^{\trans}(w) + \sum_{k \in K_n} \beta_v(k) \pi_k(w) |$ as $x_v^{(i)}(w) - \alpha_n 1 ( w = v )$ (where $x_v^{(i)}$ is defined in Appendix \ref{appEstErrBound}), for each $w \in V_n \setminus K_n$, and for each $v$ in a subset of $V_n \setminus K_n$ of size $\approx 10^4$. (These $v$ were chosen uniformly from nodes with average error $\in (0.08,0.25)$, which corresponds to the regime of  linear decay in Figure \ref{figAllErrPlots}.) We also estimate the relative error, i.e.\ the ratio of this absolute error to an estimate of $\pi_v(w)$, for the same set of $(v,w)$. The estimate of $\pi_v$ is computed using the same power iteration scheme in Appendix \ref{appEstErrBound}, but replacing $\tilde{P}$ with $P$. Note this gives a lower bound on the true value of $\pi_v(w)$, thereby upper bounding relative error. Unfortunately, we cannot compute this relative error estimate when the estimate of $\pi_v(w)$ is zero; this occurred for only 10\% of $(v,w)$ pairs considered. Finally, for both absolute and relative error, we compute the number of error values lying in log-spaced bins and divide these values by $n$ to estimate the frequency of each error value. (We add values lying beyond the first and last bin edges to the first and last bins, respectively.)

Results are shown for the soc-Pokec dataset in Figure \ref{figPokecDetails}. (We note the spikes at left occur due to values lying beyond the first bin edge.) As an illustration for absolute error, the frequency of values above $10^{-3}$ was $\approx 10^{-5}$, i.e.\ the vast majority of nodes had estimated absolute error below $10^{-3}$. To illustrate the relative error, the frequency of values above $0.2$ was $\approx 0.09$, i.e.\ over 90\% of nodes had estimated relative error below $0.2$. The results for web-Google are shown in Figure \ref{figGoogleDetails}. For absolute error, the frequency above $10^{-3}$ was again $\approx 10^{-5}$; for relative error, again over 90 \% of nodes had error below $0.2$.

\begin{figure}[t]
\centering
\includegraphics[height=2.5in]{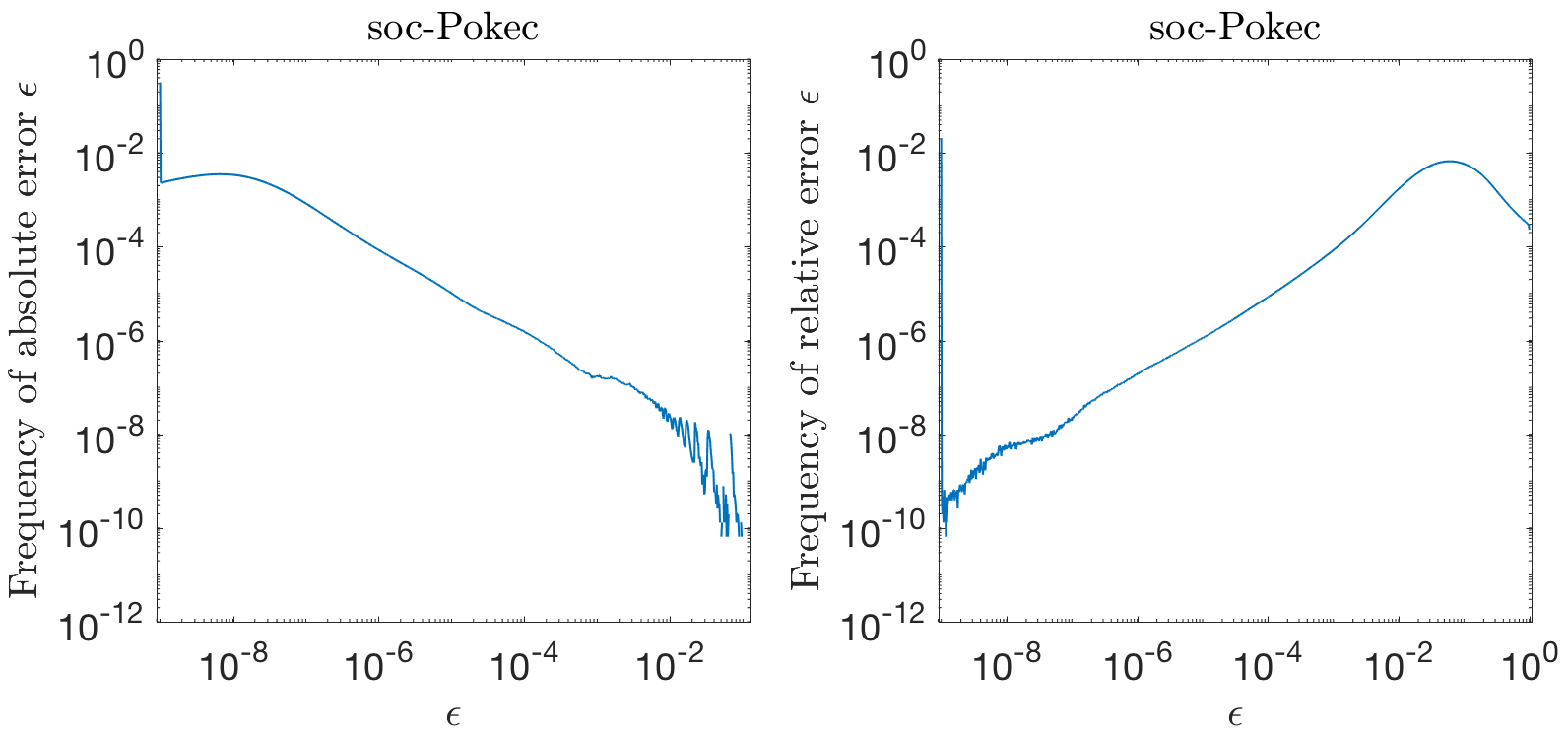}
\caption{Detailed error analysis for soc-Pokec social network.} \label{figPokecDetails}
\end{figure}

\begin{figure}[t]
\centering
\includegraphics[height=2.5in]{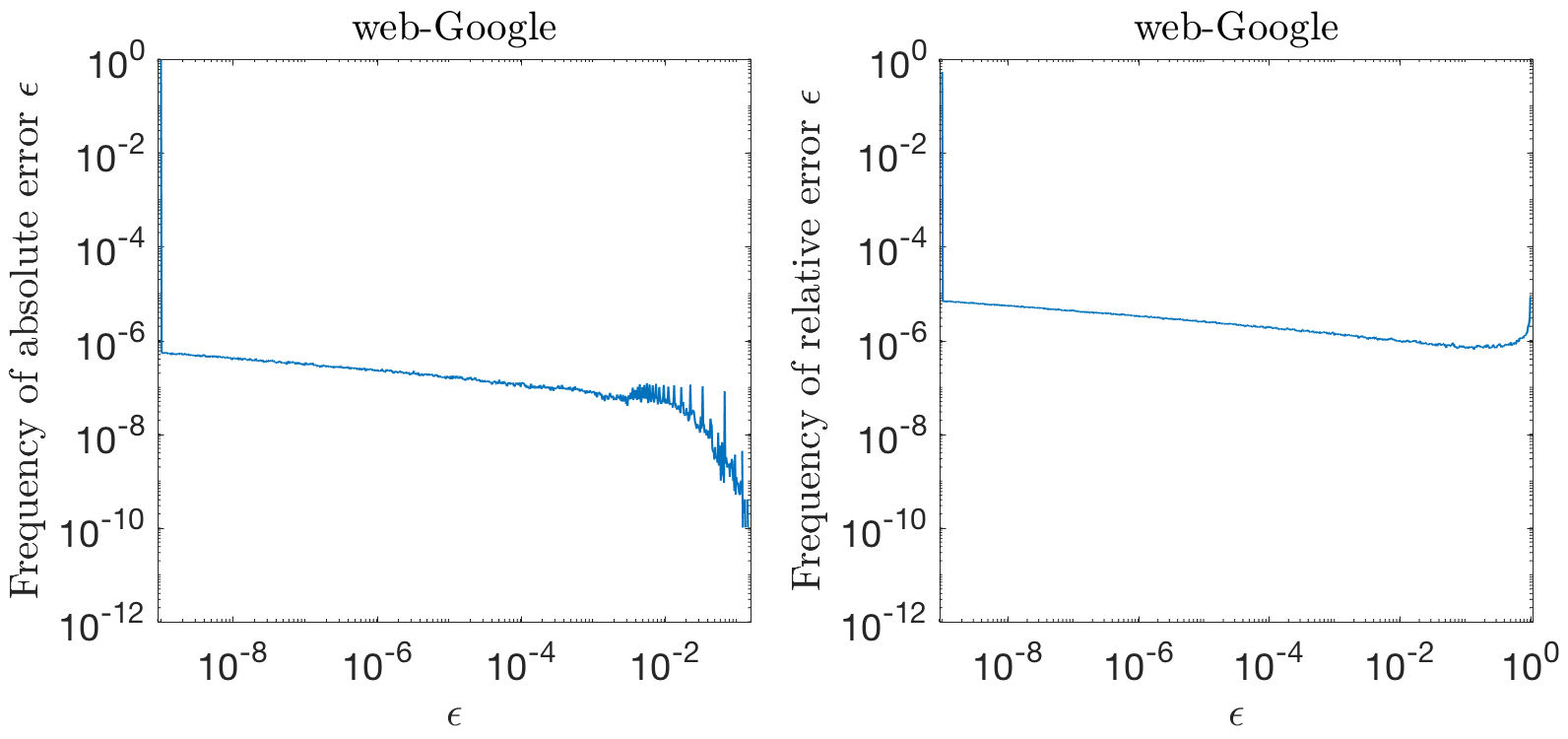}
\caption{Detailed error analysis for web-Google web graph.} \label{figGoogleDetails}
\end{figure}